\documentclass[12pt]{article}
\usepackage{amsfonts,amssymb,amsmath}	%
\usepackage{amsthm}	%
\usepackage[left=1in,right=1in,top=1in,bottom=1in,letterpaper]{geometry}	%
\usepackage{setspace}	%
\usepackage{graphicx}	%
\usepackage[hidelinks]{hyperref}    %
\usepackage{cleveref,natbib}	%
\usepackage[toc,page]{appendix}	%

\usepackage{mathpazo}

\usepackage[pagewise,mathlines]{lineno}

\newtheorem{lemma}{Lemma}
\newtheorem{proposition}{Proposition}
\newtheorem{theorem}{Theorem}

\theoremstyle{definition}

\newtheorem{assumption}{Assumption}
\newtheorem{example}{Example}

\crefname{definition}{definition}{definitions}
\Crefname{definition}{Definition}{Definitions}
\crefname{corollary}{corollary}{corollaries}
\Crefname{corollary}{Corollary}{Corollaries}
\crefname{assumption}{assumption}{assumptions}
\Crefname{assumption}{Assumption}{Assumptions}
\crefname{example}{example}{examples}
\Crefname{example}{Example}{Examples}
\crefname{remark}{remark}{remarks}
\Crefname{remark}{Remark}{Remarks}
\crefname{appendix}{appendix}{appendices}
\Crefname{appendix}{Appendix}{Appendices}

\newcommand{\pp}{\mathbb{P}}
\newcommand{\E}{\mathbb{E}}
\newcommand{\var}{\text{Var}}

\DeclareMathOperator*{\argmax}{argmax}
\DeclareMathOperator*{\argmin}{argmin}

\newcommand{\veps}{\varepsilon}

\newcommand{\mf}{\mathbf}
\newcommand{\mb}{\mathbb}
\newcommand{\mc}{\mathcal}

\newcommand{\st}{\qquad\text{subject to}\qquad}
\newcommand{\qand}{\quad\text{and}\quad}

\newcommand{\ra}{\quad\Rightarrow\quad}

\newcommand{\bsni}{\bigskip\noindent}

\title{Identification and estimation of \\ dynamic random coefficient models}
\author{Wooyong Lee\thanks{Economics Discipline Group, UTS Business School, University of Technology Sydney. Address: 14-28 Ultimo Road, Ultimo, NSW 2007, Australia. Tel: (02) 9514 3074. Email: \href{mailto:wooyong.lee.econ@gmail.com}{wooyong.lee.econ@gmail.com}}}
\date{April 24, 2026}
\begin{document}

\onehalfspacing
\maketitle

\begin{abstract}

\normalsize

I study linear panel data models with predetermined regressors (such as lagged dependent variables) where coefficients are individual-specific, allowing for heterogeneity in the effects of the regressors on the dependent variable. I show that the model is not point-identified in a short panel context but rather partially identified, and I characterize the identified sets for the mean, variance, and CDF of the coefficient distribution. This characterization is general, accommodating discrete, continuous, and unbounded data, and it leads to computationally tractable estimation and inference procedures. I apply the method to study lifecycle earnings dynamics among U.S. households using the Panel Study of Income Dynamics (PSID) dataset. The results suggest the presence of unobserved heterogeneity in earnings persistence, implying that households face varying levels of earnings risk which, in turn, contribute to heterogeneity in their consumption and savings behaviors.

\bsni

\end{abstract}

\newpage
\section{Introduction}

Linear panel data models with predetermined regressors (e.g., lagged dependent variables) are widely used in empirical research \citep{arellano1991,blundell1998}. Many of these models incorporate fixed effects, which are individual-specific intercepts that account for unobserved heterogeneity in the levels of the dependent variable. Fixed effects provide a flexible means of controlling for such heterogeneity, facilitating empirical research such as evaluation of a public policy. Fixed effects models are well understood in the context of short panel data (i.e., panel data with a small number of periods).

In addition to heterogeneity in the levels of dependent variables, there is ample evidence that individuals exhibit unobserved heterogeneity in the effects of regressors on dependent variables. For example, firms have varying degrees of labor efficiency in production; individuals experience different returns on education; and households differ in the persistence of earnings with respect to their past earnings. Such heterogeneous effects are crucial mechanisms for generating heterogeneous responses to exogenous shocks and policies, such as employment subsidies, tuition assistance, and income tax reforms. Moreover, these heterogeneous effects play a first-order role in determining outcomes in various economic models. For instance, heterogeneity in earnings persistence drives differences in the earnings risk faced by households, which, in turn, influences their heterogeneous motives for precautionary savings within lifecycle consumption models.

This paper examines a linear panel data model with predetermined regressors that permits unobserved heterogeneity in both the effects of regressors and the levels (i.e., a dynamic random coefficient model) in a short panel context. Consider a stylized example:
\begin{displaymath}
    Y_{it} = \beta_{i0} + \beta_{i1}Y_{i, t-1} + \veps_{it},
\end{displaymath}
where all variables are scalars and $\veps_{it}$ is uncorrelated with the current history of $Y_{it}$ (up to time $t-1$) but may be correlated with its future values. In this model, both the coefficient $\beta_{i1}$ and the intercept $\beta_{i0}$ are individual-specific, capturing heterogeneity in the effects of regressors and the levels. Moreover, the inclusion of the lagged dependent variable $Y_{i,t-1}$ as a regressor makes it a dynamic model.

Analysis of this model is challenging in short panels, as it is impossible to learn about individual values of the $\beta_i$ parameters with a small number of periods. An influential study by \citet{chamberlain1993,chamberlain2022} showed that the mean of the $\beta_i$s in dynamic random coefficient models is not point-identified, implying that it cannot be consistently estimated. Since this negative result in the 1990s, progress in the literature has been limited. \citet{arellano2012} showed that, for binary regressors, the mean of the $\beta_i$s for certain subpopulations is identifiable and thus consistently estimable, but they did not establish a general identification result applicable to non-binary regressors. Most research on random coefficient models in short panels has focused on non-dynamic contexts \citep{chamberlain1992,wooldridge2005,arellano2012,graham2012}, which exclude important dynamic mechanisms, such as the feedback from the current dependent variable to future regressors. For instance, a firm's labor purchase decision in the following period may depend on its current output, as the firm might learn about its own labor efficiency from that output. Moreover, understanding these dynamic mechanisms is of independent interest. For example, a household's earnings persistence with respect to its past earnings is an important parameter, since high persistence increases the duration of earnings shocks, diminishing the household's ability to smooth consumption and, ultimately, impacting welfare.

This paper is, to the best of my knowledge, the first to present a general identification result for dynamic random coefficient models in a short panel context. Identification results are presented for various features of the coefficients, including their mean, variance, and cumulative distribution function (CDF). In addition, this paper proposes a computationally feasible estimation and inference procedure for these features. The procedure is then applied to investigate unobserved heterogeneity in lifecycle earnings dynamics among U.S. households using the Panel Study of Income Dynamics (PSID) dataset. These are presented in three steps.

First, I show that dynamic random coefficient models are partially identified, and I characterize finite lower and upper bounds for a class of parameters including the mean, variance, and CDF of the coefficient distribution. While these characterizations yield bounds that are not necessarily sharp, they are sufficiently general to accommodate discrete, continuous, or unbounded data. Moreover, for the mean of the coefficient distribution, the characterization yields a simple closed-form expression for its bounds, %
which clearly demonstrates that the bounds remain finite even when the data are unbounded, provided that certain moments of the data are finite. These results are obtained by recasting the identification problem as a linear programming problem \citep{honore2006,honore2006b,mogstad2018,torgovitsky2019}, which becomes infinite-dimensional when the data or the coefficients are continuous. I then employ the dual representation of infinite-dimensional linear programming \citep{galichon2009,schennach2014} to derive the bounds for the parameters of interest.

Second, I propose computationally efficient estimation and inference procedures for the bounds. For the mean of the coefficient distribution, the closed-form expressions for its lower and upper bounds yield a simple and easy-to-implement estimation and inference procedure. In particular, I adopt the approach of \citet{stoye2020} and develop a simple procedure for constructing confidence intervals that are not only valid but also robust to overidentification and model misspecification. For other features of the coefficient distribution, such as the variance and the CDF, I use the approach of \citet{andrews2017}, which performs inference on a continuum of moment inequalities and includes countably many moment inequalities as a special case. Although this procedure is computationally more demanding than that for the mean parameters, it remains computationally feasible for inference on various features of the coefficients.

Third, I estimate a reduced-form lifecycle model of earnings dynamics. Lifecycle earnings processes are key inputs in various economic models, including those of lifecycle consumption dynamics \citep{hall1982,blundell2008,blundell2016,arellano2017}. Specifying an earnings process that captures features of real data is important for calibrating and drawing conclusions from these models. I investigate unobserved heterogeneity in the earnings of U.S. households using the Panel Study of Income Dynamics (PSID) dataset. \citet{guvenen2007,guvenen2009} pointed out that, when allowing for unobserved heterogeneity in the time trend of earnings (known as a heterogeneous income profile, HIP), the estimated persistence of the income process is significantly below 1, with the latter being the estimate from the model that assumes no heterogeneity in the time trend (known as a restricted income profile, RIP). I extend this analysis by estimating a more general model that also permits unobserved heterogeneity in earnings persistence itself. 
I find that both the HIP and RIP specifications yield similar estimates of the average earnings persistence, with values significantly below 1. This suggests that misspecifying HIP as RIP (or vice versa) may not lead to serious model misspecification when earnings persistence is allowed to vary across households. Moreover, I find evidence of unobserved heterogeneity in earnings persistence itself, implying that households face different levels of earnings risk, which in turn contributes to heterogeneity in their consumption and savings behavior.

The identification results in this paper can be extended to other structural models to accommodate heterogeneous effects. For example, these results can be applied to models with individual-specific coefficients and intercepts in probit and logit regressions\footnote{In related studies, \citet{bonhomme2023,bonhomme2025} analyzed general nonlinear panel data models with sequentially exogenous regressors and individual-specific intercepts.}. They can also be extended to vector-valued regressions, such as panel data vector autoregressive (VAR) models and systems of panel data regressions.

The remainder of this paper is structured as follows. \Cref{sec.model} introduces the dynamic random coefficient model. \Cref{sec.mean,sec.general} present the identification results for the model, focusing on the mean in \Cref{sec.mean} and on more general features in \Cref{sec.general}. \Cref{sec.estimation} discusses the estimation and inference procedures, and \Cref{sec.application} applies these methods to lifecycle earnings dynamics. \Cref{sec.conclusion} concludes the paper. All proofs are provided in Online Appendix \ref{sec.proof}.

\section{Model and motivating examples}

\label{sec.model}

The dynamic random coefficient model is specified as follows:
\begin{displaymath}
    Y_{it} = Z_{it}'\gamma_i + X_{it}'\beta_{i} + \veps_{it}, \qquad t=1, \ldots, T,
\end{displaymath}
where $i$ is an index of individuals, $T$ is the length of panel data, $(Y_{it}, Z_{it}, X_{it}) \in \mb{R} \times \mb{R}^q \times \mb{R}^p$ are observed real vectors at time $t \in \{1, \ldots, T\}$, and $\veps_{it} \in \mb{R}$ is the idiosyncratic error at time $t$. I assume
\begin{displaymath}
    \E(\veps_{it}|\gamma_i, \beta_i, Z_{i1}, \ldots, Z_{iT}, X_{i1}, \ldots, X_{it}) = 0,
\end{displaymath}
which states that $\veps_{it}$ is mean-independent of the full history of $\{Z_{it}\}$ (i.e., strict exogeneity) and of the current history of $\{X_{it}\}$ (i.e., sequential exogeneity). The inclusion of a sequentially exogenous regressor $\{X_{it}\}$ makes it a dynamic model. For example, the lagged dependent variable $Y_{i,t-1}$ can be included in $X_{it}$.

Let $R_{it} = (Z_{it}', X_{it}')'$ be the vector of regressors at time $t$, and let $B_i = (\gamma_i', \beta_i')'$ be the vector of random coefficients. In addition, let $Y_i = (Y_{i1}, \ldots, Y_{iT})$ be the full history of $\{Y_{it}\}$, and let $Y_i^t = (Y_{i1}, \ldots, Y_{it})$ be the history of $\{Y_{it}\}$ up to time $t$. Define $X_i$, $X_i^t$, $Z_i$, $Z_i^t$ similarly. With these definitions, I concisely write the model as:
\begin{equation}
    Y_{it} = R_{it}'B_i + \veps_{it}, \quad t=1, \ldots, T,
    \label{eq.crc}
\end{equation}
and
\begin{equation}
    \E(\veps_{it}|B_i, Z_i, X_i^t) = 0.
    \label{eq.meanindep}
\end{equation}

The model is studied in a short panel context, which corresponds to the asymptotics that the number of individuals $N \rightarrow \infty$ while the number of time periods $T$ remains fixed. The random coefficients $\gamma_i$ and $\beta_i$ are unobserved random variables that follow nonparametric distributions, and they may be arbitrarily correlated with each other as well as with $(Z_i, X_{i1})$. This is how the random coefficient model extends a fixed effects model.

I summarize the variables of the model as two random vectors: the observable data $W_i = (Y_i', Z_i', X_i')' \in \mc{W}$ and the unobservable random coefficients $B_i \in \mc{B}$. Note that $(W_i, B_i)$ also summarizes $\veps_{it}$ by the relationship $\veps_{it} = Y_{it} - R_{it}'B_i$.

Given this model, I consider a parameter $\theta$ of the form
\begin{displaymath}
    \theta = \E(m(Y_i, Z_i, X_i, \gamma_i, \beta_i)) = \E(m(W_i, B_i))
\end{displaymath}
for some known function $m$. I present identification results for a generic function $m$, but I focus on the case in which $m$ is either a polynomial or an indicator function of $B_i$, which allows for computationally feasible estimation and inference. This choice of $m$ includes many important parameters of interest. For example, $\theta$ can be an element of the mean of the random coefficients $\E(B_i)$ or an element of the second moments $\E(B_iB_i')$. It can also represent the error variance $\E(\veps_{it}^2)$ because $\veps_{it}^2 = (Y_{it} - R_{it}'B_i)^2$ is a quadratic polynomial in $B_i$. Another example is the CDF of $B_i$ evaluated at $b$, in which case one sets $m = \mf{1}(B_i \leq b)$ so that $\theta = \E(\mf{1}(B_i \leq b)) = \pp(B_i \leq b)$.

\begin{example}[Household earnings]

    \label{sec.ex.income}

    One of the simplest examples of (\ref{eq.crc}) is the AR(1) model with heterogeneous coefficients:
    \begin{equation}
        Y_{it} = \gamma_i + \beta_i Y_{i,t-1} + \veps_{it},
        \label{eq.income}
    \end{equation}
    where all variables are scalars. This is a special case of (\ref{eq.crc}), with $Z_{it} = 1$ and $X_{it} = Y_{i,t-1}$.

    The AR(1) model is a popular choice for empirical specification of the lifecycle earnings process, with $Y_{it}$ representing log-earnings, an important input in models of consumption and savings behavior\footnote{In the literature, it is standard to add a transitory shock to (\ref{eq.income}).}. The earnings persistence parameter, $\beta_i$, governs the earnings risk experienced by households, which is a fundamental motive for precautionary savings. Specifying an earnings process that highlights features of real data is important for drawing conclusions from models of consumption and savings behavior. In the literature, the earnings process is often modeled as an AR(1) process with homogeneous coefficients \citep{lillard1979,blundell2013,gu2017}, or as a unit root process, i.e., an AR(1) model with $\gamma_i=0$ and $\beta_i=1$ \citep{hall1982,meghir2004,kaplan2014}.

\end{example}

\begin{example}[Household consumption behavior]

    \label{sec.ex.consumption}

    Consider a model of lifecycle consumption behavior:
    \begin{equation}
        C_{it} = \gamma_{i0} + \gamma_{i1} Y_{it} + \beta_{i} A_{it} + \nu_{it},
        \label{eq.consumption}
    \end{equation}
    where all variables are scalars. In this equation, $C_{it}$ is non-durable consumption, $Y_{it}$ is earnings, and $A_{it}$ is asset holdings at time $t$, all measured in logs. In this specification, $Y_{it}$ can be regarded as strictly exogenous, implying that future earnings are unaffected by the current consumption choice. In contrast, $A_{it}$ must be taken as sequentially exogenous, as assets and consumption are interrelated through the intertemporal budget constraint.

    The model in (\ref{eq.consumption}) can be viewed as an approximation of the consumption rule derived from a structural model \citep{blundell2016}. One parameter of interest is $\gamma_{i1}$, the elasticity of consumption with respect to earnings. This elasticity measures a household's ability to smooth consumption in response to exogenous changes in earnings, such as earnings shocks, thereby mitigating adverse impacts on household welfare. %
    Another parameter of interest is $\beta_i$, the elasticity of consumption with respect to asset holdings, which measures the household's capacity to smooth consumption in response to exogenous asset changes. Note that the model in (\ref{eq.consumption}) remains agnostic about the evolution of assets over time, i.e., it allows for a nonparametric evolution of the asset process.

\end{example}

\begin{example}[Production function]

    \label{sec.ex.production}

    An influential paper by \citet{olley1996} considered the estimation of the production function for firms operating with Cobb-Douglas technology. In their work, the following model was analyzed:
    \begin{displaymath}
        Y_{it} = \gamma_0 + \gamma_a A_{it} + \gamma_k K_{it} + \gamma_l L_{it} + \omega_{it} + \veps_{it}
    \end{displaymath}
    where $Y_{it}$ is the log-output of firm $i$ at time $t$, $A_{it}$ is the firm's age, and $K_{it}$ and $L_{it}$ are the logs of capital and labor inputs, respectively. In this model, $\omega_{it}$ and $\veps_{it}$ are productivity shocks that are unobservable to the econometrician, while the firm observes $\omega_{it}$.

    \citet{olley1996} assume that firm $i$'s investment, $I_{it}$, is a strictly increasing function of $\omega_{it}$, so that $I_{it} = g_t(\omega_{it}, A_{it}, K_{it})$. They then invert this function to obtain $\omega_{it} = h_t(I_{it}, A_{it}, K_{it})$, where $h_t = g_t^{-1}(\cdot, A_{it}, K_{it})$. When $h_t$ is specified as a series function of its arguments, for example, a linear function $h_t = h_I I_{it} + h_A A_{it} + h_K K_{it}$ (for simplicity, the coefficients do not vary with $t$), the production function becomes
    \begin{displaymath}
        Y_{it} = \gamma_0 + \tilde\gamma_a A_{it} + \tilde\gamma_k K_{it} + \gamma_L L_{it} + h_I I_{it} + \veps_{it},
    \end{displaymath}
    where $\tilde\gamma_a = \gamma_a + h_A$ and $\tilde\gamma_k = \gamma_k + h_K$. \citet{olley1996} then exploit additional moment restrictions implied by the model to separately identify $(\gamma_a, \gamma_k)$ and $(h_A, h_K)$. This approach has been extended and generalized by \citet{levinsohn2003} and \citet{ackerberg2015}.
    
    In a recent contribution, \citet{kasahara2023} estimated a version of this model using a finite mixture specification for $(\gamma_0, \gamma_a, \gamma_k, \gamma_L)$, where they found an evidence of heterogeneity in these coefficients.

\end{example}

This paper also considers an extension of (\ref{eq.crc}) that also involves regressors with homogeneous coefficients. Let $M_{it} = ({Z_{it}^{homo}}', {X_{it}^{homo}}')' \in \mb{R}^{q_m+p_m}$ be another vector of regressors, where $Z_{it}^{homo}$ is a vector of strictly exogenous regressors and $X_{it}^{homo}$ is a vector of sequentially exogenous regressors. Consider the model
\begin{equation}
    Y_{it} = R_{it}'B_i + M_{it}'\delta + \veps_{it}, \qquad t=1, \ldots, T,
    \label{eq.crc.constant}
\end{equation}
where $\delta \in \mb{R}^{q_m+p_m}$ is an unknown parameter, and assume that
\begin{equation}
    \E(\veps_{it}|B_i, Z_i, X_i^t, Z_i^{homo}, (X_i^{homo})^t) = 0,
    \label{eq.meanindep.constant}
\end{equation}
where $Z_i^{homo} = ({Z_{i1}^{homo}}', \ldots, {Z_{iT}^{homo}}')'$ is the full history and $(X_i^{homo})^t = ({X_{i1}^{homo}}', \ldots, {X_{it}^{homo}}')'$ is the history up to time $t$. While I consider (\ref{eq.crc}) as the main model of interest, I will also discuss how the results extend to the model in (\ref{eq.crc.constant}) in the context of the mean parameters.

The results of this paper also extend to a multivariate version of (\ref{eq.crc}), namely, a system of random coefficient models. For example, one can combine the models in (\ref{eq.income}) and (\ref{eq.consumption}) to develop a joint lifecycle model of earnings and consumption behavior. This multivariate model permits the coefficients from the two processes to freely correlate among themselves and with $(Y_{i0}, A_{i1})$, allowing for potential correlation between the earnings and consumption processes. A full description of the multivariate model is provided in Online Appendix \ref{sec.appendix.multivar}.

\section{Identification of the mean parameters}

\label{sec.mean}

This section and the next section present identification results for the dynamic random coefficient model defined in (\ref{eq.crc}) and (\ref{eq.meanindep}). This section focuses on the identification of the mean parameters, and the next section extends the results to a more general class of parameters. Consider the mean of the random coefficient distribution:
\begin{displaymath}
    \mu_e = \E(e_\gamma'\gamma_i + e_\beta' \beta_i) = \E(e'B_i)
\end{displaymath}
where $e_\gamma$ and $e_\beta$ are real-valued vectors chosen by the econometrician and $e = (e_\gamma', e_\beta')'$. For example, if $e_\gamma = 0$ and $e_\beta = (1,0,\ldots,0)'$, then $\mu_e$ is the mean of the first entry of $\beta_i$.

This section is organized into four subsections. In the first, I show that $\mu_e$ is generally not point-identified. In the second, I show that $\mu_e$ is partially identified, for which I derive closed-form expressions for the finite lower and upper bounds of $\mu_e$. The third subsection then derives the closed-form bounds of $\mu_e$ when the model also includes regressors with homogeneous coefficients. Lastly, the fourth subsection provides a numerical illustration on the size of the closed-form bounds presented in this section. The results presented in this section are special cases of the more general results discussed in the next section and in Online Appendix \ref{sec.appendix.gmm}.

\subsection{Failure of point identification}

\label{sec.mean.pointid}

This subsection shows that $\mu_e$ is generally not point-identified, by considering a specific example of (\ref{eq.crc}) and showing that $\mu_e$ is not point-identified in that example.

The example considered is the AR(1) model with heterogeneous coefficients in which two waves are observed:
\begin{equation}
    Y_{it} = \gamma_i + \beta_iY_{i,t-1} + \veps_{it}, \qquad \E(\veps_{it}|\gamma_i, \beta_i, Y_i^{t-1})=0, \qquad t=1,2.
    \label{eq.ar1}
\end{equation}

The following proposition states that $\E(\beta_i)$ is not point-identified in this model, which implies that there exists no consistent estimator for $\E(\beta_i)$. %

\begin{proposition}

    \label{prop.mean.failure}

    Consider the model defined in (\ref{eq.ar1}). Assume that $(Y_{i0}, Y_{i1}, Y_{i2}, \gamma_i, \beta_i) \in \mc{C}$, where $\mc{C}$ is a compact subset of $\mb{R}^5$. Also assume that $(Y_{i0}, Y_{i1}, Y_{i2}, \gamma_i, \beta_i)$ is absolutely continuous with respect to the Lebesgue measure and that its joint density is strictly positive on $\mc{C}$. Then, %
    $\E(\beta_i)$ is not point-identified.

\end{proposition}

\citet{chamberlain1993,chamberlain2022} showed that $\E(\beta_i)$ is not point-identified in a version of (\ref{eq.ar1}) where the regressor is discrete and $\veps_{it}$ is mean-independent of the regressor. \Cref{prop.mean.failure} complements this result by showing that point identification also fails under stronger assumptions and with the continuous regressor. The failure of point identification in both the discrete and continuous cases in (\ref{eq.ar1}) suggests that this is a general feature of dynamic random coefficient models.

An intuition for \Cref{prop.mean.failure} is as follows. Taking the first difference of (\ref{eq.ar1}) gives
\begin{displaymath}
    Y_{i2}-Y_{i1} = \beta_i (Y_{i1}-Y_{i0}) + \veps_{i2}-\veps_{i1}.
\end{displaymath}
Since $\veps_{i2}-\veps_{i1}$ has zero mean conditional on $(\gamma_i,\beta_i,Y_{i0})$, I obtain
\begin{displaymath}
    \E(Y_{i2}-Y_{i1}|\gamma_i,\beta_i,Y_{i0}) = \E(\beta_i (Y_{i1}-Y_{i0})|\gamma_i,\beta_i,Y_{i0}),
\end{displaymath}
which can be rewritten as
\begin{equation}
    \E(Y_{i2}-Y_{i1} - \beta_i (Y_{i1}-Y_{i0})|\gamma_i,\beta_i,Y_{i0}) = 0.
    \label{eq.ar1.diff}
\end{equation}
Now, consider a function $k(\gamma_i,\beta_i,Y_{i0},Y_{i1})$ that is orthogonal to $Y_{i1}-Y_{i0}$ conditional on $(\gamma_i,\beta_i,Y_{i0})$, i.e.,
\begin{displaymath}
    \E(k(\gamma_i,\beta_i,Y_{i0},Y_{i1})(Y_{i1}-Y_{i0})|\gamma_i,\beta_i,Y_{i0}) = 0.
\end{displaymath}
For example, in the proof of \Cref{prop.mean.failure}, I choose such a function $k$ to be
\begin{displaymath}
    k(\gamma_i, \beta_i, Y_{i0},Y_{i1}) = 1 - \frac{\E(Y_{i1}-Y_{i0}|\gamma_i, \beta_i, Y_{i0})}{\E((Y_{i1}-Y_{i0})^2|\gamma_i, \beta_i, Y_{i0})}(Y_{i1}-Y_{i0}).
\end{displaymath}
Then, it follows that (\ref{eq.ar1.diff}) holds true even if the original random coefficients $(\gamma_i, \beta_i)$ are replaced with the following modified random coefficients:
\begin{displaymath}
    \begin{aligned}
        \tilde\gamma_i &= \gamma_i - Y_{i1}k(\gamma_i,\beta_i,Y_{i0},Y_{i1}), \\
        \tilde\beta_i &= \beta_i + k(\gamma_i,\beta_i,Y_{i0},Y_{i1}).
    \end{aligned}
\end{displaymath}
To see this, note first that
\begin{displaymath}
    \begin{aligned}
        &\E(\tilde\beta_i (Y_{i1}-Y_{i0})|\gamma_i,\beta_i,Y_{i0}) \\
        &= \E(\beta_i (Y_{i1}-Y_{i0})|\gamma_i,\beta_i,Y_{i0}) + \E(k(\gamma_i,\beta_i,Y_{i0},Y_{i1})(Y_{i1}-Y_{i0})|\gamma_i,\beta_i,Y_{i0}) \\
        &= \E(\beta_i (Y_{i1}-Y_{i0})|\gamma_i,\beta_i,Y_{i0})
    \end{aligned}
\end{displaymath}
by the orthogonality property of $k$. Note also that conditioning on $(\gamma_i, \beta_i, \tilde\gamma_i, \tilde\beta_i, Y_{i0}, Y_{i1})$ is equivalent to conditioning on $(\gamma_i, \beta_i, Y_{i0}, Y_{i1})$ since $(\tilde\gamma_i, \tilde\beta_i)$ is a deterministic function of $(\gamma_i, \beta_i, Y_{i0}, Y_{i1})$. Then, by the law of iterated expectations and by (\ref{eq.ar1.diff}), it follows that (\ref{eq.ar1.diff}) holds true for the modified random coefficients $(\tilde\gamma_i, \tilde\beta_i)$:
\begin{displaymath}
    \begin{aligned}
        &\E(Y_{i2}-Y_{i1} - \tilde\beta_i (Y_{i1}-Y_{i0})|\tilde\gamma_i,\tilde\beta_i,Y_{i0}) \\
        &= \E(\E(Y_{i2}-Y_{i1} - \tilde\beta_i (Y_{i1}-Y_{i0})|\tilde\gamma_i,\tilde\beta_i, \gamma_i, \beta_i, Y_{i0},Y_{i1})|\tilde\gamma_i,\tilde\beta_i,Y_{i0}) \\
        &= \E(\E(Y_{i2}-Y_{i1} - \tilde\beta_i (Y_{i1}-Y_{i0})|\gamma_i, \beta_i, Y_{i0},Y_{i1})|\tilde\gamma_i,\tilde\beta_i,Y_{i0}) \\
        &= \E(\E(Y_{i2}-Y_{i1} - \tilde\beta_i (Y_{i1}-Y_{i0})|\gamma_i, \beta_i, Y_{i0})|\tilde\gamma_i,\tilde\beta_i,Y_{i0}) \\
        &= \E(\E(Y_{i2}-Y_{i1} - \beta_i (Y_{i1}-Y_{i0})|\gamma_i, \beta_i, Y_{i0})|\tilde\gamma_i,\tilde\beta_i,Y_{i0})
        = \E(0|\tilde\gamma_i,\tilde\beta_i,Y_{i0}) = 0.
    \end{aligned}
\end{displaymath}
However, if the function $k$ is chosen such that $\E(k(\gamma_i,\beta_i,Y_{i0},Y_{i1})) \neq 0$, which is true for the choice of $k$ above, it follows that $\E(\tilde\beta) \neq \E(\beta)$.

Another intuition for \Cref{prop.mean.failure} follows from an alternative proof of \Cref{prop.mean.failure}, which uses that $\E(\beta_i)$ is point-identified if and only if there exists an unbiased estimator of $\beta_i$ in the individual time series. I state this result as a separate lemma below, which follows as a corollary of the general result in Online Appendix \ref{sec.appendix.gmm}.

\begin{lemma}

    \label{lemma.mean.failure}

    Suppose that the assumptions of \Cref{prop.mean.failure} hold, and that the regularity conditions stated as \Cref{ass.appendix.regularity} in Online Appendix \ref{sec.appendix.gmm} hold. Then $\E(\beta_i)$ is point-identified if and only if there exists a function $S^*(Y_{i0}, Y_{i1}, Y_{i2})$, which is a linear functional on the space of finite and countably additive signed Borel measures that are absolutely continuous with respect to the Lebesgue measure, such that
    \begin{displaymath}
        \E(S^*(Y_{i0}, Y_{i1}, Y_{i2})|\beta_i) = \beta_i
    \end{displaymath}
    almost surely. When such $S^*$ exists, $\E(\beta_i)$ is identified by $\E(\beta_i) = \E(S^*(Y_{i0}, Y_{i1}, Y_{i2}))$.

\end{lemma}

\Cref{prop.mean.failure} can then be proved by showing that there is no unbiased estimator of $\beta_i$ (see Online Appendix \ref{sec.appendix.proof}). The intuition for \Cref{lemma.mean.failure} is as follows. Since the distribution of $\beta_i$ is unrestricted, information about individual $\beta_i$ can only be obtained from its own time series. In a long panel context, a time series estimator of $\beta_i$ that is consistent as $T \rightarrow \infty$ would reliably provide such information. In a short panel context, however, such an estimator is not reliable because $T$ is finite. \Cref{lemma.mean.failure} shows that a time series estimator that is unbiased for finite $T$ is the only reliable source of information on $\beta_i$ when it comes to point identification in short panels.

\subsection{Partial identification}

\label{sec.mean.partialid}

A natural question following the last subsection is whether the data are at all informative about $\mu_e = \E(e'B_i)$, or whether they provide no information. This subsection shows that the data are indeed informative about $\mu_e$. I show that there exist finite bounds $L$ and $U$ such that
\begin{displaymath}
    L \leq \mu_e \leq U
\end{displaymath}
where $L$ and $U$ are estimable from the observed data.

To identify $\mu_e$, I use unconditional moment restrictions that are implications of (\ref{eq.meanindep}). It is known that the set of unconditional moment restrictions of the form
\begin{equation}
    \E(g(B_i, Z_i, X_i^t)\veps_{it}) = 0,
    \label{eq.orthogonal}
\end{equation}
indexed by a suitable class of functions $g$, is equivalent to the conditional moment restriction in (\ref{eq.meanindep}) \citep{bierens1990,%
andrews2013}. I choose the class of $g$ to be the set of polynomial functions and select a finite subset of these functions. This yields a finite number of unconditional moment restrictions that are fixed in the asymptotics that $N \rightarrow \infty$. This finite set of unconditional moment restrictions contains less information than the full conditional moment restriction in (\ref{eq.meanindep}), yielding an outer bound rather than the sharp bound, but it leads to estimation and inference procedures that are computationally tractable. In addition, the empirical application in \Cref{sec.application} shows that this finite set is sufficiently restrictive to provide informative bounds. Partial identification results based on the full conditional moment restriction in (\ref{eq.meanindep}) are presented in Online Appendix \ref{sec.appendix.gmm}.

I now study the identification of $\mu_e$. Recall the dynamic random coefficient model defined in (\ref{eq.crc}) and (\ref{eq.meanindep}):
\begin{displaymath}
    Y_{it} = R_{it}'B_i + \veps_{it}, \qquad \E(\veps_{it}|B_i, Z_i, X_i^t) = 0, \qquad t=1, \ldots, T,
\end{displaymath}
where $R_{it} = (Z_{it}', X_{it}')'$. For brevity of notation, define
\begin{displaymath}
    Y_i \equiv \left(\begin{array}{c}
        Y_{i1} \\
        \vdots \\
        Y_{iT}
    \end{array}\right) \qand
    R_i \equiv \left(\begin{array}{c}
        R_{i1}' \\
        \vdots \\
        R_{iT}'
    \end{array}\right)
\end{displaymath}
as a random vector and a random matrix stacking $Y_{it}$ and $R_{it}'$ rowwise across $t$, respectively. Consider the following assumptions:

\begin{assumption}
    \label{ass.crc}
    $(Y_i, Z_i, X_i, B_i)$ satisfies (\ref{eq.crc}) and (\ref{eq.meanindep}).
\end{assumption}

\begin{assumption}
    \label{ass.mean.nomulticollinearity}
    $R_i'R_i$ is positive definite with probability 1.
\end{assumption}

\Cref{ass.crc} states that the dynamic random coefficient model is correctly specified. \Cref{ass.mean.nomulticollinearity} is a no-multicollinearity assumption imposed on the individual time series. This is stronger than the assumption that $\E(R_i'R_i)$ is positive definite, a common assumption in standard dynamic fixed effect models. A stronger assumption is required because $B_i$ is individual-specific with an unrestricted distribution, and each $B_i$ can only be learned from its own individual data\footnote{\citet{graham2012} studied a violation of \Cref{ass.mean.nomulticollinearity} in a non-dynamic context.}~\footnote{If \Cref{ass.mean.nomulticollinearity} is violated because of $Z_i$, one may choose to consider the subpopulation where $\text{det}(Z_i'Z_i) \geq d_0$ for some $d_0 > 0$. The results of this paper extend straightforwardly to this subpopulation because $\veps_{it}$ is assumed to have zero mean conditional on $Z_i$, as stated in (\ref{eq.meanindep}). By contrast, the results of this paper do not extend to the subpopulation for which $\text{det}(X_i'X_i) \geq d_0$, since $\veps_{it}$ has zero mean only conditional on the current history of $X_{it}$, not the full history.\label{footnote.trimming}}.

I now state a theorem showing that $\mu_e$ is partially identified under \Cref{ass.crc,ass.mean.nomulticollinearity}. This theorem is a special case of \Cref{prop.gmm} presented in the next section. For brevity of notation, define
\begin{displaymath}
    \widehat{B}_i = (R_i'R_i)^{-1}R_i'Y_i, \qand B_0 = \E(R_i'R_i)^{-1}\E(R_i'Y_i).
\end{displaymath}

\begin{theorem}

    \label{prop.mean.closedform}

    Suppose that \Cref{ass.crc,ass.mean.nomulticollinearity} hold. 
    Then $L \leq \mu_e \leq U$ where
    \begin{displaymath}
        [L, U] = \left[\mathcal{B}_R - \frac{1}{2}\sqrt{\mathcal{E}_R\mathcal{D}_R}, ~~\mathcal{B}_R + \frac{1}{2}\sqrt{\mathcal{E}_R\mathcal{D}_R}\right],
    \end{displaymath}
    and
    \begin{displaymath}
        \begin{aligned}
            \mathcal{B}_R &= \frac{1}{2}e'\E(\widehat B_i) + \frac{1}{2}e'B_0, \\
            \mathcal{E}_R &= e'\E((R_i'R_i)^{-1})e - e'\E(R_i'R_i)^{-1}e, \\
            \mathcal{D}_R &= \E(Y_iR_i(R_i'R_i)^{-1}R_i'Y_i) - \E(Y_iR_i)\E(R_i'R_i)^{-1}\E(R_i'Y_i).
        \end{aligned}
    \end{displaymath}
    In addition, $\mathcal{E}_R \geq 0$ and $\mathcal{D}_R \geq 0$, and each is equal to zero if and only if $(R_i'R_i)^{-1}e$ and $(R_i'R_i)^{-1}R_i'Y_i$ are degenerate across individuals, respectively.
\end{theorem}

Note that $\widehat B_i$ is the individual-specific OLS estimator of $B_i$ from its individual time series, $B_0$ is the pooled OLS estimator obtained by considering $B_i$ as constant across individuals, and $R_i'R_i$ is the squared design matrix of the individual time series.

The closed-form expressions in \Cref{prop.mean.closedform} provide intuition on when $L$ and $U$ are finite. In particular, $L$ and $U$ are finite even if $(Y_i, R_i, B_i)$ are unbounded, as long as the moments involved in the expression for $[L,U]$ are finite --- that is, $\E(R_i'R_i)$, $\E((R_i'R_i)^{-1})$, $\E(R_i'Y_i)$, $\E((R_i'R_i)^{-1}R_i'Y_i)$, and $\E(Y_iR_i(R_i'R_i)^{-1}R_i'Y_i)$ are finite.

The general result in \Cref{prop.gmm} presented in the next section provides insights on what type of information is used to construct the bounds in \Cref{prop.mean.closedform}. It can be shown that, under the additional regularity conditions stated as \Cref{ass.proof.regularity} in the next section, the bounds in \Cref{prop.mean.closedform} are the sharp bounds of $\mu_e$ when the conditional moment restriction (\ref{eq.meanindep}) is replaced by the following unconditional moment restrictions:
\begin{equation}
    \E\left(\sum_{t=1}^T (R_{it}'B_i)\veps_{it}\right) = 0, \qand \E\left(\sum_{t=1}^T R_{it}\veps_{it}\right) = 0,
    \label{eq.orthogonal.closedform}
\end{equation}
where the first restriction is interpreted as that the ``error term'' ($\veps_{it}$) is orthogonal to the ``explained term'' ($R_{it}'B_i$), and the second is interpreted as that $\veps_{it}$ is orthogonal to the current-period regressors $R_{it}$, on average across $t$. For empirical applications, the amount of information contained in (\ref{eq.orthogonal.closedform}) is small relative to that in (\ref{eq.meanindep}), and its refinement will be discussed later in this subsection. From a theoretical perspective, \Cref{prop.mean.closedform} suggests that the two moment conditions in (\ref{eq.orthogonal.closedform}) are the key identifying restrictions that yield finite $L$ and $U$, out of the infinite number of unconditional moment restrictions in (\ref{eq.orthogonal}) that is equivalent to (\ref{eq.meanindep}).

I now explain the intuition behind \Cref{prop.mean.closedform}, focusing on the upper bound $U$. Consider a Lagrangian where the objective function is the parameter of interest $e'B_i$ and the constraints are the moment functions in (\ref{eq.orthogonal.closedform}):
\begin{displaymath}
    Q(\lambda, \mu, W_i, B_i) = e'B_i + \lambda \sum_{t=1}^T (R_{it}'B_i)\veps_{it} + \mu' \sum_{t=1}^TR_{it}\veps_{it},
\end{displaymath}
where $\lambda \in \mb{R}$ and $\mu$ has the same dimension as $R_{it}$. Note that $\E(Q) = \E(e'B_i) = \mu_e$ because the constraints have zero expectations by (\ref{eq.orthogonal.closedform}).

If I substitute $\veps_{it} = Y_{it} - R_{it}'B_i$ into $Q$ and use the matrix notations $R_i$ and $Y_i$, I obtain the expression:
\begin{displaymath}
    Q(\lambda, \mu, W_i, B_i) = e'B_i + \lambda Y_i'R_iB_i - \lambda B_i'(R_i'R_i)B_i + \mu'R_i'Y_i - \mu'R_i'R_iB_i.
\end{displaymath}
This is a quadratic polynomial in $B_i$ whose second-order derivative is
\begin{displaymath}
    \frac{d^2Q}{dB_idB_i'} = - 2\lambda (R_i'R_i).
\end{displaymath}
If $\lambda>0$, then this second-order derivative is a negative definite matrix, in which case $Q$ attains a global maximum at the solution to the first-order condition $dQ/dB_i = 0$. Let $P = \max_{b \in \mb{R}^{q+p}}Q(\lambda, \mu, W_i, b)$ be the resulting maximum, which is only a function of $(\lambda, \mu, W_i)$ since $B_i$ is ``maximized out.'' Then, by construction:
\begin{displaymath}
    P(\lambda, \mu, W_i) = \max_{b \in \mb{R}^{q+p}}Q(\lambda, \mu, W_i, b) \geq Q(\lambda, \mu, W_i, B_i),
\end{displaymath}
which implies
\begin{displaymath}
    \E(P(\lambda, \mu, W_i)) \geq \E(Q(\lambda, \mu, W_i, B_i)) = \mu_e.
\end{displaymath}
This shows that $\E(P)$ is an upper bound of $\mu_e$ for any choice of $\lambda > 0$ and $\mu$. I then obtain a smallest upper bound for $\mu_e$ by minimizing $\E(P)$ with respect to $\lambda > 0$ and $\mu$:
\begin{displaymath}
    \min_{\lambda > 0, ~ \mu}\E(P(\lambda, \mu, W_i)) \geq \mu_e.
\end{displaymath}
This coincides with $U$ in \Cref{prop.mean.closedform}. The lower bound can be obtained by repeating the same process with $\lambda < 0$.

As discussed earlier, the amount of information used to construct the bounds in \Cref{prop.mean.closedform}, namely the moment restrictions in (\ref{eq.orthogonal.closedform}), is small relative to that in (\ref{eq.meanindep}). I now develop a refinement of \Cref{prop.mean.closedform}.

For each $t$, choose a vector of observable random variables $S_{it}$ such that $\E(S_{it}\veps_{it}) = 0$ under (\ref{eq.meanindep}). For example, one may choose $S_{it}$ to be $S_{it} = R_{it}$ (the vector of current regressors) or $S_{it} = (Z_i', {X_i^t}')'$ (the vector of the full history of $Z_{it}$ and the current history of $X_{it}$). One may also choose $S_{it}$ to include the square terms such as $X_{it}^2$ and $Z_{it}^2$. The dimension of $S_{it}$ is allowed to vary across $t$. Consider the following assumption:

\begin{assumption}
    \label{ass.mean.nomulticollinearity.s}
    For every nonzero vector $a = (a_1', \ldots, a_T')'$, $\pp(\sum_{t=1}^T R_{it} S_{it}' a_t \neq 0) > 0$.
\end{assumption}

Recall that one chooses $S_{it}$. Assumption 3 is a regularity condition requiring that each entry of $\E(S_{it}\veps_{it}) = 0$, for $t = 1, \ldots, T$, contains distinct information. It is implied by \Cref{ass.mean.nomulticollinearity} if $S_{it} = R_{it}$ and $R_{it}$ consists of an intercept and a continuous regressor. \Cref{ass.mean.nomulticollinearity.s} is trivially violated if $S_{it}$ includes duplicate variables, for example, if $S_{it} = (X_{it}', X_{it}', Z_{it}')'$. However, it is not necessarily violated if $S_{it}$ and $S_{iv}$ for $t \neq v$ have duplicate variables. In the empirical application, I estimate the model with $R_{it} = (1, Y_{i,t-1})'$ using $S_{it} = (1, Y_{i,t-1}, \ldots, Y_{i,\max\{0,t-5\}})'$. If Assumption 3 fails to hold for a particular vector $a$, its nonzero entries indicate which entry of $S_{it}$ is redundant, and one can drop the corresponding entry.

I now state a refinement of \Cref{prop.mean.closedform} under \Cref{ass.crc,ass.mean.nomulticollinearity,ass.mean.nomulticollinearity.s}. For brevity of notation, define a block diagonal matrix
\begin{displaymath}
    S_i \equiv \left(\begin{array}{cccc}
        S_{i1} & 0 & \cdots & 0 \\
        0 & S_{i2} & \cdots & 0 \\
        \vdots & \vdots & & \vdots \\
        0 & 0 & \cdots & S_{iT}
    \end{array}\right)
\end{displaymath}
where $S_{it}$ appears in the diagonal as a column vector, so that $S_i$ has $T$ columns. In addition, define
\begin{displaymath}
    \begin{aligned}
        \mc{V}_S &= \E(S_iR_i(R_i'R_i)^{-1}R_i'S_i'), \\
        \mc{Y}_S &= \E(S_iR_i(R_i'R_i)^{-1}R_i'Y_i), \\
        \mc{P}_S &= \E(S_iR_i(R_i'R_i)^{-1}), \\
        Y_S &= \E(S_iY_i), \\
        m_0 &= \E(Y_i'R_i(R_i'R_i)^{-1}R_i'Y_i).
    \end{aligned}
\end{displaymath}

\begin{proposition}
    \label{prop.mean.closedform.refined}
    Suppose that \Cref{ass.crc,ass.mean.nomulticollinearity,ass.mean.nomulticollinearity.s} hold. Then $\mc{V}_S$ is invertible, and $L_S \leq \mu_e \leq U_S$ where
    \begin{displaymath}
        [L_S, U_S] = \left[\mathcal{B}_S - \frac{1}{2}\sqrt{\mathcal{E}_S\mathcal{D}_S}, ~~\mathcal{B}_S + \frac{1}{2}\sqrt{\mathcal{E}_S\mathcal{D}_S}\right]
    \end{displaymath}
    and
    \begin{displaymath}
        \begin{aligned}
            \mathcal{B}_S &= \frac{1}{2}e'\E(\widehat{B}_i) + \frac{1}{2}e'\mc{P}_S'\mc{V}_S^{-1}(2Y_S - \mc{Y}_S), \\
            \mathcal{E}_S &= e'\E((R_i'R_i)^{-1})e-e'\mc{P}_S'\mc{V}_S^{-1}\mc{P}_Se, \\
            \mathcal{D}_S &= m_0-(2Y_S-\mc{Y}_S)'\mc{V}_S^{-1}(2Y_S-\mc{Y}_S).
        \end{aligned}
    \end{displaymath}
\end{proposition}

Similarly to \Cref{prop.mean.closedform}, under the additional regularity conditions stated as \Cref{ass.proof.regularity} in the next section, it can be shown that the bounds in \Cref{prop.mean.closedform.refined} are the sharp bounds of $\mu_e$ if (\ref{eq.meanindep}) is replaced by the following unconditional moments:
\begin{equation}
    \E\left(\sum_{t=1}^T (R_{it}'B_i)\veps_{it}\right) = 0, \qand \E\left(S_{it}\veps_{it}\right) = 0 \quad\text{for}\quad t=1, \ldots, T,
    \label{eq.orthogonal.closedform.refined}
\end{equation}
where the first expression gives one unconditional moment restriction, and the second expression gives $\text{dim}(S_{it})$ unconditional moment restrictions for each $t$. 

While (\ref{eq.orthogonal.closedform.refined}) still contains less information than (\ref{eq.meanindep}), it is found to be sufficiently informative in practice. %
The empirical application in \Cref{sec.application} shows that \Cref{prop.mean.closedform.refined} can produce informative bounds. In addition, the closed-form expressions in \Cref{prop.mean.closedform.refined} lead to a simple estimation and inference procedure that is robust to overidentification and model misspecification, which are generally not simple to deal with in partially identified models.%

\subsection{Extension to models with homogeneous coefficients}

\label{sec.mean.homo}

In this subsection, I extend the partial identification results of the previous subsection to the model that also involves homogeneous coefficients. Recall the model introduced in (\ref{eq.crc.constant}) and (\ref{eq.meanindep.constant}):
\begin{displaymath}
    Y_{it} = R_{it}'B_i + M_{it}'\delta + \veps_{it}, \qquad \E(\veps_{it}|B_i, Z_i, X_i^t, Z_i^{homo}, (X_i^{homo})^t) = 0, \qquad t=1, \ldots, T,
\end{displaymath}
where $M_{it} = ({Z_{it}^{homo}}', {X_{it}^{homo}}')'$ denotes the regressors with homogeneous coefficients. Let $U_{it} = (R_{it}', M_{it}')'$ be the vector of all regressors, and let $M_i$ and $U_i$ be random matrices stacking $M_{it}'$ and $U_{it}'$ rowwise across $t$, hence having $T$ rows, respectively.

Similarly to the last subsection, choose a vector of observable random variables $S_{it}$ such that $\E(S_{it}\veps_{it}) = 0$ under (\ref{eq.meanindep.constant}). Consider the following modifications to \Cref{ass.crc,ass.mean.nomulticollinearity} and the restatement of \Cref{ass.mean.nomulticollinearity.s}.

\begin{assumption}
    \label{ass.crc.constant}
    $(Y_i, R_i, M_i, B_i)$ and $\delta$ satisfy (\ref{eq.crc.constant}) and (\ref{eq.meanindep.constant}).
\end{assumption}

\begin{assumption}
    \label{ass.mean.nomulticollinearity.constant}
    $R_i'R_i$ is positive definite with probability 1. In addition, for every nonzero vector $a \in \mb{R}^{q_m+p_m}$, $\pp(M_ia \notin \text{col}(R_i)) > 0$.
\end{assumption}

\begin{assumption}
    \label{ass.mean.nomulticollinearity.s.constant}
    For every nonzero vector $a = (a_1', \ldots, a_T')'$, $\pp(\sum_{t=1}^T R_{it} S_{it}' a_t \neq 0) > 0$.
\end{assumption}

\Cref{ass.mean.nomulticollinearity.constant} requires that, for every nonzero linear combination of $M_i$, it is not multicollinear with $R_i$ with positive probability, that is, at least for some individuals. Note that a necessary condition for \Cref{ass.mean.nomulticollinearity.constant} is that $\E(M_i'M_i)$ is positive definite, rather than $M_i'M_i$ itself. Therefore, the no-multicollinearity requirement for $M_{it}$ is the same as those for the regressors in standard fixed effect models.

Under these assumptions, the following proposition extends the bounds in \Cref{prop.mean.closedform.refined} to the model defined in (\ref{eq.crc.constant}) and (\ref{eq.meanindep.constant}). For brevity of notation, define
\begin{displaymath}
    \begin{aligned}
        \mc{V}_M &= \E(M_i'R_i(R_i'R_i)^{-1}R_i'M_i), & \mathcal{C} &= \E(S_iR_i(R_i'R_i)^{-1}R_i'M_i), \\
        \mc{Y}_M &= \E(M_i'R_i(R_i'R_i)^{-1}R_i'Y_i), & C &= \E(S_iM_i), \\
        \mc{P}_M &= \E(M_i'R_i(R_i'R_i)^{-1}),        & M_0 &= \E(M_i'M_i).\\
        Y_M &= \E(M_i'Y_i), \\
    \end{aligned}
\end{displaymath}

\begin{proposition}
    \label{prop.mean.closedform.refined.constant}
    Suppose that \Cref{ass.crc.constant,ass.mean.nomulticollinearity.constant,ass.mean.nomulticollinearity.s.constant} hold. Then both the matrix $\mc{V}_M-M_0$ and the matrix
    \begin{displaymath}
        \mc{V} = \mc{V}_S - (C - \mc{C})(\mc{V}_M-M_0)^{-1}(C - \mc{C})'
    \end{displaymath}
    are invertible, and $L_M \leq \mu_e \leq U_M$ where
    \begin{displaymath}
        [L_M, U_M] = \left[\mathcal{B}_M - \frac{1}{2}\sqrt{\mathcal{E}_M\mathcal{D}_M}, ~~\mathcal{B}_M + \frac{1}{2}\sqrt{\mathcal{E}_M\mathcal{D}_M}\right]
    \end{displaymath}
    and
    {\small
    \begin{displaymath}
        \begin{aligned}
            \mathcal{B}_M &= \frac{1}{2}e'\E(\widehat{B}_i) + \frac{1}{2}e'\mc{P}_M'(\mc{V}_M-M_0)^{-1}(Y_M - \mc{Y}_M) + \\
            &\frac{1}{2}(\mc{P}_Se + (C - \mc{C})(\mc{V}_M-M_0)^{-1}\mc{P}_Me)'\mc{V}^{-1} (2Y_S - \mc{Y}_S + (C - \mc{C})(\mc{V}_M-M_0)^{-1}(Y_M - \mc{Y}_M)), \\
            \mathcal{E}_M &= e'\E((R_i'R_i)^{-1})e - e'\mc{P}_M'(\mc{V}_M-M_0)^{-1}\mc{P}_Me - \\
            &(\mc{P}_Se + (C - \mc{C})(\mc{V}_M-M_0)^{-1}\mc{P}_Me)'\mc{V}^{-1}(\mc{P}_Se + (C - \mc{C})(\mc{V}_M-M_0)^{-1}\mc{P}_Me), \\
            \mathcal{D}_M &= m_0 - (Y_M - \mc{Y}_M)'(\mc{V}_M-M_0)^{-1}(Y_M - \mc{Y}_M) - \\
            (2Y_S &- \mc{Y}_S + (C - \mc{C})(\mc{V}_M-M_0)^{-1}(Y_M - \mc{Y}_M))'\mc{V}^{-1}  (2Y_S - \mc{Y}_S + (C - \mc{C})(\mc{V}_M-M_0)^{-1}(Y_M - \mc{Y}_M)).
        \end{aligned}
    \end{displaymath}
    }
\end{proposition}

The empirical application in \Cref{sec.application} shows that \Cref{prop.mean.closedform.refined.constant} can produce highly informative bounds. The empirical application involves a total of 59 regressors in $M_{it}$ (and 58 in an alternative specification), demonstrating the practicality of \Cref{prop.mean.closedform.refined.constant} even when the number of regressors with homogeneous coefficients is large.

\subsection{Numerical illustration}

\label{sec.mean.numerical}

This subsection provides a numerical illustration of the sizes of the identified sets presented in the previous sections in a simple panel data model. This highlights the practical implications of considering unconditional moment restrictions instead of the conditional ones. Specifically, consider the model
\begin{displaymath}
    Y_{it} = \gamma_i + \beta_i X_{it} + \veps_{it}, \qquad t=1, \ldots, T,
\end{displaymath}
where
\begin{equation}
    \E(\veps_{it}|\gamma_i, \beta_i, X_i^{t}) = 0, \qquad t=1, \ldots, T.
    \label{eq.illustration.meanindep}
\end{equation}
For this model, I numerically compute the sharp identified set of $\E(\beta_i)$ under the conditional moment restriction in (\ref{eq.illustration.meanindep}), and compare it to the outer identified set in \Cref{prop.mean.closedform.refined}, which are based on the unconditional moment restrictions in (\ref{eq.orthogonal.closedform.refined}). Computation of the sharp identified set of $\E(\beta_i)$ is generally prohibitively expensive (see the discussion in \Cref{sec.general}), but it becomes relatively tractable when $(\gamma_i, \beta_i, X_i)$ has a small number of discrete support points, where the sharp characterization reduces to solving optimization problems over a large but finite-dimensional Euclidean spaces. 

Let $\gamma_i$ and $\beta_i$ be independent discrete random variables such that $\gamma_i \in \{-1, 0, 1\}$ with equal probabilities and $\beta_i \in \{0, 0.5, 1\}$ with equal probabilities. In addition, let $\veps_{it}$ be independent of $(\gamma_i, \beta_i)$ and $\veps_{it} \in \{-1, 0, 1\}$ with equal probabilities. Lastly, let $X_{i1} = 1$, and define for $t \geq 2$:
\begin{displaymath}
    \begin{aligned}
        X_{it} = \left\{\begin{array}{ll}
            -1 & \text{ if } Y_{i,t-1} < -1, \\
            0 & \text{ if } -1 \leq Y_{i,t-1} < 1, \\
            1 & \text{ if } Y_{i,t-1} \geq 1, \\
        \end{array}\right.
    \end{aligned}
\end{displaymath}
so that $X_{it}$ depends on $Y_{i,t-1}$.

Under this data generating process, I compute both the sharp and the outer bounds of $\E(\beta_i)$ for $T \in \{3, 4, 5\}$. I also calculate the outer bounds for $T \in \{6, 8\}$ to illustrate how the outer bound tightens as $T$ increases. I choose $S_{it} = (1, X_{i1}, \ldots, X_{it})'$ to compute the outer bounds in \Cref{prop.mean.closedform.refined}. The calculated sharp and outer bounds are presented in \Cref{table.illustration}. Although the outer bounds are wider than the sharp bounds, I will show in \Cref{sec.application} that the outer bounds remain sufficiently informative in empirical applications.%

\begin{table}[!htbp]
    \centering %
    \begin{tabular}{c c c c c c c} %
    \hline\hline %
    & $T=3$ & $T=4$ & $T=5$ & $T=6$ & $T=8$\\ %
    \hline\hline %
    Sharp  & [0.401, 0.593] & [0.452, 0.552] & [0.473, 0.532] & - & - \\ 
    Outer  & [0.216, 0.617] & [0.267, 0.613] & [0.306, 0.613] & [0.330, 0.613] & [0.368, 0.598] \\
    \hline %
    \end{tabular}
    \caption{Numerical illustration of the sharp and the outer identified sets. Sharp refers to the sharp identified set of $\E(\beta_i)$ under (\ref{eq.illustration.meanindep}), and Outer refers to the outer bounds of $\E(\beta_i)$ under (\ref{eq.orthogonal.closedform.refined}). The sharp identified sets for $T=6$ and $T=8$ are not computed because they are computationally prohibitive. Note that the data generating process implies $\E(\beta_i) = 0.5$.}
    \label{table.illustration}
\end{table}

\section{Identification of the general parameters}

\label{sec.general}

This section presents a general partial identification result for dynamic random coefficient models. This section is structured into two subsections. First, I present a general partial identification result for a generic parameter. Second, I apply this general result to derive the bounds for the variance and the CDF of the random coefficient distribution.

\subsection{Identification of the general parameters}

Recall that $W_i \in \mc{W}$ is the vector of observable variables and $B_i \in \mc{B}$ is the vector of unobservable random coefficients. I consider parameters of the form
\begin{displaymath}
    \theta = \E(m(W_i, B_i))
\end{displaymath}
where the function $m:\mc{W} \times \mc{B} \mapsto \mb{R}$ is known. I consider a generic set of unconditional moment restrictions:

\begin{assumption}
    \label{ass.gmm}
    The random vectors $(W_i, B_i)$ satisfy:
    \begin{displaymath}
        \begin{aligned}
            \E(\phi_k(W_i, B_i)) &= 0, \quad k=1, \ldots, K,
        \end{aligned}
    \end{displaymath}
    where $\phi_k:\mc{W} \times \mc{B} \mapsto \mb{R}$ are known moment functions and $K \in \mb{N}$ is the number of moment restrictions.
\end{assumption}

Note that, in the asymptotics, $K$ is fixed when $N  \rightarrow \infty$. Note also that $\veps_{it}$ does not appear in \Cref{ass.gmm} because $(W_i, B_i)$ summarizes $\veps_{it}$ by the relationship $\veps_{it} = Y_{it} - R_{it}'B_i$. More generally, without connection to random coefficient models, \Cref{ass.gmm} imposes generic unconditional moment restrictions that involve both observed and unobserved random vectors. A more general formulation that also involves conditional moment restrictions is studied in Online Appendix \ref{sec.appendix.gmm}.

I characterize the sharp identified set of $\theta$ under \Cref{ass.gmm} and the regularity conditions that are introduced below. To do so, I first recast the identification problem as a linear programming problem. I then show that its dual representation yields a tractable characterization of the identified set.

Let $P \in \mc{M}_{W \times B}$, where $\mc{M}_{W \times B}$ is the linear space of finite and countably additive signed Borel measures on $\mc{W} \times \mc{B}$, equipped with the total variation norm. Let $P_W \in \mc{M}_{W}$ be the observed marginal distribution of $W_i$. The sharp identified set $I$ of $\theta$ is \emph{defined} by:
\begin{displaymath}
    \begin{aligned}
        I \equiv \left\{ \int m(w, b) dP ~\right|~  P \in \mc{M}_{W \times B}, \quad P \geq 0, \quad
        & \int dP = 1, \\
        & \int \phi_k(w, b) dP = 0, \quad k=1, \ldots, K, \\
        & \left. \int P(w, db) = P_W(w) ~\text{ for all } w \in \mc{W} \right\}.
    \end{aligned}
\end{displaymath}
The set $I$ is the collection of all $\int m(W_i, B_i) dP$ values over $P$ such that (i) $P$ is a probability distribution of $(W_i, B_i)$, (ii) $P$ satisfies the moment restrictions, and (iii) the marginal distribution of $W_i$ implied from $P$ equals the observed distribution $P_W$. Dependence of $I$ on $m$, $P_W$, and the $\phi_k$s are suppressed in the notation.

All defining properties of $I$ are linear in $P$, which means that $I$ is a convex set in $\mb{R}$ (i.e., an interval). Therefore, $I$ can be characterized by its lower and upper bounds. The sharp lower bound $L$ of $I$ is \emph{defined} by:
\begin{equation}
    \begin{aligned}
        \min_{P \in \mc{M}_{W \times B}, ~P \geq 0} \int m(w, b) dP \st 
        & \int \phi_k(w, b) dP = 0, \quad k=1, \ldots, K, \\
        & \int P(w, db) = P_W(w) ~\text{ for all } w \in \mc{W}.
    \end{aligned}
    \label{eq.primal}
\end{equation}
Note that the constraint $\int dP = 1$ is omitted in (\ref{eq.primal}), because it is implied by the constraint $\int P(w, db) = P_W(w)$ where $P_W$ is a probability distribution.

Equation (\ref{eq.primal}) is a linear program in $P$, with the caveat that $P$ is an infinite-dimensional object. It is not a tractable characterization of $L$ for dynamic random coefficient models, in the sense that the estimation methods it imply are computationally infeasible. For example, discretizing the space of $(W_i, B_i)$ and solving the discretized problem \citep{honore2006,gunsilius2019} is computationally infeasible because the dimension of $(W_i, B_i)$ is large. Recall that $W_i$ contains the full history of regressors and dependent variables and $B_i$ contains all random coefficients. For the random coefficient model with $R$ regressors and $T$ waves, $P$ is a distribution on an $(RT+T+R)$-dimensional space.

My approach is to use the dual representation of (\ref{eq.primal}) obtained by the duality theorem for infinite-dimensional linear programming \citep{galichon2009,schennach2014}. I consider the following regularity conditions:

\begin{assumption}
    \label{ass.proof.regularity}
    The following conditions hold.
    \begin{itemize}
        \item[(i)] $\mc{W} \times \mc{B}$ is a compact set in a Euclidean space.
        \item[(ii)] $(m, \phi_1, \ldots, \phi_K)$ are bounded Borel measurable functions on $\mc{W} \times \mc{B}$.
    \end{itemize}
\end{assumption}

Under these conditions, the following theorem characterizes the sharp identified set of $\theta$ using the dual representation of (\ref{eq.primal}) and the corresponding problem for the upper bound.

\begin{theorem}
    \label{prop.gmm}
    Suppose \Cref{ass.gmm,ass.proof.regularity} hold. Let $\lambda = (\lambda_1, \ldots, \lambda_K)' \in \mb{R}^K$. Then $I = [L, U]$ where
    \begin{equation}
        L = \max_{\lambda \in \mb{R}^K}
        \E\left[ \min_{b \in \mc{B}} \left\{ m(W_i, b) + \sum_{k=1}^{K}\lambda_k \phi_k(W_i, b) \right\} \right]
        \label{eq.lb}
    \end{equation}
    and
    \begin{equation}
        U = \min_{\lambda \in \mb{R}^K}
        \E\left[ \max_{b \in \mc{B}} \left\{ m(W_i, b) + \sum_{k=1}^{K}\lambda_k \phi_k(W_i, b) \right\} \right]
        \label{eq.ub}
    \end{equation}
    provided that the optimization problems in (\ref{eq.lb}) and (\ref{eq.ub}) possess finite solutions.
\end{theorem}

Note that the result in \Cref{prop.gmm} is not specific to dynamic random coefficient models. It is a general duality result for moment equality models where the moment functions involve both observables and unobservables \citep{schennach2014,li2018}.

The characterization that also involves conditional moment restrictions is developed in Online Appendix \ref{sec.appendix.gmm}. To illustrate, consider the following conditional moment restrictions:%
\begin{displaymath}
    \E(\psi_k(W_i, B_i)|\mathsf{W}_{ik}, \mathsf{B}_{ik}) = 0, \qquad k=1, \ldots, K,
\end{displaymath}
where $\mathsf{W}_{ik}$ and $\mathsf{B}_{ik}$ are subvectors of $W_i$ and $B_i$. Note that (\ref{eq.meanindep}) has $T$ moment restrictions of this type, one for each $t=1, \ldots, T$. Assume that $W_i$ and $B_i$ are absolutely continuous with respect to the Lebesgue measure, and that the regularity conditions stated as \Cref{ass.appendix.regularity} (i)-(iii) in Online Appendix \ref{sec.appendix.gmm} hold. Then, under these assumptions, \Cref{prop.appendix.gmm} in Online Appendix \ref{sec.appendix.gmm} implies that the sharp lower bound of $\theta$ is given by
\begin{equation}
    L = \max_{\{\mu_k(\textsf{w}_k,\textsf{b}_k) \in L^2(\mathsf{W}_{ik}, \mathsf{B}_{ik})\}_{k=1}^K}
    \E\left[ \min_{b \in \mc{B}} \left\{ m(W_i, b) + \sum_{k=1}^K \mu_k(\mathsf{W}_{ik}, \textsf{b}_k) \psi_k(W_i, b) \right\} \right],
    \label{eq.lb.conditional}
\end{equation}
where $\textsf{b}_k$ is the subvector of $b$ corresponding to $\textsf{B}_{ik}$ and $\mu_k(\textsf{w}_k,\textsf{b}_k)$ is a square integrable function of $(\mathsf{W}_{ik}, \mathsf{B}_{ik})$, denoted by $L^2(\mathsf{W}_{ik}, \mathsf{B}_{ik})$. Therefore, for conditional moment restrictions, the dual representation involves optimization over the functional choice variables $\{\mu_k(\textsf{w}_k,\textsf{b}_k)\}_{k=1}^K$. In general, such functional optimization is not computationally tractable because the inner optimization problem over $b$ is potentially highly nonconvex and $(W_i,B_i)$ is potentially high-dimensional. For the random coefficient model with $R$ regressors and $T$ waves, each $\mu_k$ is a function on a space of dimension at most $(RT+R)$, and one must optimize over $K$ such functions in (\ref{eq.lb.conditional}). In contrast, (\ref{eq.lb}) involves optimization over the finite-dimensional Euclidean space $\mb{R}^K$. In the previous subsection, I used a parsimonious set of moment restrictions in (\ref{eq.orthogonal.closedform.refined}) to derive closed-form bounds for the mean parameters that are computationally efficient. In the next subsection, using the same parsimonious set, I derive computationally efficient bounds of the variance and the CDF parameters. While these bounds are the outer bounds relative to the sharp bounds in (\ref{eq.lb.conditional}), I demonstrate in the empirical application in \Cref{sec.application} that they produce informative bounds in practice.

The condition that (\ref{eq.lb}) and (\ref{eq.ub}) possess finite solutions is mild due to the following key property. Define the value functions of the inner optimization problems in (\ref{eq.lb}) and (\ref{eq.ub}) as $G_L$ and $G_U$, respectively:
\begin{displaymath}
    \begin{aligned}
        G_L(\lambda, w) &= \min_{b \in \mc{B}} \left\{ m(w, b) + \sum_{k=1}^{K}\lambda_k \phi_k(w, b) \right\}, \\
        G_U(\lambda, w) &= \max_{b \in \mc{B}} \left\{ m(w, b) + \sum_{k=1}^{K}\lambda_k \phi_k(w, b) \right\}.
    \end{aligned}
\end{displaymath}
Note that, given the model ingredients $m$ and $\phi_1, \ldots, \phi_K$, these are deterministic functions of $(\lambda,w)$. They have the following key property.

\begin{proposition}
    \label{prop.convex}
    $G_L(\lambda, w)$ is globally concave in $\lambda$ for every $w$, and $G_U(\lambda, w)$ is globally convex in $\lambda$ for every $w$.
\end{proposition}

Since concave and convex functions on $\mb{R}^K$ are continuous, (\ref{eq.lb}) and (\ref{eq.ub}) possess finite solutions whenever the optimizers in $\lambda$ lie in the interior of $\mb{R}^K$. For dynamic random coefficient models in (\ref{eq.crc}) and (\ref{eq.meanindep}), this can be achieved by a suitable choice of the moment functions $(\phi_1, \ldots, \phi_K)$ derived from (\ref{eq.meanindep}), which I illustrate in the next subsection for the variance and the CDF parameters.

Using the definitions of $G_L$ and $G_U$, the bounds in \Cref{prop.gmm} can be written as
\begin{displaymath}
    L = \max_{\lambda \in \mb{R}^K}
    \E\left[ G_L(\lambda, W_i) \right], \qand
    U = \min_{\lambda \in \mb{R}^K}
    \E\left[ G_U(\lambda, W_i) \right].
\end{displaymath}
Under suitable conditions, $G_L$ and $G_U$ are differentiable when $K=1$ \citep[Theorem 3]{milgrom2002}, which can be extended to show that $G_L$ and $G_U$ are directionally differentiable for $K > 1$. \Cref{prop.convex} then implies that the optimization problems over $\lambda$ can be solved using fast convex optimization algorithms such as gradient descent, provided that the inner optimization problems over $b$ can be solved efficiently. In the next subsection, I illustrate the choice of the moment functions for the variance and the CDF parameters that admits computationally efficient solutions to the inner optimization problems.

A direct consequence of \Cref{prop.gmm} is that $\theta$ is point-identified if and only if $L = U$. Proof of \Cref{prop.gmm} then implies a necessary and sufficient condition for point identification of $\theta$, which I state as a separate lemma below.

\begin{lemma}
    \label{lemma.gmm}
    Suppose that the assumptions of \Cref{prop.gmm} hold. Suppose also that $(W_i, B_i)$ are absolutely continuous with respect to the Lebesgue measure, and that their joint density is strictly positive on $\mc{W} \times \mc{B}$. Then $\theta$ is point-identified if and only if there exists a function $S^*$, which is a linear functional on $\mc{M}_{W}$, and real numbers $\lambda_1^*, \ldots, \lambda_{K}^* \in \mb{R}$ such that:
    \begin{displaymath}
        m(W_i, B_i) + \sum_{k=1}^{K}\lambda_k^* \phi_k(W_i, B_i) = S^*(W_i)
    \end{displaymath}
    almost surely on $\mc{W} \times \mc{B}$. When such $S^*$ exists, $\theta$ is identified by $\theta = \E(S^*(W_i))$.
\end{lemma}

\Cref{lemma.gmm} states that $\theta$ is point-identified if and only if the Lagrangian reduces to a function of data only. Note that $S^*$ can be considered as an unbiased estimator because the term $\sum_{k=1}^{K}\lambda_k^* \phi_k(W_i, B_i)$ has zero expectation\footnote{Whether there exist point-identified parameters in (\ref{eq.crc}) and (\ref{eq.meanindep}) that are similar to the examples in \citet{chamberlain1993,chamberlain2022} and \citet{bonhomme2025b} remains an open question and is not pursued here.}.

Lastly, I highlight the connection between \Cref{prop.gmm} and the support function approach of \citet{beresteanu2011}. Let $\delta$ be a structural parameter and consider the moment conditions
\begin{displaymath}
    \E(\phi_k(W_i, B_i, \delta)) = 0, \quad k=1, \ldots, K.
\end{displaymath}
In what follows, I fix the value of $\delta$ and consider each $\phi_k(\cdot,\cdot,\delta)$ as a function of $(W_i, B_i)$ only. In addition, I set $m(W_i, B_i) = 0$, so that $\theta = \E(m(W_i, B_i)) = 0$. In this case, the sharp lower bound of $\theta = 0$ is obtained by specializing (\ref{eq.primal}) with $m = 0$:
\begin{displaymath}
    \begin{aligned}
        L_{primal}(\delta) = \min_{P \in \mc{M}_{W \times B}, ~P \geq 0} 0 \st 
        & \int \phi_k(w, b, \delta) dP = 0, \quad k=1, \ldots, K, \\
        & \int P(w, db) = P_W(w) ~\text{ for all } w \in \mc{W}.
    \end{aligned}
\end{displaymath}
The solution of this problem is trivially $0$, but only if there exists a probability distribution $P$ that satisfies the moment conditions. If no such $P$ exists, the problem is infeasible, and I set $L_{primal}(\delta) = \infty$. This characterization is similar in spirit to those in \citet{honore2006}, \citet{honore2006b} and \citet{molinari2008}, extended here to allow $P$ to be a continuous distribution. I can then write the identified set of $\delta$ as
\begin{displaymath}
    \{\delta ~|~ L_{primal}(\delta) = 0\}.
\end{displaymath}
\Cref{prop.gmm} then implies that the dual representation of $L_{primal}(\delta)$ is
\begin{equation}
    L_{dual}(\delta) = \max_{\lambda \in \mb{R}^K}
    \E\left[ \min_{b \in \mc{B}} \left\{ \sum_{k=1}^{K}\lambda_k \phi_k(W_i, b, \delta) \right\} \right].
    \label{eq.lb.support}
\end{equation}
Using this, the identified set of $\delta$ can also be written as $\{\delta ~|~ L_{dual}(\delta) = 0\}$. This characterization coincides with the support function characterization in \citet[Section 4]{beresteanu2011} given for regression coefficients with interval data. In particular, in their Theorem 4.1, the negative of their support function coincides with the inner objective function $\sum_{k=1}^{K}\lambda_k \phi_k(W_i, b, \delta)$ in (\ref{eq.lb.support}), and their variable $u$ coincides with the Lagrange multiplier $\lambda$ in (\ref{eq.lb.support}).

\subsection{Examples: the variance and the CDF of random coefficients}

\label{sec.general.varcdf}

In this subsection, I apply \Cref{prop.gmm} to derive the bounds for the variance and the CDF parameters. Recall the dynamic random coefficient model defined in (\ref{eq.crc}) and (\ref{eq.meanindep}):
\begin{displaymath}
    Y_{it} = R_{it}'B_i + \veps_{it}, \qquad \E(\veps_{it}|B_i, Z_i, X_i^t) = 0, \qquad t=1, \ldots, T,
\end{displaymath}
where $R_{it} = (Z_{it}', X_{it}')'$. I first consider the second moments of the random coefficients:
\begin{displaymath}
    V_e = \E(e_1'B_iB_i'e_2) = \E(B_i'e_1e_2'B_i),
\end{displaymath}
where $e_1$ and $e_2$ are real-valued constant vectors chosen by the econometrician. %

A full characterization of the bounds of $V_e$ for generic choices of $e_1$ and $e_2$ is discussed in Online Appendix \Cref{sec.appendix.var}. Here, I focus on the special case where $e_1 = e_2$ and this common vector has a single entry equal to $1$ and zeros elsewhere. In this case, $V_e$ is the second moment of a particular coefficient --- a key ingredient of the variance parameter. This particular case deserves separate discussion because its bounds can be computed more efficiently than those in the general case.

In the case where $e_1 = e_2$ and this common vector has a single entry equal to $1$ and zeros elsewhere, define $e_0 = e_1e_1'$. Then $e_0$ is a diagonal matrix which has $1$ in only one entry and zeros elsewhere. For example, if $B_i = (\beta_{i1}, \beta_{i2})'$ and $e_1 = e_2 = (0,1)'$, then
\begin{displaymath}
    e_0 = \left(\begin{array}{cc}
        0 & 0 \\
        0 & 1
    \end{array}\right), \qand
    V_e = \E(B_i'e_0B_i) = \E(\beta_{i2}^2). 
\end{displaymath}
Recall the moment restrictions used for the bounds in \Cref{prop.mean.closedform.refined}, namely those in (\ref{eq.orthogonal.closedform.refined}):
\begin{displaymath}
    \E\left(\sum_{t=1}^T (R_{it}'B_i)\veps_{it}\right) = 0, \qand \E\left(S_{it}\veps_{it}\right) = 0 \quad\text{for}\quad t=1, \ldots, T.
\end{displaymath}
Let $L = \sum_{t=1}^T \text{dim}(S_{it})$. Applying \Cref{prop.gmm} with these restrictions yields the following lower bound for $V_e$, denoted by $L_V$:
\begin{displaymath}
    \begin{aligned}
        L_V &= \max_{\lambda \in \mb{R}, ~\mu \in \mb{R}^L} 
        \E\left[ \min_{b \in \mc{B}} \left\{ b'e_0b + \lambda b'R_i'(Y_i-R_ib) + \mu'S_i(Y_i - R_ib) \right\} \right] \\
        &= \max_{\lambda \in \mb{R}, ~\mu \in \mb{R}^L} 
        \E\left[ \min_{b \in \mc{B}} \left\{ \mu'S_iY_i + (\lambda R_i'Y_i - R_i'S_i'\mu)'b - b'(\lambda R_i'R_i - e_0)b \right\} \right], \\
    \end{aligned}
\end{displaymath}

Suppose that \Cref{ass.crc,ass.mean.nomulticollinearity,ass.mean.nomulticollinearity.s} hold. Note that, in $L_V$, the objective function of the inner minimization problem is a quadratic polynomial in $b$, where the leading coefficient matrix is $-(\lambda R_i'R_i - e_0) = e_0 - \lambda R_i'R_i$. Since $e_0$ is positive semidefinite, and since $R_i'R_i$ is positive definite by \Cref{ass.mean.nomulticollinearity}, %
a finite lower bound is obtained only when $\lambda < 0$. I then obtain the following expression for the lower bound, which is a direct application of \Cref{prop.gmm} and is stated without proof.

\begin{proposition}
    \label{prop.variance.lb}
    Suppose that \Cref{ass.crc,ass.mean.nomulticollinearity,ass.mean.nomulticollinearity.s,ass.proof.regularity} hold. Then $L_V \leq V_e$ where
    \begin{displaymath}
        \begin{aligned}
            L_V &= \max_{\lambda <0, ~\mu \in \mb{R}^L} 
            \E\left[ \min_{b \in \mc{B}} \left\{ \mu'S_iY_i + (\lambda R_i'Y_i - R_i'S_i'\mu)'b - b'(\lambda R_i'R_i - e_0)b \right\} \right]. \\
        \end{aligned}
    \end{displaymath}
\end{proposition}

Note that the objective function in \Cref{prop.variance.lb} is concave in $(\lambda, \mu)$ by \Cref{prop.convex}. Therefore, the maximization over $(\lambda,\mu)$ can be performed efficiently using standard convex optimization methods. The inner minimization over $b$ can also be solved efficiently using quadratic optimization softwares.

An upper bound of $V_e$ can be obtained similarly, but with a stronger assumption. Specifically, by applying \Cref{prop.gmm}, I obtain the following upper bound of $V_e$:
\begin{displaymath}
    \begin{aligned}
        U_V &= \min_{\lambda \in \mb{R}, ~\mu \in \mb{R}^L} 
        \E\left[ \max_{b \in \mc{B}} \left\{ \mu'S_iY_i + (\lambda R_i'Y_i - R_i'S_i'\mu)'b - b'(\lambda R_i'R_i - e_0)b \right\} \right]. \\
    \end{aligned}
\end{displaymath}
The inner objective function is a quadratic polynomial in $b$, where the leading coefficient matrix is $-(\lambda R_i'R_i - e_0) = e_0 - \lambda R_i'R_i$. %
A finite upper bound is obtained only in the region where $e_0 - \lambda R_i'R_i$ is negative definite, i.e., all of its eigenvalues are negative. %
Then %
Weyl's inequality implies that the largest eigenvalue of $e_0 - \lambda R_i'R_i$ is bounded above by $1 - \lambda \nu$, where $\nu>0$ is the smallest eigenvalue of $R_i'R_i$. Therefore, all eigenvalues of $e_0 - \lambda R_i'R_i$ are negative if %
$\lambda > 1/\nu$.
Under this additional condition, I obtain the following upper bound of $V_e$, which is a direct application of \Cref{prop.gmm} and is stated without proof.

\begin{assumption}
    \label{ass.var.nomulticollinearity}
    There exists $\lambda_{min} > 0$ such that the smallest eigenvalue of $R_i'R_i$ is strictly larger than $1/\lambda_{min}$ almost surely.
\end{assumption}

\begin{proposition}
    \label{prop.variance.ub}
    Suppose that \Cref{ass.crc,ass.mean.nomulticollinearity,ass.mean.nomulticollinearity.s,ass.proof.regularity,ass.var.nomulticollinearity} hold. Then $V_e \leq U_V$ where
    \begin{displaymath}
        \begin{aligned}
            U_V &= \min_{\lambda \geq \lambda_{min}, ~\mu \in \mb{R}^L} 
            \E\left[ \max_{b \in \mc{B}} \left\{ \mu'S_iY_i + (\lambda R_i'Y_i - R_i'S_i'\mu)'b - b'(\lambda R_i'R_i - e_0)b \right\} \right].
        \end{aligned}
    \end{displaymath}
\end{proposition}

Note that \Cref{ass.var.nomulticollinearity} is stronger than \Cref{ass.mean.nomulticollinearity}. While \Cref{ass.mean.nomulticollinearity} requires the smallest eigenvalue of $R_i'R_i$ to be positive almost surely, \Cref{ass.var.nomulticollinearity} further requires that it is strictly bounded away from zero\footnote{Analogously to the discussion on \Cref{ass.mean.nomulticollinearity} in Footnote \ref{footnote.trimming}, if \Cref{ass.var.nomulticollinearity} is violated because of $Z_i$, one may choose to consider the subpopulation where $\text{det}(Z_i'Z_i) \geq d_0$ for some $d_0 > 1/\lambda_{min}$.}. I also derive the bounds for $V_e$ that do not require \Cref{ass.var.nomulticollinearity} in Online Appendix \ref{sec.appendix.var}, but their estimation will involve computationally more intensive methods.

Next, I consider the CDF of the random coefficients. Consider the parameter of the form
\begin{displaymath}
    F_{e,c} = \pp(e'B_i \leq c) = \E(\textbf{1}(e'B_i \leq c))
\end{displaymath}
where $e$ is a real-valued constant vector and $c$ is a scalar. To derive the identified set of $F_{e,c}$, I consider the moment restrictions used for the bounds in \Cref{prop.mean.closedform.refined}, namely those in (\ref{eq.orthogonal.closedform.refined}). Applying \Cref{prop.gmm} with (\ref{eq.orthogonal.closedform.refined}) yields the bounds with the inner objective function
\begin{displaymath}
    \mc{L} = \textbf{1}(e'B_i \leq c) + \lambda B_i'R_i'(Y_i-R_iB_i) + \mu'S_i(Y_i-R_iB_i)
\end{displaymath}
that must be optimized over $B_i$ for fixed $(\lambda,\mu)$. Note that the indicator $\textbf{1}(e'B_i \leq c)$ partitions the support of $B_i$ into two disjoint sets, where it equals to $1$ on the set $\{B_i | e'B_i \leq c\}$ and $0$ on the set $\{e'B_i > c\}$. Moreover, within each set, $\mc{L}$ reduces to a standard quadratic polynomial in $B_i$, which can be solved efficiently. Therefore, the optimization of $\mc{L}$ over $B_i$ for fixed $(\lambda,\mu)$ can be carried out in two steps: (i) solve for the quadratic polynomial within each set, and then (ii) take the optimum between the two. 
I then obtain the following characterization for the identified set of $F_{e,c}$, which is a direct application of \Cref{prop.gmm} and is stated without proof.

\begin{proposition}
    \label{prop.cdf}
    Suppose that \Cref{ass.crc,ass.mean.nomulticollinearity,ass.mean.nomulticollinearity.s,ass.proof.regularity} hold. 
    Then $L_F \leq F_{e,b} \leq U_F$ where
    \begin{displaymath}
        L_F = \max_{\lambda < 0, ~\mu \in \mb{R}^L}
        \E\left[ G_{L,F}(W_i, \lambda, \mu) \right],
        \qand
        U_F = \min_{\lambda > 0, ~\mu \in \mb{R}^L}
        \E\left[ G_{U,F}(W_i, \lambda, \mu) \right],
    \end{displaymath}
    where
    \begin{displaymath}
        \begin{aligned}
            G_{L,F}(W_i, \lambda, \mu) = \min\Big\{ & \min_{b ~\in~ \{b \in \mc{B}~|~e'b \leq c\}} \Big[1 + \lambda b'R_i'(Y_i-R_ib) + \mu'S_i(Y_i-R_ib)\Big], \\
            & \min_{b ~\in~ \{b \in \mc{B}~|~e'b > c\}} \Big[\lambda b'R_i'(Y_i-R_ib) + \mu'S_i(Y_i-R_ib)\Big] \Big\},
        \end{aligned}
    \end{displaymath}
    and
    \begin{displaymath}
        \begin{aligned}
            G_{U,F}(W_i, \lambda, \mu) = \max\Big\{ & \max_{b ~\in~ \{b \in \mc{B}~|~e'b \leq c\}} 1 + \lambda b'R_i'(Y_i-R_ib) + \mu'S_i(Y_i-R_ib), \\
            & \max_{b ~\in~ \{b \in \mc{B}~|~e'b > c\}} \Big[\lambda b'R_i'(Y_i-R_ib) + \mu'S_i(Y_i-R_ib)\Big] \Big\}.
        \end{aligned}
    \end{displaymath}
\end{proposition}

If $B_i$ is continuous, then one may replace the set $\{b \in \mc{B}~|~e'b > c\}$ with its closure $\{b \in \mc{B}~|~e'b \geq c\}$, which facilitates estimation and inference.

\section{Estimation and inference}

\label{sec.estimation}

This section discusses estimation and inference for the identified sets derived in \Cref{sec.mean,sec.general}. This section is structured into two subsections. In the first, I consider inference for the mean parameters discussed in \Cref{sec.mean}. I exploit their simple closed-form expressions to present a procedure that is both straightforward to implement and robust to overidentification and model misspecification. In the second, I consider inference for the general parameters discussed in \Cref{sec.general}, presenting a procedure under the assumption of correct model specification.

\subsection{Estimation and inference for the mean parameters}

\label{sec.inference.mean}

In this subsection, I discuss estimation and inference for the mean parameters, focusing on the refined bounds in \Cref{prop.mean.closedform.refined,prop.mean.closedform.refined.constant}. In what follows, I present a procedure for the bounds in \Cref{prop.mean.closedform.refined}. The same procedure applies to the bounds in \Cref{prop.mean.closedform.refined.constant}.

Note that the bounds $[L_S, U_S]$ in \Cref{prop.mean.closedform.refined} are deterministic functions of the following moments:
\begin{displaymath}
    \begin{aligned}
        V_0 &= \mc{V}_S = \E(S_iR_i(R_i'R_i)^{-1}R_i'S_i'), \\
        Y_0 &= 2Y_S - \mc{Y}_S = \E(2S_iY_i - S_iR_i(R_i'R_i)^{-1}R_i'Y_i), \\
        P_0 &= \mc{P}_Se = \E(S_iR_i(R_i'R_i)^{-1}e), \\
        m_0 &= \E(Y_i'R_i(R_i'R_i)^{-1}R_i'Y_i), \\
        b_0 &= \E(e'\widehat{B}_i) = \E(e'(R_i'R_i)^{-1}R_i'Y_i), \\
        R_0 &= \E((R_i'R_i)^{-1}).
    \end{aligned}
\end{displaymath}
Let $D_i$ be the vector that collects all of the entries inside these expectations. In other words, $D_i$ is defined as
\begin{displaymath}
    \begin{aligned}
        D_i = \Big( &\text{vech}(S_iR_i(R_i'R_i)^{-1}R_i'S_i')', ~(2S_iY_i - S_iR_i(R_i'R_i)^{-1}R_i'Y_i)', ~(S_iR_i(R_i'R_i)^{-1}e)', \\
        &\qquad Y_i'R_i(R_i'R_i)^{-1}R_i'Y_i, ~e'(R_i'R_i)^{-1}R_i'Y_i, ~\text{vech}((R_i'R_i)^{-1})'\Big)'.
    \end{aligned}
\end{displaymath}
Note that $\E(D_i) = (\text{vech}(V_0)', Y_0', P_0', m_0, b_0, \text{vech}(R_0)')'$. Now, given an independent and identically distributed (i.i.d.) sample $\{D_i\}_{i=1}^N$ of size $N$, define
\begin{displaymath}
    \overline{D}_N = \frac{1}{N}\sum_{i=1}^N D_i.
\end{displaymath}
I assume that $\overline{D}_N$ is asymptotically normal with rate $\sqrt{N}$, which holds if the conditions for multivariate Central Limit Theorem hold for $D_i$ \citep[Section 2]{van2000}. %

\begin{assumption}
    \label{ass.mean.asympnormal}
    $\sqrt{N}(\overline{D}_N - \E(D_i))$ converges in distribution to $N(0, V_D)$ for some variance matrix $V_D$.
\end{assumption}

Now I discuss estimation and inference for $[L_S, U_S]$ under assumptions of \Cref{prop.mean.closedform.refined} and \Cref{ass.mean.asympnormal}. %
Recall that the expressions for $[L_S, U_S]$ are:
\begin{displaymath}
    [L_S, U_S] = \left[ \mathcal{B}_S - \frac{1}{2}\sqrt{\mc{E}_S\mc{D}_S}, ~\mathcal{B}_S + \frac{1}{2}\sqrt{\mc{E}_S\mc{D}_S} \right].
\end{displaymath}
Let $\widehat{\mc{B}}_S$, $\widehat{\mc{E}}_S$, and $\widehat{\mc{D}}_S$ be the sample counterparts of $\mathcal{B}_S$, $\mc{E}_S$, and $\mc{D}_S$ calculated with $\overline{D}_N$. For example, $\widehat{\mc{B}}_S$ is given by
\begin{displaymath}
    \begin{aligned}
        \widehat{\mc{B}}_S &= \frac{1}{2}e'\left(\frac{1}{N}\sum_{i=1}^N(R_i'R_i)^{-1}R_i'Y_i\right) + \frac{1}{2}e'\left(\frac{1}{N}\sum_{i=1}^N(R_i'R_i)^{-1}R_i'S_i'\right) \times \\
        &\qquad\qquad\qquad\qquad\left(\frac{1}{N}\sum_{i=1}^NS_iR_i(R_i'R_i)^{-1}R_i'S_i'\right)^{-1} \left(\frac{1}{N}\sum_{i=1}^N (2S_iY_i - S_iR_i(R_i'R_i)^{-1}R_i'Y_i)\right).
    \end{aligned}
\end{displaymath}
Then, define an estimator of $[L_S, U_S]$ as
\begin{displaymath}
    [\widehat{L}_S, \widehat{U}_S] = \left[ \widehat{\mc{B}}_S - \frac{1}{2}\sqrt{\widehat{\mc{E}}_S\widehat{\mc{D}}_S}, ~\widehat{\mc{B}}_S + \frac{1}{2}\sqrt{\widehat{\mc{E}}_S\widehat{\mc{D}}_S} \right].
\end{displaymath}

Since $[L_S, U_S]$ is a smooth function of $\E(D_i)$ provided that $\mathcal{E}_S > 0$ and $\mathcal{D}_S > 0$, the Delta method \citep[Section 3]{van2000} implies that $[\widehat{L}_S, \widehat{U}_S]$ is asymptotically normal. A key practical issue, however, is that the quantity $\widehat{\mc{D}}_S$ may be negative in finite samples, causing the term $\sqrt{\widehat{\mc{E}}_S\widehat{\mc{D}}_S}$ and thus the estimator $[\widehat{L}_S, \widehat{U}_S]$ to be not well-defined. This issue is related to a well-known challenge in inference for partially identified models --- overidentification and model misspecification. In what follows, I discuss this issue in detail and propose an inference procedure that addresses it. 

Recall that the bounds $[L_S, U_S]$ arise as the dual representations of the primal problem in (\ref{eq.primal}) (and the corresponding problem for the upper bound) based on the moment restrictions in (\ref{eq.orthogonal.closedform.refined}). It can be shown that the estimated bounds $[\widehat{L}_S, \widehat{U}_S]$ are the dual of the sample version of (\ref{eq.primal}) where the population distribution $P_W$ is replaced with the finite-sample empirical distribution $\hat P_W$. Overidentification then arises when the population problem (\ref{eq.primal}) is feasible but its sample version with $\hat P_W$ is infeasible. This mirrors the familiar overidentification problem in generalized method of moments (GMM) estimation where, even if the population satisfies all the moment restrictions so that the population GMM criterion achieves zero, the finite sample may not satisfy all moment restrictions simultaneously, yielding a strictly positive sample GMM criterion\footnote{For the empirical likelihood approach, this translates into the empirical likelihood criterion failing to attain its optimum at equal probabilities.}. 
In terms of the closed-form expressions of \Cref{prop.mean.closedform.refined}, this corresponds to having $\mathcal{D}_S > 0$ but $\widehat{\mc{D}}_S < 0$. 
In contrast, model misspecification arises when 
the population problem (\ref{eq.primal}) itself is infeasible.
In this case, the population quantity $\mathcal{D}_S$ is negative, and thus its sample counterpart $\widehat{\mc{D}}_S$ is also likely to be negative.

A recently growing literature on inference under misspecification in partially identified models \citep{stoye2020,andrews2024} propose solutions to these issues. In what follows, I adopt the procedure of \citet{stoye2020} who provides a simple, easy-to-implement method for conducting inference on bounds that are smooth functions of the moments.

To apply the procedure of \citet{stoye2020}, I first construct a smooth approximation and extension of the bounds $[L_S, U_S]$ that remains well-defined for any values of $\mc{B}_S$, $\mc{E}_S$, and $\mc{D}_S$. Let $r > 0$ be a small constant, and define the smoothed square root function
\begin{displaymath}
    s(x,y) = \sqrt{\frac{xy + \sqrt{(xy)^2 + r^2}}{2}}.
\end{displaymath}
For small $r > 0$, this function satisfies $s(x,y) = \sqrt{xy} + O(r)$ if $xy>0$, and $s(x,y) = O(r)$ if $xy < 0$. In other words, $s(x,y)$ coincides with the ordinary square root when $xy > 0$ and vanishes if $xy < 0$. In addition, because of the term $r^2 >0$, it is smooth everywhere, including at $xy=0$. I then define the smooth approximation and extension of $[L_S, U_S]$ as:
\begin{displaymath}
    [L_{Smth}, U_{Smth}] = \left[ \mc{B}_S - \frac{1}{2}\left(s(\mc{E}_S,\mc{D}_S) - s(\mc{E}_S,-\mc{D}_S)\right), ~\mc{B}_S + \frac{1}{2}\left(s(\mc{E}_S,\mc{D}_S) - s(\mc{E}_S,-\mc{D}_S)\right) \right].
\end{displaymath}
Note that $\mc{E}_S > 0$ even under overidentification and misspecification because
\begin{displaymath}
    \begin{aligned}
        \mc{E}_S &= e'\E((R_i'R_i)^{-1})e-e'\mc{P}_S'\mc{V}_S^{-1}\mc{P}_Se \\
        &= \E\left( (e'-e'\E((R_i'R_i)^{-1}R_i'S_i')S_iR_i)(R_i'R_i)^{-1}(e-R_i'S_i'\E(S_iR_i(R_i'R_i)^{-1})e) \right),
    \end{aligned}
\end{displaymath}
which is a quadratic form associated with a positive definite matrix. Therefore, $\mc{D}_S$ is the only quantity that can become negative in the square root function. Then, if $\mc{D}_S > 0$, $[L_{Smth}, U_{Smth}]$ simplifies to
\begin{displaymath}
    [L_{Smth}, U_{Smth}] \approx \left[ \mc{B}_S - \frac{1}{2}s(\mc{E}_S,\mc{D}_S), ~\mc{B}_S + \frac{1}{2}s(\mc{E}_S,\mc{D}_S) \right],
\end{displaymath}
which coincides with $[L_S, U_S]$ up to the error term $O(r)$. If $\mc{D}_S < 0$, then
\begin{displaymath}
    [L_{Smth}, U_{Smth}] \approx \left[ \mc{B}_S + \frac{1}{2}s(\mc{E}_S,-\mc{D}_S), ~\mc{B}_S - \frac{1}{2}s(\mc{E}_S,-\mc{D}_S) \right]
\end{displaymath}
so that $L_{Smth} > U_{Smth}$, indicating that the estimated bound is empty.

Now I discuss inference for $[L_{Smth}, U_{Smth}]$. Define the estimator of $[L_{Smth}, U_{Smth}]$ as
\begin{displaymath}
    [\widehat{L}_{Smth}, \widehat{U}_{Smth}] = \left[ \widehat{\mc{B}}_S - \frac{1}{2}\left(s(\widehat{\mc{E}}_S,\widehat{\mc{D}}_S) - s(\widehat{\mc{E}}_S,-\widehat{\mc{D}}_S)\right), ~\widehat{\mc{B}}_S + \frac{1}{2}\left(s(\widehat{\mc{E}}_S,\widehat{\mc{D}}_S) - s(\widehat{\mc{E}}_S,-\widehat{\mc{D}}_S)\right) \right].
\end{displaymath}
\Cref{ass.mean.asympnormal} and the Delta method \citep[Section 3]{van2000} then imply that $(\widehat{L}_{Smth}, \widehat{U}_{Smth})$ is asymptotically normal:
\begin{displaymath}
    \sqrt{N}((\widehat{L}_{Smth}, \widehat{U}_{Smth})' - (L_{Smth}, U_{Smth})') 
    \overset{d}{\longrightarrow} 
    N\left(0, \left[\begin{array}{cc}
        \sigma_L^2 & \rho \sigma_L \sigma_U \\
        \rho \sigma_L \sigma_U & \sigma_U^2
    \end{array}\right]\right)
\end{displaymath}
for some $\sigma_L$, $\sigma_U$, and $\rho$. This verifies Assumption 1 of \citet{stoye2020}. Note that $\sigma_L$, $\sigma_U$, and $\rho$ can be consistently estimated by bootstrap \citep[Section 23]{van2000}.

Next, define the pseudo-true parameter \citep{stoye2020,andrews2024}:
\begin{displaymath}
    \mu_e^* = \frac{\sigma_UL_{Smth} + \sigma_LU_{Smth}}{\sigma_L + \sigma_U}.
\end{displaymath}
Note that $\mu_e^*$ is well-defined even if $\mc{D}_S < 0$ and that $\mu_e^* \approx \mc{B}_S$. Define its estimator as
\begin{displaymath}
    \widehat{\mu}_e^* = \frac{\widehat{\sigma}_U\widehat{L}_{Smth} + \widehat{\sigma}_L\widehat{U}_{Smth}}{\widehat{\sigma}_L + \widehat{\sigma}_U},
\end{displaymath}
where $\widehat{\sigma}_L$, $\widehat{\sigma}_U$, and $\widehat{\rho}$ are consistent estimators of $\sigma_L$, $\sigma_U$, and $\rho$, respectively. Both $\mu_e^*$ and $\widehat{\mu}_e^*$ are well-defined under overidentification or misspecification.

Then, the $(1-\alpha)$-level confidence interval for $\mu_e$ is constructed as follows. First, consider an interval for $\mu_e$ based on the smoothed bounds $[\widehat{L}_{Smth}, \widehat{U}_{Smth}]$:
\begin{displaymath}
    I_{\mu_e} = \left[\widehat{L}_{Smth} - \widehat{c}(\alpha)\frac{\widehat{\sigma}_L}{\sqrt{N}}, ~~\widehat{U}_{Smth} + \widehat{c}(\alpha)\frac{\widehat{\sigma}_U}{\sqrt{N}}\right],
\end{displaymath}
where $\widehat{c}(\alpha)$ is the critical value specified in Table 1 of \citet{stoye2020}. For instance, if $\alpha = 0.05$, then $\widehat{c}(0.05) = 1.64$ if $\widehat{\rho} < 0.8$, and $\widehat{c}(0.05) = 1.96$ if $\widehat{\rho} \approx 1$. Note that $I_{\mu_e}$ may be empty under overidentification or misspecification. Second, consider an interval for the pseudo-true parameter $\mu_e^*$:
\begin{displaymath}
    I_{\mu_e^*} = \left[\widehat{\mu}_e^*
    - \Phi\left(1-\frac{\alpha}{2}\right)\frac{\widehat{\sigma}^*}{\sqrt{N}}, 
    ~~\widehat{\mu}_e^*
    + \Phi\left(1-\frac{\alpha}{2}\right)\frac{\widehat{\sigma}^*}{\sqrt{N}}\right]
\end{displaymath}
where $\Phi$ is the standard normal CDF and
\begin{displaymath}
    \widehat{\sigma}^* = \frac{\widehat{\sigma}_L\widehat{\sigma}_U\sqrt{2+2\widehat{\rho}}}{\widehat{\sigma}_L + \widehat{\sigma}_U}.
\end{displaymath}
Then, the $(1-\alpha)$-level confidence interval for $\mu_e$ that is valid under overidentification and misspecification is given by
\begin{displaymath}
    CI_{\mu_e} = I_{\mu_e} \cup I_{\mu_e^*}.
\end{displaymath}
Theorem 1 of \citet{stoye2020} establishes the validity of $CI_{\mu_e}$. Under overidentification, $CI_{\mu_e}$ asymptotically achieves the $(1-\alpha)$ coverage rate for the true parameter $\mu_e$, where overidentification is resolved as $N \rightarrow \infty$. Under misspecification, $CI_{\mu_e}$ asymptotically achieves the coverage rate for the pseudo-true parameter $\mu_e^*$.

\subsection{Inference for the general parameters}

\label{sec.inference.general}

I now discuss construction of a confidence interval for a general parameter $\theta$ in \Cref{sec.general}. By \Cref{prop.gmm}, the bounds $[L, U]$ of $\theta$ are given by
\begin{displaymath}
    L = \max_{\lambda \in \mb{R}^K}\E(G_L(\lambda, W_i)), \qand U = \min_{\lambda \in \mb{R}^K}\E(G_U(\lambda, W_i)).
\end{displaymath}
Note first that any $\theta \in [L, U]$ must satisfy
\begin{displaymath}
    \begin{aligned}
        \theta \geq L &= \max_{\lambda \in \mb{R}^{K}}\E(G_L(\lambda, W_i)), \\
        \theta \leq U &= \min_{\lambda \in \mb{R}^{K}}\E(G_U(\lambda, W_i)).
    \end{aligned}
\end{displaymath}
For regularity of the inference procedure, consider a large compact set $R^K \subseteq \mb{R}^K$ and consider the inequalities
\begin{displaymath}
    \begin{aligned}
        \theta \geq \tilde L &= \max_{\lambda \in R^K}\E(G_L(\lambda, W_i)), \\
        \theta \leq \tilde U &= \min_{\lambda \in R^K}\E(G_U(\lambda, W_i)).
    \end{aligned}
\end{displaymath}
I choose the set $R^K$ to be large enough so that both $\lambda_0^L = \argmax_\lambda \E(G_L(\lambda, W_i))$ and $\lambda_0^U = \argmin_\lambda \E(G_U(\lambda, W_i))$ lie in the interior of $R^K$, in which case $[L, U] = [\tilde L, \tilde U]$. Otherwise, $[\tilde L, \tilde U]$ becomes an outer identified set of $[L, U]$. I then rewrite the above as
\begin{displaymath}
    \begin{aligned}
        \theta \geq \E(G_L(\lambda, W_i)) \quad \text{ for all } \lambda \in R^K, \\
        \theta \leq \E(G_U(\lambda, W_i)) \quad \text{ for all } \lambda \in R^K.
    \end{aligned}
\end{displaymath}
Equivalently, these can be written as the following moment inequalities:
\begin{equation}
    \begin{aligned}
        \E(G_L(\lambda, W_i) - \theta) \leq 0 \quad \text{ for all } \lambda \in R^K, \\
        \E(\theta - G_U(\lambda, W_i)) \leq 0 \quad \text{ for all } \lambda \in R^K,
    \end{aligned}
    \label{eq.inference}
\end{equation}
which is a moment inequalities model with infinitely many restrictions indexed by $\lambda$.

The literature on many moment inequalities \citep{romano2014,andrews2017,chernozhukov2019,bai2022} develops procedures for constructing a confidence interval for $\theta$. In this paper, I adopt the inference procedure of \citet{andrews2017} on a continuum of moment inequalities, which includes countably many moment inequalities as a special case. Note first that $G_L$ is concave and $G_U$ is convex in $\lambda$ by \Cref{prop.convex}, which implies that both functions are continuous in $\lambda$. This means that, for inference on $\theta$, it suffices to consider:
\begin{equation}
    \begin{aligned}
        \E(G_L(\lambda, W_i) - \theta) \leq 0 \quad \text{ for all } \lambda \in Q^K, \\
        \E(\theta - G_U(\lambda, W_i)) \leq 0 \quad \text{ for all } \lambda \in Q^K,
    \end{aligned}
    \label{eq.inference.countable}
\end{equation}
where $Q^K \subseteq R^K$ is a set of rational numbers in $R^K$, which is dense in $R^K$. Section 9.2 of \citet{andrews2017} develops an inference procedure for this countable set of moment inequalities. The conditions for the validity of their procedure are given in Lemma 9.2 of \citet{andrews2017}. Although this lemma is given for a single moment restriction with one-dimensional $\lambda$, its extension to two moment restrictions and to a $K$-dimensional $\lambda$ is straightforward. 
In what follows, I assume that the conditions of their Lemma 9.2 are satisfied, which are mild given \Cref{ass.proof.regularity} and the fact that $Q^K$ is compact\footnote{Lemma 9.2 of \citet{andrews2017} introduces a weight function associated with an ordering of the moment inequalities. %
This weight function does not affect the inference procedure, since it cancels out in the construction of the test statistics and therefore does not appear in any of the expressions. %
}.

\begin{assumption}
    \label{ass.inference.as}
    There is $\underline{\sigma} > 0$ such that $\var(G_L(\lambda_0, W_i)) \geq \underline{\sigma}^2$ and $\var(G_U(\lambda_0, W_i)) \geq \underline{\sigma}^2$ for some fixed $\lambda_0 \in Q^K$. Also, there is a measurable function $g$ such that $|G_L(\lambda_0, W_i)| \leq g(W_i)$, $|G_U(\lambda_0, W_i)| \leq g(W_i)$, and $\E((g(W_i)/\underline{\sigma})^{2+r}) \leq C$ for some $r > 0$ and $C < \infty$.
\end{assumption}

In what follows, I apply their inference method under assumptions of \Cref{prop.gmm} and \Cref{ass.inference.as}, where I choose the tuning parameters appropriately for brevity of discussion. Given an i.i.d. sample $\{W_i\}_{i=1}^N$ of size $N$, define the sample quantities
\begin{displaymath}
    \hat\mu_{G_L}(\lambda) = \frac{1}{N}\sum_{i=1}^N G_L(\lambda, W_i) \quad\text{and}\quad \hat\sigma_{G_L}(\lambda) = \sqrt{(1+\kappa)\frac{1}{N}\sum_{i=1}^N \left( G_L(\lambda, W_i) - \hat\mu_{G_L}(\lambda)\right)^2}
\end{displaymath}
where $\kappa = 0.05$ is a small number, and where $\hat\mu_{G_U}(\lambda)$ and $\hat\sigma_{G_U}^2(\lambda)$ are defined similarly with $G_U$. 
Define the test statistic as
\begin{displaymath}
    T_{AS}(\theta) = \sup_{\lambda \in Q^K} ~\max\left\{\frac{\sqrt{N}(\hat\mu_{G_L}(\lambda) - \theta)}{\hat\sigma_{G_L}(\lambda)}, ~\frac{\sqrt{N}(\theta - \hat\mu_{G_U}(\lambda))}{\hat\sigma_{G_U}(\lambda)}, ~0\right\}^2,
\end{displaymath}
which corresponds to the function $S_3$ in \citet{andrews2017}. This test statistic is then compared to the critical value $c_{AS}(\alpha)$, which can be computed in two ways: the plug-in asymptotic (PA) type and the generalized moment selection (GMS) type. For brevity of discussion, I briefly outline the PA type critical value here, which yields an intuitive expression for a confidence interval of $\theta$, and refer to \citet{andrews2017} for the GMS type critical value.

Let $\{W_i^{(b)}\}_{i=1}^N$ be the empirical bootstrap sample of $\{W_i\}_{i=1}^N$, meaning each $\{W_i^{(b)}\}_{i=1}^N$ is drawn from $\{W_i\}_{i=1}^N$ with replacement. Let $\hat\mu_{G_L}^{(b)}$, $\hat\mu_{G_U}^{(b)}$, $\hat\sigma_{G_L}^{(b)}$, $\hat\sigma_{G_U}^{(b)}$ be the values of $\hat\mu_{G_L}$, $\hat\mu_{G_U}$, $\hat\sigma_{G_L}$, $\hat\sigma_{G_U}$ computed with $\{W_i^{(b)}\}_{i=1}^N$ instead of $\{W_i\}_{i=1}^N$. Then, compute the statistic
\begin{displaymath}
    \begin{aligned}
        c_{AS}^{(b)}(\theta) = \sup_{\lambda \in Q^K} ~ \max \Big\{ &
        \frac{\sqrt{N}(\hat\mu_{G_L}^{(b)}(\lambda) - \hat\mu_{G_L}(\lambda))%
        }{\hat\sigma_{G_L}^{(b)}(\lambda)}, %
        \frac{\sqrt{N}(\hat\mu_{G_U}(\lambda) - \hat\mu_{G_U}^{(b)}(\lambda))%
        }{\hat\sigma_{G_U}^{(b)}(\lambda)},~0\Big\}^2.
    \end{aligned}
\end{displaymath}
The critical value $c_{AS}(\theta,\alpha)$ is then defined as the $(1-\alpha)$ quantile of the bootstrapped $c_{AS}^{(b)}$ values. The confidence set for $\theta$ is then given by $\{\theta ~|~ T_{AS}(\theta) \leq c_{AS}(\theta,\alpha)\}$. Note that %
the critical value $c_{AS}(\theta,\alpha)$ does not depend on $\theta$. %
Consequently, the PA type confidence set simplifies to the interval
\begin{displaymath}
    \left[\sup_{\lambda}\left\{\hat\mu_{G_L}(\lambda) - \sqrt{c_{AS}(\alpha)} \times \frac{\hat\sigma_{G_L}(\lambda)}{\sqrt{N}}\right\}, ~~
    \inf_{\lambda}\left\{\hat\mu_{G_U}(\lambda) + \sqrt{c_{AS}(\alpha)} \times \frac{\hat\sigma_{G_U}(\lambda)}{\sqrt{N}}\right\} \right].
\end{displaymath}

When $K$ is large, searching for supremum over all $\lambda \in Q^K$ in $T_{AS}(\theta)$ and $c_{AS}^{(b)}(\theta)$ can be computationally prohibitive. However, note that the inequalities in (\ref{eq.inference.countable}) bind only at two $\lambda$ values, namely at $\lambda_L^* = \argmax_\lambda\E( G_L(\lambda, W_i))$ and $\lambda_U^* = \argmin_\lambda\E( G_U(\lambda, W_i))$. Moreover, since $G_L$ is concave and $G_U$ is convex, the inequalities become loose for $\lambda$ values that are far from $\lambda_L^*$ and $\lambda_U^*$. Consequently, in practice, one can focus the search for $\lambda$ on neighborhoods of $\lambda_L^*$ and $\lambda_U^*$. 
While $\lambda_L^*$ and $\lambda_U^*$ are population quantities, they can be approximated with their sample analogues. %

The procedure naturally extends to a vector-valued parameter $\theta \in \mb{R}^d$, by considering (\ref{eq.inference}) on each component of $\theta$. For example, the moment inequalities for $\theta = (\theta_1, \theta_2) \in \mb{R}^2$ are:
\begin{equation}
    \begin{aligned}
        \E(G_{L1}(\lambda, W_i) - \theta_1) \leq 0 \quad \text{ for all } \lambda \in Q^K, \\
        \E(\theta_1 - G_{U1}(\lambda, W_i)) \leq 0 \quad \text{ for all } \lambda \in Q^K, \\
        \E(G_{L2}(\lambda, W_i) - \theta_2) \leq 0 \quad \text{ for all } \lambda \in Q^K, \\
        \E(\theta_2 - G_{U2}(\lambda, W_i)) \leq 0 \quad \text{ for all } \lambda \in Q^K,
    \end{aligned}
\end{equation}
where $G_{Uk}$ and $G_{Lk}$ denote the functions $G_L$ and $G_U$ in (\ref{eq.inference}) corresponding to $\theta_k$ for $k=1,2$. Applying the same inference procedure then yields a confidence region in $\mb{R}^2$. This extension can be used to construct a confidence interval for the variance of random coefficients, which involves both first and second moments. Alternatively, it can be constructed by the Bonferroni correction to the individual bounds.

Lastly, I discuss overidentification and model misspecification in inference for the general parameters. Under overidentification or misspecification, %
the test statistic $T_{AS}(\theta)$ and its bootstrap critical value $c_{AS}^{(b)}(\theta)$ all diverge to $+\infty$. In contrast to the case of mean parameters, it is substantially more challenging to deal with these issues for general parameters. \citet{andrews2024} develop a general method for constructing valid confidence intervals under overidentification or misspecification, but their approach applies to a finite number of moment restrictions and therefore is not readily applicable to the countably infinite set considered here. Extending their approach to countably many moment restrictions is beyond the scope of this paper and is not pursued here. Instead, in Online Appendix \ref{sec.inference.general.overid}, I discuss a heuristic modification of the procedure of \citet{andrews2017}, %
which closely align with the spirit of \citet{andrews2024}. %
I check the performance of this heuristic procedure via simulation, also in Online Appendix \ref{sec.inference.general.overid}.

\section{Application to lifecycle earnings dynamics}

\label{sec.application}

\subsection{Overview}

Lifecycle earnings dynamics serve as a key input in various macroeconomic models. For example, in models of consumption and savings dynamics \citep{hall1982,blundell2008,blundell2016,arellano2017}, households facing a higher earnings risk accumulate more precautionary savings to smooth consumption over time. %
As \citet{guvenen2009} points out, specifying an earnings process that highlights features of real data is important for properly calibrating and drawing conclusions from these models.

When used as an input, it is common to specify earnings dynamics using a parsimonious linear model. \citet{guvenen2007,guvenen2009} studied two leading views on parsimonious specification of the earnings dynamics. Consider two earnings processes\footnote{As \citet{guvenen2007} points out, this is a stylized version of what is used in the literature, but it still captures features important for the discussion.}:
\begin{equation}
    \begin{array}{llll}
        Y_{it} &= \alpha_i + z_{it}, \qquad &z_{it} = \rho z_{i,t-1} + \eta_{it}, & \qquad\text{(RIP)} \\
        Y_{it} &= \alpha_i + \beta_i h_{it} + z_{it}, \qquad &z_{it} = \rho z_{i,t-1} + \eta_{it}, & \qquad\text{(HIP)}
    \end{array}
    \label{eq.application.incomeprocess.intro}
\end{equation}
where $h = \text{age} - \max\{\text{years of schooling}, 12\} - 6$ is potential years of experience, $Y_{it}$ is the residual log-earnings obtained by regressing log-earnings on time indicators and their interactions with a cubic polynomial in $h$, and $(\alpha_i, \beta_i)$ are heterogeneous coefficients. In addition, $\{z_{it}\}$ is an AR(1) process with a mean zero shock $\eta_{it}$\footnote{In the literature, it is standard to add a transitory income process to (\ref{eq.application.incomeprocess.intro}). I present estimation results that account for a transitory income process in Online Appendix \ref{sec.appendix.deconvolution}. The estimation results yield similar qualitative conclusions outlined in this subsection.}. These two models are known as the Restricted Income Profiles (RIP) process and the Heterogeneous Income Profiles (HIP) process, respectively. In both models, $\rho$ captures the earnings persistence that households face. As \citet{guvenen2009} summarizes, the literature reports $0.5 < \rho < 0.7$ and $\var(\beta_i) > 0$ for the HIP process (e.g., \citealp{lillard1979,baker1997}), meaning that households experience modest persistence and heterogeneous trends. By contrast, \citet{macurdy1982} tested the hypothesis that $\var(\beta_i) = 0$ and did not reject it. The literature reports $\rho \approx 1$ for the RIP process (e.g., \citealp{abowd1989,topel1992}), meaning households experience extreme persistence and homogeneous trends. \citet{guvenen2007} demonstrated that the HIP process better aligns with features of consumption data, and \citet{guvenen2009} showed that misspecifying the HIP process as a RIP process leads to an upward biased estimate of $\rho$, often obtaining $\rho \approx 1$.

While there is vast literature on unobserved heterogeneity in $\beta_i$ and its influence on $\rho$, relatively few studies examines heterogeneity in $\rho$ itself. Notable recent studies include \citet{browning2010}, \citet{alan2018}, and \citet{pesaran2024}; the first two assume a factor structure for $\rho_i$, and the latter imposes stationarity of (\ref{eq.application.incomeprocess.intro}) and assumes $\eta_{it}$ are i.i.d. over $i$ and $t$. In this section, I estimate a generalization of (\ref{eq.application.incomeprocess.intro}) where $\rho$ varies across individuals, writing $\rho=\rho_i$, where the distribution of $\rho_i$ and its correlation with $(\alpha_i, \beta_i, Y_{i0})$ are unrestricted. Differently from \citet{pesaran2024}, who also extend \citet{guvenen2009}, the distribution of $\eta_{it}$ also remains unrestricted and may depend on $(\alpha_i, \beta_i, \rho_i)$, allowing for heteroskedasticity.

In the remainder of this section, I find that, when $\rho$ is allowed to vary across individuals, both RIP and HIP specifications deliver similar estimates of $\E(\rho_i)$ that are significantly less than $1$. At the 95\% confidence level, the upper bounds of the confidence intervals for $\E(\rho_i)$ under both processes are between $0.5$ and $0.6$, and the two intervals have substantial overlap. This result suggests that, when $\rho$ is allowed to be heterogeneous, choosing RIP over HIP or vice versa may not lead to serious misspecification of $\rho_i$. Moreover, the 90\% confidence intervals for $\var(\rho_i)$ and $\pp(\rho_i \leq r)$ for $r \in (0,1)$ in the RIP model suggest the presence of heterogeneity in $\rho_i$. In particular, the lower confidence limit for $\var(\rho_i)$ is $0.009$, implying a standard deviation of $0.097$, and the confidence intervals for the CDF of $\rho_i$ suggest that at least 41\% of individuals have $\rho_i \leq 0.8$.

\subsection{Data and models}

I analyze data on U.S. households from the Panel Study of Income Dynamics (PSID) dataset. I use the dataset of \citet{guvenen2009}, who analyzed the PSID dataset of male heads of households collected annually. The dataset consists of male head of households who are not in the poverty (SEO) subsample and who consecutively reported positive hours (between 520 and 5110 hours a year) and earnings (between a preset minimum and maximum wage). From the dataset of \citet{guvenen2009}, I select individuals observed consecutively from 1976 to 1991, yielding $N=800$ and $T=15$, where the first wave serves as the initial value of earnings. I estimate two dynamic random coefficient models:
\begin{equation}
    \begin{array}{llll}
        Y_{it} &= \alpha_i + \rho_i Y_{i,t-1} + \eta_{it}, & \E(\eta_{it}|\alpha_i, \rho_i, Y_i^{t-1}) = 0, & \qquad\text{(RIP-RC)} \\
        Y_{it} &= \alpha_i + \beta_i h_{it} + \rho_i Y_{i,t-1} + \eta_{it}, & \E(\eta_{it}|\alpha_i, \beta_i, \rho_i, Y_i^{t-1}, h_i) = 0. & \qquad\text{(HIP-RC)}
    \end{array}
    \label{eq.application.incomeprocess}
\end{equation}
These models generalize (\ref{eq.application.incomeprocess.intro}), and they can be derived by quasi-differencing $Y_{it}$ in (\ref{eq.application.incomeprocess.intro}) and assuming $h_{it}\approx h_{i,t-1}+1$. Specifically, quasi-differencing the RIP process gives
\begin{displaymath}
    Y_{it} = \alpha_i(1-\rho_i) + \rho_i Y_{i,t-1} + \eta_{it} \equiv \tilde\alpha_i + \rho_i Y_{i,t-1} + \eta_{it}.
\end{displaymath}
Likewise, quasi-differencing the HIP process gives
\begin{displaymath}
    Y_{it} = \alpha_i(1-\rho_i) + \beta_i\rho_i + \beta_i(1-\rho_i)h_{it} + \rho_i Y_{i,t-1} + \eta_{it} \equiv \tilde\alpha_i + \tilde\beta_i h_{it} + \rho_i \tilde Y_{i,t-1} + \eta_{it}.
\end{displaymath}
Note that \citet{guvenen2009} defines $Y_{it}$ in (\ref{eq.application.incomeprocess.intro}) as the residual from the regression on time indicators and their interactions with a cubic polynomial in $h_{it}$, i.e., the regression
\begin{equation}
    \begin{aligned}
        Y_{it} &= \sum_{s=1976}^{1991} \left(\textbf{1}(t=s)\delta_{0,s} + \textbf{1}(t=s) h_{it} \delta_{1,s} + \textbf{1}(t=s) h_{it}^2 \delta_{2,s} + \textbf{1}(t=s) h_{it}^3 \delta_{3,s}\right) + \nu_{it} \\
        &\equiv X_{it}'\delta + \nu_{it},
    \end{aligned}
    \label{eq.application.firststage}
\end{equation}
where $Y_{it}$ is now the raw log-earnings data, and $X_{it}$ and $\delta$ denote the regressors and the coefficients in (\ref{eq.application.firststage}), i.e., $X_{it} = \text{vec}\left\{(\textbf{1}(t=s), \textbf{1}(t=s) h_{it}, \textbf{1}(t=s) h_{it}^2, \textbf{1}(t=s) h_{it}^3)_{s=1976}^{1991}\right\}$ and $\delta = \text{vec}\left\{(\delta_{0,s}, \delta_{1,s}, \delta_{2,s}, \delta_{3,s})_{s=1976}^{1991}\right\}$. \citet{guvenen2009} estimate the RIP and HIP models in (\ref{eq.application.incomeprocess.intro}) using the two-step procedure that is standard in the literature, where one first obtains the residuals from the regression in (\ref{eq.application.firststage}), and then one treats these residuals as $Y_{it}$ and estimate the RIP and HIP models in (\ref{eq.application.incomeprocess.intro}). The motivation of this approach is to first ``partial out'' the control variables $X_{it}$ and then consider earnings dynamics that are free of $X_{it}$. However, this approach may understate the standard errors of the RIP and HIP estimates, since it fails to account for the sampling variability introduced by the first-stage regression. Moreover, any estimation error in the first stage may appear as heterogeneity in $(\alpha_i, \beta_i)$ in the RIP-RC and HIP-RC specifications. To address these issues, I also consider a joint model of the control variables term in (\ref{eq.application.firststage}) and the RIP-RC and HIP-RC models in (\ref{eq.application.incomeprocess}). Specifically, I estimate
\begin{equation}
    \begin{array}{llll}
        Y_{it} &= X_{it}'\delta + \alpha_i + \rho_i Y_{i,t-1} + \eta_{it}, & \E(\eta_{it}|\alpha_i, \rho_i, Y_i^{t-1}, X_i) = 0, & \text{(RIP-RC-J)} \\
        Y_{it} &= X_{it}'\delta + \alpha_i + \beta_i h_{it} + \rho_i Y_{i,t-1} + \eta_{it}, & \E(\eta_{it}|\alpha_i, \beta_i, \rho_i, Y_i^{t-1}, h_i, X_i) = 0. & \text{(HIP-RC-J)}
    \end{array}
    \label{eq.application.incomeprocess.homo}
\end{equation}
where the homogeneous coefficients $\delta$ and the heterogeneous coefficients $(\alpha_i, \beta_i, \rho_i)$ are jointly considered. I refer to these specifications in (\ref{eq.application.incomeprocess.homo}) as RIP-RC-J and HIP-RC-J. %
I estimate the mean parameters of (\ref{eq.application.incomeprocess.homo}) using the bounds in \Cref{prop.mean.closedform.refined.constant}. 

Note that, for estimation of the RIP-RC-J model, the regressor $\textbf{1}(t=1976)$ must be removed from the model because it is multicollinear with the individual-specific intercept $\alpha_i$. Likewise, for estimation of the HIP-RC-J model, both $\mathbf{1}(t=1976)$ and $\mathbf{1}(t=1976)h_{it}$ must be dropped from $X_{it}$ to avoid multicollinearity with the $\alpha_i$ and the $h_{it}$ terms. These exclusions ensure that the no-multicollinearity condition of \Cref{ass.mean.nomulticollinearity.constant} holds. After these removals, $X_{it}$ has 59 regressors in RIP-RC-J and 58 in HIP-RC-J models. The estimation result below will show that the bounds in \Cref{prop.mean.closedform.refined.constant} produces informative confidence intervals under this setup, demonstrating its practical applicability with a large number of regressors with homogeneous coefficients.

In what follows, I construct confidence intervals for $\E(\rho_i)$ under the RIP-RC, HIP-RC, RIP-RC-J, and HIP-RC-J specifications, using the bounds presented in \Cref{prop.mean.closedform.refined,prop.mean.closedform.refined.constant}. For the RIP-RC and HIP-RC models, I employ the two-step procedure of \citet{guvenen2009} which first obtains residuals from (\ref{eq.application.firststage}) and then uses these residuals as $Y_{it}$ in the RIP-RC and HIP-RC models. In contrast, for the RIP-RC-J and HIP-RC-J models, I directly estimate $\E(\rho_i)$ from (\ref{eq.application.incomeprocess.homo}), jointly considering the control variables term. In addition, I construct the confidence intervals for $\var(\rho_i)$ and $\pp(\rho_i \leq r)$ over the grid $r \in \{0.1, \ldots, 0.9\}$ under the RIP-RC model, using the bounds presented in \Cref{sec.general.varcdf} and assuming $\mc{B} = [-3,3] \times [0,1]$ as the support of $(\alpha_i, \rho_i)$.

For calculation of the mean bounds, I choose $S_{it} = (1, Y_{i,\max\{0, ~t-5\}}, \ldots, Y_{i,t-1})'$ for RIP-type models, and $S_{it} = (1, Y_{i,\max\{0, ~t-5\}}, \ldots, Y_{i,t-1}, h_{i,\max\{1, ~t-5\}}, \ldots, h_{i,\min\{T, ~t+5\}})$ for HIP-type models. For calculation of the variance and the CDF bounds, I choose $S_{it} = (1, Y_{i,\max\{0, ~t-4\}}, \ldots, Y_{i,t-1})'$ for the RIP-RC model. For these choices of $S_{it}$, the overidentification issue arises in the estimated bounds. For inference on $\E(\rho_i)$, I apply the procedure of \citet{stoye2020} discussed in \Cref{sec.inference.mean}. For inference on $\var(\rho_i)$ and $\pp(\rho_i \leq r)$, I adopt the heuristic modification of \citet{andrews2017} described in Online Appendix \ref{sec.inference.general.overid}. Guided by simulation results, I evaluate the supremum with $100$ grid points in the neighborhoods. All critical values are calculated with $1000$ bootstrap replications, using the PA type for $\var(\rho_i)$ and $\pp(\rho_i \leq r)$. The interval for $\var(\rho_i)$ is constructed with the Bonferroni correction.

\subsection{Results}

\label{sec.application.empirical}

The 95\% confidence intervals for $\E(\rho_i)$ are reported in \Cref{table.application.mean}. Both models estimate $\E(\rho_i)$ to be significantly less than 1, and the confidence intervals in RIP-RC and HIP-RC demonstrate substantial overlap, having similar upper confidence limits. This suggests that specifying homogeneous or heterogeneous $\beta_i$ does not lead to serious misspecification when $\rho_i$ is allowed to be heterogeneous. The confidence intervals for RIP-RC-J and HIP-RC-J are similar to those for RIP-RC and HIP-RC, supporting the same argument. Note that these intervals are computed with the procedure described in \Cref{sec.inference.mean}, which is robust to overidentification and model misspecification. These findings are qualitatively similar when also considering the transitory income process, as reported in Online Appendix \ref{sec.appendix.deconvolution}.

\begin{table}[!tbp]
    \centering %
    \begin{tabular}{c c c c c} %
    \hline\hline %
    Parameter & RIP-RC & HIP-RC & RIP-RC-J & HIP-RC-J \\ %
    \hline\hline %
    $\E(\rho_i)$           & [0.425, 0.596] & [0.253, 0.566] & [0.455, 0.588] & [0.260, 0.543] \\
    \hline %
    \end{tabular}
    \caption{Confidence intervals of $\E(\rho_i)$ for the RIP type and the HIP type processes with heterogeneous coefficients. The nominal coverage probability is $0.95$. These confidence intervals are robust to overidentification and model misspecification.}
    \label{table.application.mean}
\end{table}

The 90\% confidence intervals for $\var(\rho_i)$ and $\pp(\rho_i \leq r)$ over the grid $r \in \{0.1, \ldots, 0.9\}$ for the RIP-RC model are reported in \Cref{table.application.varcdf}. The lower confidence limit of $\var(\rho_i)$ is $0.009$, implying a standard deviation of $0.097$, suggesting heterogeneity in $\rho_i$. Similar evidence is observed from confidence intervals for the CDF of $\rho_i$. They indicate that at least 41\% of households have $\rho_i \leq 0.8$ and at least 31\% have $\rho_i \leq 0.5$. These findings suggest unobserved heterogeneity in the earnings risk that households face, highlighting the importance of allowing for heterogeneity in $\rho_i$ in modeling income processes that reflect features of real data. %

\begin{table}[!tbp]
    \centering %
    \begin{tabular}{c c c} %
    \hline\hline %
    Parameter & RIP-RC \\ %
    \hline\hline %
    $\var(\rho_i)$         & [0.009, 0.233] \\
    \hline
    $\pp(\rho_i \leq 0.1)$ & [0.087, 0.996] \\
    $\pp(\rho_i \leq 0.2)$ & [0.124, 0.996] \\
    $\pp(\rho_i \leq 0.3)$ & [0.159, 1.000] \\
    $\pp(\rho_i \leq 0.4)$ & [0.196, 1.000] \\
    $\pp(\rho_i \leq 0.5)$ & [0.310, 1.000] \\
    $\pp(\rho_i \leq 0.6)$ & [0.365, 1.000] \\
    $\pp(\rho_i \leq 0.7)$ & [0.396, 1.000] \\
    $\pp(\rho_i \leq 0.8)$ & [0.419, 1.000] \\
    $\pp(\rho_i \leq 0.9)$ & [0.413, 1.000] \\
    \hline %
    \end{tabular}
    \caption{Confidence intervals of $\var(\rho_i)$ and $\pp(\rho_i \leq r)$ for the RIP and the HIP processes with heterogeneous coefficients. The nominal coverage probability is $0.90$.}
    \label{table.application.varcdf}
\end{table}

\section{Conclusion}

\label{sec.conclusion}

This paper studies the identification and estimation of dynamic random coefficient models in a short panel context. The model extends the standard dynamic panel linear model with fixed effects \citep{arellano1991,blundell1998}, allowing coefficients to be individual-specific. I show that the model is not point-identified but rather partially identified, and I characterize the identified sets of the mean, variance and CDF of the random coefficients using the dual representation of an infinite-dimensional linear program. I propose a computationally tractable estimation and inference procedure by adopting the approach of \citet{stoye2020} for the mean parameters and \citet{andrews2017} for the variance and CDF parameters. The procedure of \citet{stoye2020} is robust to overidentification and model misspecification.

I use my method to estimate unobserved heterogeneity in earnings persistence across U.S. households using the PSID dataset. I find that the average earnings persistence is significantly less than 1 when it is allowed to be heterogeneous. Moreover, its confidence interval under the RIP and HIP specifications show substantial overlap, suggesting that choosing RIP over HIP or vice versa does not lead to serious misspecification about the earnings process when persistence is heterogeneous. Lastly, confidence intervals for the variance and CDF of the earnings persistence parameter suggest the presence of unobserved heterogeneity, which is a key source of heterogeneity in consumption and savings behaviors.

\section{Acknowledgements}

This is based on my PhD dissertation at the University of Chicago. I am deeply indebted to St\'ephane Bonhomme, Alexander Torgovitsky, and Guillaume Pouliot for their invaluable guidance and support. I also thank Francesca Molinari and three anonymous reviewers for their insightful comments and suggestions. I thank Manuel Arellano, Timothy Armstrong, Antonio Galvao, Greg Kaplan, Roger Koenker, Zhipeng Liao, Jack Light, Hashem Pesaran, Azeem Shaikh, Shuyang Sheng, Panagiotis Toulis, and Ying Zhu for helpful discussions and comments. I also thank participants of the Econometrics Working Group at the University of Chicago.

\singlespacing
\addcontentsline{toc}{chapter}{\numberline{}Bibliography}
\bibliographystyle{jpe}
\bibliography{refCRC,refLC,refEXTRA}
\onehalfspacing

\newpage

\begin{appendices}

\label{sec.appendix}

\section{Online Appendix: Proofs}

\label{sec.proof}

\textbf{Proof of Proposition \ref{prop.mean.failure}}.
To prove \Cref{prop.mean.failure}, I start from a data generating process that satisfies (\ref{eq.ar1}) and then modify it to construct a new process that is observationally equivalent.

Let $(\gamma_i, \beta_i, Y_{i0}, Y_{i1}, Y_{i2})$ be the random variables that satisfy (\ref{eq.ar1}). These variables must satisfy two conditions. First, they generate the observed data:
\begin{equation}
    \pp(Y_{i0} \leq y_0, Y_{i1} \leq y_1, Y_{i2} \leq y_2) = \mathbb{F}(y_0,y_1,y_2),
    \label{eq.proof.failure.data}
\end{equation}
where $\mathbb{F}$ is the observed cumulative distribution function of $(Y_{i0}, Y_{i1}, Y_{i2})$. Second, they satisfy the model constraints:
\begin{equation}
    \begin{aligned}
        \E(Y_{i1} - \gamma_i - \beta_i Y_{i0}|\gamma_i, \beta_i, Y_{i0}) &= 0 \qquad \text{ for all } (\gamma_i,\beta_i,Y_{i0}), \\
        \E(Y_{i2} - \gamma_i - \beta_i Y_{i1}|\gamma_i, \beta_i, Y_{i0}, Y_{i1}) &= 0 \qquad \text{ for all } (\gamma_i,\beta_i,Y_{i0},Y_{i1}). \\
    \end{aligned}
    \label{eq.proof.failure.epsilon}
\end{equation}

Now, I construct a new data generating process that also satisfies (\ref{eq.proof.failure.data}) and (\ref{eq.proof.failure.epsilon}). Given the joint distribution of $(\gamma_i, \beta_i, Y_{i0}, Y_{i1})$, let $k$ be a deterministic function of $(\gamma_i, \beta_i, Y_{i0}, Y_{i1})$ such that
\begin{equation}
    \E((Y_{i1}-Y_{i0})k(\gamma_i, \beta_i, Y_{i0}, Y_{i1})|\gamma_i, \beta_i, Y_{i0}) = 0,
    \label{eq.proof.failure.k}
\end{equation}
so that $k$ is orthogonal to $Y_{i1}-Y_{i0}$ conditional on $(\gamma_i, \beta_i, Y_{i0})$. I will specify the explicit expression for $k$ later in the proof. Then, given the function $k$, I define new random coefficients $(\tilde\gamma_i, \tilde\beta_i)$ as the following deterministic functions of $(\gamma_i, \beta_i, Y_{i0}, Y_{i1})$:
\begin{equation}
    \left(\left.\begin{array}{c}
        \tilde\gamma_i \\
        \tilde\beta_i
    \end{array}\right)\right|(Y_{i0}=y_0, Y_{i1}=y_1, \gamma_i=r, \beta_i=b)
    =
    \left(\begin{array}{cc}
        r - y_1 k(r, b, y_0, y_1) \\
        b + k(r, b, y_0, y_1)
    \end{array}\right).
    \label{eq.proof.failure.coef}
\end{equation}
I then define $\tilde Y_{i2}$ conditional on $(Y_{i0}, Y_{i1}, \gamma_i, \beta_i, \tilde\gamma_i, \tilde\beta_i)$ as
\begin{equation}
    \begin{aligned}
        &\tilde Y_{i2}|(Y_{i0}=y_0, ~Y_{i1}=y_1, ~\gamma_i=r, ~\beta_i=b, ~\tilde\gamma_i=\tilde r, ~\tilde\beta_i=\tilde b) \\
        &\overset{d}{=} Y_{i2}|(Y_{i0}=y_0, ~Y_{i1}=y_1, ~\gamma_i=r, ~\beta_i=b).
    \end{aligned}
    \label{eq.proof.failure.y2}
\end{equation}
Equations (\ref{eq.proof.failure.coef}) and (\ref{eq.proof.failure.y2}) together yield a joint distribution of $(\gamma_i, \beta_i, \tilde\gamma_i, \tilde\beta_i, Y_{i0}, Y_{i1}, \tilde Y_{i2})$, which in turn implies a joint distribution of $(\tilde\gamma_i, \tilde\beta_i, Y_{i0}, Y_{i1}, \tilde Y_{i2})$. I now show that this new data generating process, given by $(\tilde\gamma_i, \tilde\beta_i, Y_{i0}, Y_{i1}, \tilde Y_{i2})$, satisfies both (\ref{eq.proof.failure.data}) and (\ref{eq.proof.failure.epsilon}).

I first show that the new data generating process satisfies (\ref{eq.proof.failure.data}). Since the joint distribution of $(Y_{i0}, Y_{i1})$ remains unchanged, it suffices to verify that, conditional on $(Y_{i0}=y_0, Y_{i1}=y_1)$, the distribution of $\tilde Y_{i2}$ coincides with that of $Y_{i2}$. By the definition of $\tilde Y_{i2}$ in (\ref{eq.proof.failure.y2}) and the law of iterated expectations, I obtain
\begin{displaymath}
    \begin{aligned}
        &\pp(\tilde Y_{i2} \leq y_2 | Y_{i0}, Y_{i1}) \\
        &= \E(\pp(\tilde Y_{i2} \leq y_2 | Y_{i0}, Y_{i1}, \gamma_i, \beta_i, \tilde\gamma_i, \tilde\beta_i)|Y_{i0}, Y_{i1}) \\
        &= \E(\pp(Y_{i2} \leq y_2 | Y_{i0}, Y_{i1}, \gamma_i, \beta_i)|Y_{i0}, Y_{i1}) \\
        &= \pp(Y_{i2} \leq y_2 | Y_{i0}, Y_{i1}), \\
    \end{aligned}
\end{displaymath}
which shows that (\ref{eq.proof.failure.data}) holds for the new process. Next, to show that the new data generating process also satisfies (\ref{eq.proof.failure.epsilon}), note that in conditional expectations, conditioning on $(\gamma_i, \beta_i, \tilde\gamma_i, \tilde\beta_i, Y_{i0}, Y_{i1})$ is equivalent to conditioning on $(\gamma_i, \beta_i, Y_{i0}, Y_{i1})$ because $(\tilde\gamma_i, \tilde\beta_i)$ is a deterministic function of $(\gamma_i, \beta_i, Y_{i0}, Y_{i1})$. Then, by the law of iterated expectations, the following shows that the first line of (\ref{eq.proof.failure.epsilon}) holds for the new data generating process:
\begin{displaymath}
    \begin{aligned}
        &\E(Y_{i1} - \tilde\gamma_i - \tilde\beta_iY_{i0}|\tilde\gamma_i, \tilde\beta_i, Y_{i0}) \\
        &=\E(\E(Y_{i1} - \tilde\gamma_i - \tilde\beta_iY_{i0}|\gamma_i, \beta_i, \tilde\gamma_i, \tilde\beta_i, Y_{i0}, Y_{i1})|\tilde\gamma_i, \tilde\beta_i, Y_{i0}) \\
        &=\E(\E(Y_{i1} - \tilde\gamma_i - \tilde\beta_iY_{i0}|\gamma_i, \beta_i, Y_{i0}, Y_{i1})|\tilde\gamma_i, \tilde\beta_i, Y_{i0}) \\
        &=\E(\E(Y_{i1} - \tilde\gamma_i - \tilde\beta_iY_{i0}|\gamma_i, \beta_i, Y_{i0})|\tilde\gamma_i, \tilde\beta_i, Y_{i0}) \\
        &=\E(\E(Y_{i1} - \gamma_i - \beta_iY_{i0} + (Y_{i1} - Y_{i0})k(\gamma_i,\beta_i,Y_{i0},Y_{i1})|\gamma_i, \beta_i, Y_{i0})|\tilde\gamma_i, \tilde\beta_i, Y_{i0}) \\
        &=\E(0 + 0|\tilde\gamma_i, \tilde\beta_i, Y_{i0}) = 0, \\
    \end{aligned}
\end{displaymath}
where the last line is obtained by the first line of (\ref{eq.proof.failure.epsilon}) and by (\ref{eq.proof.failure.k}). Similarly, the following derivation shows that the second line of (\ref{eq.proof.failure.epsilon}) holds for the new data generating process:
\begin{displaymath}
    \begin{aligned}
        &\E(\tilde Y_{i2} - \tilde\gamma_i - \tilde\beta_iY_{i1}|\tilde\gamma_i, \tilde\beta_i, Y_{i0}, Y_{i1}) \\
        &=\E(\E(\tilde Y_{i2} - \tilde\gamma_i - \tilde\beta_iY_{i1}|\gamma_i, \beta_i, \tilde\gamma_i, \tilde\beta_i, Y_{i0}, Y_{i1})|\tilde\gamma_i, \tilde\beta_i, Y_{i0}, Y_{i1}) \\
        &=\E(\E(\tilde Y_{i2} - \tilde\gamma_i - \tilde\beta_iY_{i1}|\gamma_i, \beta_i, Y_{i0}, Y_{i1})|\tilde\gamma_i, \tilde\beta_i, Y_{i0}, Y_{i1}) \\
        &=\E(\E(Y_{i2} - \tilde\gamma_i - \tilde\beta_iY_{i1}|\gamma_i, \beta_i, Y_{i0}, Y_{i1})|\tilde\gamma_i, \tilde\beta_i, Y_{i0}, Y_{i1}) \\
        &=\E(\E(Y_{i2} - \gamma_i - \beta_iY_{i1} \\
        &\qquad\qquad + Y_{i1}k(\gamma_i,\beta_i,Y_{i0},Y_{i1}) - Y_{i1}k(\gamma_i,\beta_i,Y_{i0},Y_{i1})|\gamma_i, \beta_i, Y_{i0}, Y_{i1})|\tilde\gamma_i, \tilde\beta_i, Y_{i0}, Y_{i1}) \\
        &=\E(\E(Y_{i2} - \gamma_i - \beta_iY_{i1} |\gamma_i, \beta_i, Y_{i0}, Y_{i1})|\tilde\gamma_i, \tilde\beta_i, Y_{i0}, Y_{i1}) =\E(0|\tilde\gamma_i, \tilde\beta_i, Y_{i0}, Y_{i1}) = 0,\\
    \end{aligned}
\end{displaymath}
where the second-last equality follows from the second line of (\ref{eq.proof.failure.epsilon}).

In summary, I have shown that, given the function $k$, $(\tilde\gamma_i, \tilde\beta_i, Y_{i0}, Y_{i1}, \tilde Y_{i2})$ is a new data generating process that satisfies (\ref{eq.proof.failure.data}) and (\ref{eq.proof.failure.epsilon}). It now remains to specify the expression for $k$. Among many possible choices for $k$ that satisfy (\ref{eq.proof.failure.k}), I consider the following expression:
\begin{displaymath}
    k(r,b,y_0,y_1) = 1 - \frac{\E(Y_{i1}-Y_{i0}|\gamma_i=r,\beta_i=b,Y_{i0}=y_0)}{\E((Y_{i1}-Y_{i0})^2|\gamma_i=r,\beta_i=b,Y_{i0}=y_0)}(y_1-y_0).
\end{displaymath}
The following calculation shows that this choice of $k$ satisfies (\ref{eq.proof.failure.k}):
\begin{displaymath}
    \begin{aligned}
        &\E((Y_{i1}-Y_{i0})k(\gamma_i,\beta_i,Y_{i0},Y_{i1})|\gamma_i,\beta_i,Y_{i0}) \\
        &= \E\left(\left.(Y_{i1}-Y_{i0}) - \frac{\E(Y_{i1}-Y_{i0}|\gamma_i,\beta_i,Y_{i0})}{\E((Y_{i1}-Y_{i0})^2|\gamma_i,\beta_i,Y_{i0})}(Y_{i1}-Y_{i0})^2\right|\gamma_i,\beta_i,Y_{i0}\right) \\
        &= \E(Y_{i1}-Y_{i0}|\gamma_i,\beta_i,Y_{i0}) - \E(Y_{i1}-Y_{i0}|\gamma_i,\beta_i,Y_{i0}) = 0.
    \end{aligned}
\end{displaymath}
Then, under this choice of $k$, the expectation of $\tilde\beta_i$ is calculated as
\begin{displaymath}
    \begin{aligned}
        \E(\tilde\beta_i) &= \E(\beta_i + k(\gamma_i,\beta_i,Y_{i0},Y_{i1})) \\
        &= \E(\beta_i) + 1 - \E\left(\frac{\E(Y_{i1}-Y_{i0}|\gamma_i,\beta_i,Y_{i0})}{\E((Y_{i1}-Y_{i0})^2|\gamma_i,\beta_i,Y_{i0})}(Y_{i1}-Y_{i0})\right) \\
        &= \E(\beta_i) + 1 - \E\left(\E\left(\left.\frac{\E(Y_{i1}-Y_{i0}|\gamma_i,\beta_i,Y_{i0})}{\E((Y_{i1}-Y_{i0})^2|\gamma_i,\beta_i,Y_{i0})}(Y_{i1}-Y_{i0})\right|\gamma_i,\beta_i,Y_{i0}\right)\right) \\
        &= \E(\beta_i) + 1 - \E\left(\frac{\E(Y_{i1}-Y_{i0}|\gamma_i,\beta_i,Y_{i0})^2}{\E((Y_{i1}-Y_{i0})^2|\gamma_i,\beta_i,Y_{i0})}\right).
    \end{aligned}
\end{displaymath}
Note that, in the last term, the denominator is larger than the numerator because
\begin{displaymath}
    \E((Y_{i1}-Y_{i0})^2|\gamma_i,\beta_i,Y_{i0}) - \E(Y_{i1}-Y_{i0}|\gamma_i,\beta_i,Y_{i0})^2 = \var(Y_{i1}-Y_{i0}|\gamma_i,\beta_i,Y_{i0}) \geq 0,
\end{displaymath}
where equality holds if $Y_{i1}-Y_{i0}$ has zero variance conditional on $(\gamma_i,\beta_i,Y_{i0})$, which is ruled out by the assumptions. Then, it follows that
\begin{displaymath}
    \E(\tilde\beta_i) = \E(\beta_i) + 1 - \E\left(\frac{\E(Y_{i1}-Y_{i0}|\gamma_i,\beta_i,Y_{i0})^2}{\E((Y_{i1}-Y_{i0})^2|\gamma_i,\beta_i,Y_{i0})}\right)
    > \E(\beta_i) + 1 - 1 = \E(\beta_i).
\end{displaymath}
This implies that $\E(\tilde\beta_i)$ is strictly larger than $\E(\beta_i)$, which completes the proof. $~\square$ \newline

\noindent\textbf{Proof of Lemma \ref{lemma.mean.failure}}.
\Cref{lemma.appendix.gmm} in Online Appendix \ref{sec.appendix.gmm} implies that, if $\E(\beta_i)$ is point-identified, then
\begin{equation}
    f^*(Y_{i0}, Y_{i1}, Y_{i2}) + g_1^*(\gamma_i, \beta_i, Y_{i0})\veps_{i1} + g_2^*(\gamma_i, \beta_i, Y_{i0}, Y_{i1})\veps_{i2} = \beta_i
    \label{eq.proof.pointid.lemma}
\end{equation}
almost surely for $(\gamma_i, \beta_i, Y_{i0}, Y_{i1}, Y_{i2}) \in \mc{C}$, where $f^*$, $g_1^*$ and $g_2^*$ are linear functionals on the spaces of finite and countably additive signed Borel measures that are absolutely continuous with respect to the Lebesgue measure. Then (\ref{eq.proof.pointid.lemma}) implies that $S^*(Y_{i0}, Y_{i1}, Y_{i2}) = f^*(Y_{i0}, Y_{i1}, Y_{i2})$ since
\begin{displaymath}
    \E(f^*(Y_{i0}, Y_{i1}, Y_{i2})|\beta_i) = \E\left( \beta_i - g_1^*(\gamma_i, \beta_i, Y_{i0})\veps_{i1} - g_2^*(\gamma_i, \beta_i, Y_{i0}, Y_{i1})\veps_{i2}| \beta_i\right) = \beta_i.
\end{displaymath}

Conversely, if there exists $S^*(Y_{i0}, Y_{i1}, Y_{i2})$ satisfying $\E(S^*(Y_{i0}, Y_{i1}, Y_{i2})|\beta_i) = \beta_i$, then $\E(S^*(Y_{i0}, Y_{i1}, Y_{i2})) = \E(\E(S^*(Y_{i0}, Y_{i1}, Y_{i2})|\beta_i)) = \E(\beta_i)$, which completes the proof. $~\square$ \newline

\noindent\textbf{Proof of Theorem \ref{prop.mean.closedform}}.
In essence, the proof follows the logic used to show the general result of \Cref{prop.gmm}. However, because \Cref{prop.gmm} has not yet been introduced and \Cref{prop.mean.closedform} does not require the regularity assumption of \Cref{prop.gmm}, I present here a standalone proof of \Cref{prop.mean.closedform} under \Cref{ass.crc,ass.mean.nomulticollinearity}.

Let $\lambda \in \mb{R}$ and let $\mu$ be a real vector that has the same dimension as $R_{it}$. Then, define the function
\begin{displaymath}
    \mc{L}(\lambda,\mu,W_i,B_i) \equiv e'B_i + \mu' \sum_{t=1}^T R_{it}(Y_{it} - R_{it}'B_i) + \lambda \sum_{t=1}^T (R_{it}'B_i)(Y_{it} - R_{it}'B_i).
\end{displaymath}
Note that $\E(\mc{L}) = \E(e'B_i)$ since
\begin{displaymath}
    \E(\mc{L}) = \E(e'B_i) + \mu' \sum_{t=1}^T \E(R_{it}(Y_{it} - R_{it}'B_i)) + \lambda \sum_{t=1}^T \E((R_{it}'B_i)(Y_{it} - R_{it}'B_i))
\end{displaymath}
and that the moment condition in (\ref{eq.meanindep}) implies that
\begin{displaymath}
    \E(R_{it}(Y_{it} - R_{it}'B_i)) = 0 \qand \E((R_{it}'B_i)(Y_{it} - R_{it}'B_i)) = 0.
\end{displaymath}
Now, note that
\begin{displaymath}
    \E(e'B_i) = \E(\mc{L}(\lambda,\mu,W_i,B_i)) \leq \E(\max_b\mc{L}(\lambda,\mu,W_i,b)),
\end{displaymath}
since $\mc{L}(\lambda,\mu,W_i,B_i) \leq \max_b\mc{L}(\lambda,\mu,W_i,b)$ almost surely. Then, since this inequality holds for all $(\lambda,\mu)$, it follows that
\begin{displaymath}
    \E(e'B_i) \leq \min_{\lambda,\mu}\E(\max_b\mc{L}(\lambda,\mu,W_i,b)),
\end{displaymath}
which implies that the right-hand side is an upper bound of $\E(e'B_i)$. Similarly, it follows that
\begin{displaymath}
    \E(e'B_i) = \E(\mc{L}(\lambda,\mu,W_i,B_i)) \geq \E(\min_b\mc{L}(\lambda,\mu,W_i,b)),
\end{displaymath}
which then implies that
\begin{displaymath}
    \E(e'B_i) \geq \max_{\lambda,\mu}\E(\min_b\mc{L}(\lambda,\mu,W_i,b)),
\end{displaymath}
which implies that the right-hand side is a lower bound of $\E(e'B_i)$. Therefore, in summary, I obtain $L^* \leq \E(e'B_i) \leq U^*$ where
\begin{displaymath}
    L^* \equiv \max_{\lambda, ~\mu} \E\left( \min_{b}\left[ e'b + \mu' \sum_{t=1}^T R_{it}(Y_{it} - R_{it}'b) + \lambda \sum_{t=1}^T (R_{it}'b)(Y_{it} - R_{it}'b) \right] \right),
\end{displaymath}
and
\begin{displaymath}
    U^* \equiv \min_{\lambda, ~\mu} \E\left( \max_{b}\left[ e'b + \mu' \sum_{t=1}^T R_{it}(Y_{it} - R_{it}'b) + \lambda \sum_{t=1}^T (R_{it}'b)(Y_{it} - R_{it}'b) \right] \right).
\end{displaymath}
In what follows, I show that the closed-form expression for $L^*$ and $U^*$ coincide with those in \Cref{prop.mean.closedform}. In the remainder of the proof, I focus on showing that the closed-form expression for $U^*$ coincides with that in \Cref{prop.mean.closedform}. The expression for $L^*$ can be verified by a similar argument.

Now I derive the closed-form expression for $U^*$. Using the matrix notations $R_i$ and $Y_i$, I can write $U^*$ concisely as
\begin{displaymath}
    U^* = \min_{\mu, \lambda} \E\left( \max_{b}\left[ e'b + \mu' R_i'Y_i - \mu'R_i'R_ib + \lambda Y_i'R_ib - b'(\lambda R_i'R_i)b \right] \right).
\end{displaymath}
Here, the inner maximization problem optimizes a quadratic polynomial in $b$, where $-\lambda R_i'R_i$ is the leading coefficient matrix. Note that if $\lambda >0$, then by \Cref{ass.mean.nomulticollinearity}, the matrix $-\lambda R_i'R_i$ is negative definite, in which case this quadratic polynomial attains a finite closed-form maximum at the solution to the first-order condition with respect to $b$. On the other hand, if $\lambda \leq 0$, the polynomial's maximum diverges to $+\infty$ unless $R_i'R_i$ has a zero eigenvalue, which is ruled out by \Cref{ass.mean.nomulticollinearity}. Consequently, a finite upper bound is obtained only for $\lambda > 0$, and I can write $U^*$ as:
\begin{displaymath}
    U^* = \min_{\mu, \lambda > 0} \E\left( \max_{b}\left[ e'b + \mu' R_i'Y_i - \mu'R_i'R_ib + \lambda Y_i'R_ib - b'(\lambda R_i'R_i)b \right] \right).
\end{displaymath}
Now, for $\lambda >0$, I can solve for the closed-form maximum of the quadratic polynomial in $b$. The first-order condition with respect to $b$ yields
\begin{displaymath}
    e - R_i'R_i\mu + \lambda R_i'Y_i - 2(\lambda R_i'R_i)b = 0 \ra b^* = \frac{1}{2}(\lambda R_i'R_i)^{-1}(e + \lambda R_i'Y_i - R_i'R_i\mu).
\end{displaymath}
Plugging this solution back into the expression for $U^*$ yields:
\begin{equation}
    \min_{\lambda > 0, ~\mu} \E\left( \mu' R_i'Y_i + \frac{1}{4\lambda}\left[ e + \lambda R_i'Y_i - R_i'R_i\mu \right]'(R_i'R_i)^{-1} \left[ e + \lambda R_i'Y_i - R_i'R_i\mu \right] \right).
    \label{eq.proof.closedform}
\end{equation}
I solve this problem with respect to $\mu$ for a fixed $\lambda$. The first-order condition with respect to $\mu$ given $\lambda$ is:
\begin{displaymath}
    \E(R_i'Y_i) + \frac{1}{2\lambda}\E(R_i'R_i)\mu - \frac{1}{2\lambda}e - \frac{1}{2}\E(R_i'Y_i) = 0.
\end{displaymath}
The optimal $\mu$ that solves this first-order condition is $\mu^* = \E(R_i'R_i)^{-1}[e - \lambda\E(R_i'Y_i)]$. I substitute this into (\ref{eq.proof.closedform}) and then solve the resulting expression with respect to $\lambda$. The first-order condition with respect to $\lambda$ is:
\begin{displaymath}
    \frac{1}{\lambda^2}\left[e'\E((R_i'R_i)^{-1})e - e'\E(R_i'R_i)^{-1}e\right] = \E(R_i'Y_i'(R_i'R_i)^{-1}R_i'Y_i) - \E(R_i'Y_i)'\E(R_i'R_i)^{-1}\E(R_i'Y_i).
\end{displaymath}
Since $\lambda > 0$, the optimal $\lambda$ that satisfies this first-order condition is:
\begin{equation}
    \lambda^* = \sqrt{\frac{e'\E((R_i'R_i)^{-1})e - e'\E(R_i'R_i)^{-1}e}{\E(R_i'Y_i'(R_i'R_i)^{-1}R_i'Y_i) - \E(R_i'Y_i)'\E(R_i'R_i)^{-1}\E(R_i'Y_i)}}.
    \label{eq.proof.closedform.optlambda}
\end{equation}
Substituting (\ref{eq.proof.closedform.optlambda}) back into (\ref{eq.proof.closedform}) yields the expression for $U$ in \Cref{prop.mean.closedform}.

The numerator and denominator in (\ref{eq.proof.closedform.optlambda}) are both weakly positive, and each is equal to zero if and only if $(R_i'R_i)^{-1}e$ and $(R_i'R_i)^{-1}R_i'Y_i$ are degenerate across individuals, respectively. To show this, one can apply the following proposition to the functions $E(R_i'R_i) = e'(R_i'R_i)^{-1}e$ and $D(R_i'Y_i, R_i'R_i) = (R_i'Y_i)'(R_i'R_i)^{-1}R_i'Y_i$. $~\square$

\begin{proposition}[\citealp{kiefer1959}, Lemma 3.2]
    \label{lemma.psd}
    For an integer $l > 0$, let $A_1, \ldots, A_l$ be $n \times m$ matrices and  $B_1, \ldots, B_l$ be nonsingular positive definite and symmetric $n \times n$ matrices. Let $a_1, \ldots, a_l$ be positive real numbers such that $\sum_k a_k = 1$. Then
    \begin{displaymath}
        \sum_{k=1}^l a_k A_k'B_k^{-1}A_k - \left[\sum_{k=1}^l a_k A_k\right]'\left[\sum_{k=1}^l a_k B_k\right]^{-1}\left[\sum_{k=1}^l a_k A_k\right] \geq 0
    \end{displaymath}
    where `$\geq$' is the partial ordering defined in terms of positive semidefinite matrices. In addition, the equality holds if and only if $B_1^{-1}A_1 = \ldots = B_l^{-1}A_l$.
\end{proposition}

\noindent\textbf{Proof of Proposition \ref{prop.mean.closedform.refined}}.
I first show that $\mc{V}_S$ is invertible, for which I show that $\mc{V}_S$ is positive definite, i.e.,
\begin{displaymath}
    x'S_iR_i(R_i'R_i)^{-1}R_i'S_i'x > 0
\end{displaymath}
with positive probability for every nonzero $x$. 

Note first that \Cref{ass.mean.nomulticollinearity.s} implies $R_i'S_i'x \neq 0$ with positive probability for every nonzero $x$. Note also that \Cref{ass.mean.nomulticollinearity} implies that $(R_i'R_i)^{-1}$ is positive definite, meaning that $y'(R_i'R_i)^{-1}y > 0$ for any nonzero vector $y$. Now, for every nonzero $x$, define $y_i = R_i'S_i'x$. Then $y_i \neq 0$ with positive probability, and it follows that
\begin{displaymath}
    x'S_iR_i(R_i'R_i)^{-1}R_i'S_i'x = y_i'(R_i'R_i)^{-1}y_i > 0
\end{displaymath}
with positive probability. This proves that $\mc{V}_S$ is positive definite, and thus invertible.

Now I prove the remainder of \Cref{prop.mean.closedform.refined}. Following the proof of \Cref{prop.mean.closedform}, let $\lambda \in \mb{R}$ and let $\mu_t$ be a real vector that has the same dimension as $S_{it}$, and define the function
\begin{displaymath}
    \begin{aligned}
        \mc{L} &\equiv e'B_i + \sum_{t=1}^T \mu_t'S_{it}(Y_{it} - R_{it}'B_i) + \lambda \sum_{t=1}^T (R_{it}'B_i)(Y_{it} - R_{it}'B_i) \\
        &= e'B_i + \mu'S_i(Y_i-R_iB_i) + \lambda B_i'R_i'(Y_i-R_iB_i),
    \end{aligned}
\end{displaymath}
where $\mu = (\mu_1', \mu_2', \ldots, \mu_T')$ is a real vector that attaches all $\mu_t$ vectors into a single vector. Then, following the proof of \Cref{prop.mean.closedform}, $L_S \leq \E(e'B_i) \leq U_S$ where
\begin{displaymath}
    L_S \equiv \max_{\mu, \lambda} \E\left( \min_{b}\left[ e'b + \mu'S_i(Y_i-R_ib) + \lambda b'R_i'(Y_i-R_ib) \right] \right),
\end{displaymath}
and
\begin{displaymath}
    U_S \equiv \min_{\mu, \lambda} \E\left( \max_{b}\left[ e'b + \mu'S_i(Y_i-R_ib) + \lambda b'R_i'(Y_i-R_ib) \right] \right).
\end{displaymath}
In what follows, I show that the closed-form expression for $L_S$ and $U_S$ coincide with those in \Cref{prop.mean.closedform.refined}. In the remainder of the proof, I focus on showing that the closed-form expression for $U_S$ coincides with that in \Cref{prop.mean.closedform.refined}. The expression for $L_S$ can be verified by a similar argument.

Now I derive the closed-form expression for $U_S$. I first rewrite $U_S$ as
\begin{displaymath}
    U_S = \min_{\mu, \lambda} \E\left( \max_{b}\left[ \mu'S_iY_i + (e + \lambda R_i'Y_i - R_i'S_i'\mu)'b - b'\left(\lambda R_i'R_i\right)b \right] \right).
\end{displaymath}
The inner maximization problem of $U_S$ optimizes a quadratic polynomial in $b$, where $-\lambda R_i'R_i$ is the leading coefficient matrix. Note that if $\lambda >0$, then by \Cref{ass.mean.nomulticollinearity}, this matrix is negative definite, in which case the quadratic polynomial attains a finite closed-form maximum. On the other hand, if $\lambda \leq 0$, the polynomial's maximum diverges to $+\infty$ unless $R_i'R_i$ has a zero eigenvalue, which is ruled out by \Cref{ass.mean.nomulticollinearity}. Consequently, a finite upper bound is obtained only for $\lambda > 0$, and I can write $U_S$ as:
\begin{displaymath}
    U_S = \min_{\mu, \lambda > 0} \E\left( \max_{b}\left[ \mu'S_iY_i + (e + \lambda R_i'Y_i - R_i'S_i'\mu)'b - b'\left(\lambda R_i'R_i\right)b \right] \right).
\end{displaymath}
Now, for $\lambda >0$, I can solve for the closed-form maximum of the quadratic polynomial in $b$. The first order condition with respect to $b$ yields
\begin{displaymath}
    e + \lambda R_i'Y_i - R_i'S_i'\mu - 2(\lambda R_i'R_i)b = 0 \ra b^* = \frac{1}{2}(\lambda R_i'R_i)^{-1}\left(e + \lambda R_i'Y_i - R_i'S_i'\mu\right).
\end{displaymath}
Plugging this solution back into the expression for $U_S$ yields:
\begin{equation}
    \min_{\lambda > 0, ~\mu} \E\left( \mu'S_iY_i + \frac{1}{4\lambda}\left[ e + \lambda R_i'Y_i - R_i'S_i'\mu \right]'(R_i'R_i)^{-1} \left[ e + \lambda R_i'Y_i - R_i'S_i'\mu \right] \right).
    \label{eq.proof.closedform.corollary}
\end{equation}
I solve this problem with respect to $\mu$ for a fixed $\lambda$. The first-order condition with respect to $\mu$ given $\lambda$ is:
\begin{displaymath}
    \E(S_iY_i) - \frac{1}{2\lambda}\E\left(S_iR_i(R_i'R_i)^{-1} \left[ e + \lambda R_i'Y_i - R_i'S_i'\mu \right]\right) = 0.
\end{displaymath}
The optimal $\mu$ that solves this first-order condition is
\begin{displaymath}
    \mu^* = \E(S_iR_i(R_i'R_i)^{-1}R_i'S_i')^{-1}\left[ \E(S_iR_i(R_i'R_i)^{-1})e + \lambda\left( \E(S_iR_i(R_i'R_i)^{-1}R_i'Y_i) - 2\E(S_iY_i) \right) \right]
\end{displaymath}
which I can write concisely as, using the notation from the main text:
\begin{displaymath}
    \mu^* = \mc{V}_S^{-1}\left[ \mc{P}_Se + \lambda\left( \mc{Y}_S - 2Y_S \right) \right].
\end{displaymath}
I now substitute this into (\ref{eq.proof.closedform.corollary}). First, expand (\ref{eq.proof.closedform.corollary}) and obtain
\begin{displaymath}
    \begin{aligned}
        \min_{\lambda > 0, ~\mu} \Big\{\mu'Y_S + \frac{1}{4\lambda}\Big[ &e'\E((R_i'R_i)^{-1})e + \lambda^2 m_0 + \mu'\mc{V}_S\mu \\
        &+ 2\lambda e'\E(\widehat{B}_i) -2e'\mc{P}_S'\mu - 2\lambda \mc{Y}_S'\mu \Big] \Big\},
    \end{aligned}
\end{displaymath}
where I used the notation from the main text to write it concisely. I then substitute the expression for $\mu^*$ into the above, obtaining:
\begin{displaymath}
    \begin{aligned}
        \min_{\lambda > 0} \Big\{Y_S'&\mc{V}_S^{-1}\left[ \mc{P}_Se + \lambda\left( \mc{Y}_S - 2Y_S \right) \right] \\
        + \frac{1}{4\lambda}\Big[ &e'\E((R_i'R_i)^{-1})e + \lambda^2 m_0 + \left[ \mc{P}_Se + \lambda\left( \mc{Y}_S - 2Y_S \right) \right]'\mc{V}_S^{-1}\left[ \mc{P}_Se + \lambda\left( \mc{Y}_S - 2Y_S \right) \right] \\
        + &2\lambda e'\E(\widehat{B}_i) -2e'\mc{P}_S'\mc{V}_S^{-1}\left[ \mc{P}_Se + \lambda\left( \mc{Y}_S - 2Y_S \right) \right] - 2\lambda \mc{Y}_S'\mc{V}_S^{-1}\left[ \mc{P}_Se + \lambda\left( \mc{Y}_S - 2Y_S \right) \right] \Big] \Big\}.
    \end{aligned}
\end{displaymath}
Expanding this expression and collecting terms with respect to $\lambda$ yields the expression:
\begin{displaymath}
    \begin{aligned}
        \min_{\lambda > 0} \Big\{ &\frac{1}{4\lambda}[e'\E((R_i'R_i)^{-1})e-e'\mc{P}_S'\mc{V}_S^{-1}\mc{P}_Se] 
        + \frac{\lambda}{4}[m_0-(\mc{Y}_S-2Y_S)'\mc{V}_S^{-1}(\mc{Y}_S-2Y_S)] \\
        &+ \frac{1}{2}[2e'\mc{P}_S'\mc{V}_S^{-1}Y_S + e'\E(\widehat{B}_i)-e'\mc{P}_S'\mc{V}_S^{-1}\mc{Y}_S]\Big\}.
    \end{aligned}
\end{displaymath}
The first order condition with respect to $\lambda$ then yields
\begin{displaymath}
    \lambda^* = \sqrt{\frac{e'\E((R_i'R_i)^{-1})e-e'\mc{P}_S'\mc{V}_S^{-1}\mc{P}_Se}{m_0-(\mc{Y}_S-2Y_S)'\mc{V}_S^{-1}(\mc{Y}_S-2Y_S)}}.
\end{displaymath}
Substituting this expression yields the expression for $U_S$ in \Cref{prop.mean.closedform.refined}. $~\square$ \newline

\noindent\textbf{Proof of Proposition \ref{prop.mean.closedform.refined.constant}}.
I first show that $\mc{V}_M - M_0$ is invertible. Note that
\begin{displaymath}
    \mc{V}_M - M_0 = \E(M_i'R_i(R_i'R_i)^{-1}R_i'M_i) - \E(M_i'M_i) = \E(M_i'(R_i(R_i'R_i)^{-1}R_i' - I)M_i).
\end{displaymath}
Since $R_i(R_i'R_i)^{-1}R_i'$ is a projection matrix, $R_i(R_i'R_i)^{-1}R_i' - I$ is negative semidefinite. For every nonzero $x$, define $y_i = M_ix$. Then, by \Cref{ass.mean.nomulticollinearity.constant}, $y_i$ is not in the column space of $R_i$ with positive probability. This implies that
\begin{displaymath}
    x'M_i'(R_i(R_i'R_i)^{-1}R_i' - I)M_ix = y_i'(R_i(R_i'R_i)^{-1}R_i' - I)y_i < 0
\end{displaymath}
with positive probability. Therefore, its expectation $\mc{V}_M - M_0$ is negative definite and thus invertible.

Next, I show that $\mc{V}$ is invertible. First, as shown in the proof of \Cref{prop.mean.closedform.refined}, $\mc{V}_S$ is positive definite under \Cref{ass.mean.nomulticollinearity.constant,ass.mean.nomulticollinearity.s.constant}. I have also shown that $\mc{V}_M - M_0$ is negative definite, which implies that $-(\mc{V}_M - M_0)$ is positive definite. Therefore, the quantity $\mc{V}_S - (C - \mc{C})(\mc{V}_M - M_0)(C - \mc{C})'$ is positive definite and thus invertible.

Now I prove the remainder of \Cref{prop.mean.closedform.refined.constant}. Following the proof of \Cref{prop.mean.closedform.refined}, let $\lambda \in \mb{R}$ and let $\mu_t$ be a real vector that has the same dimension as $S_{it}$. For a fixed value of $\delta$, define the function
\begin{displaymath}
    \begin{aligned}
        \mc{L}(\delta) &\equiv e'B_i + \sum_{t=1}^T \mu_t'S_{it}(Y_{it} - R_{it}'B_i - M_{it}'\delta) + \lambda \sum_{t=1}^T (R_{it}'B_i+M_{it}'\delta)(Y_{it} - R_{it}'B_i-M_{it}'\delta) \\
        &= e'B_i + \mu'S_i(Y_i-R_iB_i - M_i\delta) + \lambda (B_i'R_i'+\delta'M_i')(Y_i-R_iB_i-M_i\delta),
    \end{aligned}
\end{displaymath}
where $\mu = (\mu_1', \mu_2', \ldots, \mu_T')$ is a real vector that attaches all $\mu_t$ vectors into a single vector. Then, following the proof of \Cref{prop.mean.closedform}, $L_M(\delta) \leq \E(e'B_i) \leq U_M(\delta)$ where
\begin{displaymath}
    L_M(\delta) \equiv \max_{\mu, \lambda} \E\left( \min_{b}\left[ e'b + \mu'S_i(Y_i-R_ib-M_i\delta) + \lambda (b'R_i'+\delta'M_i')(Y_i-R_ib-M_i\delta) \right] \right),
\end{displaymath}
and
\begin{displaymath}
    U_M(\delta) \equiv \min_{\mu, \lambda} \E\left( \max_{b}\left[ e'b + \mu'S_i(Y_i-R_ib-M_i\delta) + \lambda (b'R_i'+\delta'M_i')(Y_i-R_ib-M_i\delta) \right] \right).
\end{displaymath}
Then, an outer bound of $\E(e'B_i)$ is obtained by taking the union of $[L_M(\delta), U_M(\delta)]$ over all $\delta \in \mb{R}^{q_m+p_m}$ (rather than over all admissible $\delta$ only). In other words, it follows that $L_M \leq \E(e'B_i) \leq U_M$ where
\begin{displaymath}
    L_M = \min_{\delta \in \mb{R}^{q_m+p_m}} L_M(\delta), \qand U_M = \max_{\delta \in \mb{R}^{q_m+p_m}}U_M(\delta).
\end{displaymath}
In what follows, I show that the closed-form expression for $L_M$ and $U_M$ coincide with those in \Cref{prop.mean.closedform.refined.constant}. In the remainder of the proof, I focus on showing that the closed-form expression for $U_M$ coincides with that in \Cref{prop.mean.closedform.refined.constant}. The expression for $L_S$ can be verified by a similar argument.

Now I derive the closed-form expression for $U_M$, for which I first derive the closed-form expression for $U_M(\delta)$. For a fixed $\delta$, 
the inner maximization problem of $U_M(\delta)$ optimizes a quadratic polynomial in $b$, where $-\lambda R_i'R_i$ is the leading coefficient matrix. Note that if $\lambda >0$, then by \Cref{ass.mean.nomulticollinearity.constant}, this matrix is negative definite, in which case the quadratic polynomial attains a finite closed-form maximum. On the other hand, if $\lambda \leq 0$, the polynomial's maximum diverges to $+\infty$ unless $R_i'R_i$ has a zero eigenvalue, which is ruled out by \Cref{ass.mean.nomulticollinearity.constant}. Consequently, a finite upper bound is obtained only for $\lambda > 0$.
Then, for $\lambda >0$, I solve for the first order condition with respect to $b$ and substitute the solution back into $U_M(\delta)$, which yields
\begin{equation}
    \begin{aligned}
        &U_M(\delta) = \min_{\lambda > 0, ~\mu} \E\Big(\mu'S_iY_i - \mu'S_iM_i\delta + \lambda\delta'M_i'Y_i - \lambda \delta'M_i'M_i\delta \\
        &\qquad + \frac{1}{4\lambda}\left[ e - R_i'S_i'\mu + \lambda R_i'Y_i - 2\lambda R_i'M_i\delta \right]'(R_i'R_i)^{-1} \left[ e - R_i'S_i'\mu + \lambda R_i'Y_i - 2\lambda R_i'M_i\delta \right] \Big).
    \end{aligned}
    \label{eq.proof.closedform.corollary.constant}
\end{equation}
I then expand the terms in (\ref{eq.proof.closedform.corollary.constant}), which yields the expression
\begin{displaymath}
    \begin{aligned}
        U_M(\delta) = \min_{\lambda > 0, ~\mu} &\Big\{\mu'Y_S - \mu'C\delta + \lambda\delta'Y_M - \lambda\delta' M_0\delta \\
        &+\frac{1}{4\lambda}e'\E((R_i'R_i)^{-1})e + \frac{\lambda}{4}m_0 + \frac{1}{4\lambda}\mu'\mc{V}_S\mu + \lambda \delta'\mc{V}_M\delta \\
        &+ \frac{1}{2} e'\E(\widehat{B}_i) -\frac{1}{2\lambda}e'\mc{P}_S'\mu - \frac{1}{2} \mc{Y}_S'\mu - \delta'\mc{P}_Me - \lambda\delta'\mc{Y}_M + \mu'\mc{C}\delta \Big\}.
    \end{aligned}
\end{displaymath}
Let $L(\lambda,\mu,\delta)$ be the objective function of the above, so that $U_M(\delta) = \min_{\lambda > 0, ~\mu}L(\lambda,\mu,\delta)$. Then I can write
\begin{displaymath}
    U_M = \max_{\delta} U_M(\delta) = \max_{\delta}\min_{\lambda > 0, ~\mu} L(\lambda,\mu,\delta).
\end{displaymath}
I note two properties about $L(\lambda,\mu,\delta)$. First, it can be shown that $L(\lambda,\mu,\delta)$ is convex in $(\lambda,\mu)$, which I state later as \Cref{prop.convex}. Second, $L(\lambda,\mu,\delta)$ is a quadratic polynomial in $\delta$, where the leading coefficient matrix is $\lambda (\mc{V}_M - M_0)$. Since $\lambda > 0$ and $\mc{V}_M - M_0$ is negative definite, this implies that $L(\lambda,\mu,\delta)$ is concave in $\delta$. These two properties then imply:
\begin{displaymath}
    U_M = \max_{\delta}\min_{\lambda > 0, ~\mu} L(\lambda,\mu,\delta) = \min_{\lambda > 0, ~\mu} \max_{\delta} L(\lambda,\mu,\delta).
\end{displaymath}
Now I solve for the latter expression. I first solve $L(\lambda,\mu,\delta)$ with respect to $\delta$ for a fixed $(\lambda, \mu)$. The optimal $\delta$ that solves the first order condition is
\begin{displaymath}
    \delta^* = -\frac{1}{2\lambda}(\mc{V}_M - M_0)^{-1}\left( \lambda (Y_M - \mc{Y}_M) - C'\mu - \mc{P}_Me + \mc{C}'\mu \right).
\end{displaymath}
I then substitute this into $L(\lambda,\mu,\delta)$, which yields the following expression for $U_M$:
\begin{displaymath}
    \begin{aligned}
        &\min_{\lambda > 0, ~\mu} \Big\{\mu'Y_S +\frac{1}{4\lambda}e'\E((R_i'R_i)^{-1})e + \frac{\lambda}{4}m_0 + \frac{1}{4\lambda}\mu'\mc{V}_S\mu + \frac{1}{2} e'\E(\widehat{B}_i) -\frac{1}{2\lambda}e'\mc{P}_S'\mu - \frac{1}{2} \mc{Y}_S'\mu - \\
        &~~\frac{1}{4\lambda} \left( \lambda (Y_M - \mc{Y}_M) - C'\mu - \mc{P}_Me + \mc{C}'\mu \right)'(\mc{V}_M - M_0)^{-1}\left( \lambda (Y_M - \mc{Y}_M) - C'\mu - \mc{P}_Me + \mc{C}'\mu \right) \Big\}.
    \end{aligned}
\end{displaymath}
Next, I solve this expression with respect to $\mu$ for a fixed $\lambda > 0$. The optimal $\mu$ that solves the first order condition is
\begin{displaymath}
    \mu^* = \mc{V}^{-1}\left[ \mc{P}_Se + \lambda\mc{Y}_S - 2\lambda Y_S + (C-\mc{C})(\mc{V}_M - M_0)^{-1}(\mc{P}_Me + \lambda \mc{Y}_M - \lambda Y_M ) \right].
\end{displaymath}
I substitute this expression back into the above expression for $U_M$, obtaining:
\begin{displaymath}
    U_M = \min_{\lambda > 0} \left( \frac{\lambda}{4}\mc{D}_M + \mc{B}_M + \frac{1}{4\lambda}\mc{E}_M \right),
\end{displaymath}
where $\mc{D}_M$, $\mc{B}_M$, and $\mc{E}_M$ are as defined in \Cref{prop.mean.closedform.refined.constant}. The first order condition with respect to $\lambda$ then yields
\begin{displaymath}
    \lambda^* = \sqrt{\frac{\mc{E}_M}{\mc{D}_M}}.
\end{displaymath}
Substituting this expression yields the expression for $U_M$ in \Cref{prop.mean.closedform.refined.constant}. $~\square$ \newline

\noindent\textbf{Proof of Theorem \ref{prop.gmm}}.
In what follows, I show that (\ref{eq.lb}) is the dual representation of (\ref{eq.primal}). The proof is a direct application of the duality theorem for linear programming over topological vector spaces \citep{anderson1983}. The same argument applies to (\ref{eq.ub}).

To apply the theorem, I first rewrite (\ref{eq.primal}) into a standard form of linear programming, for which I introduce additional notation. Recall that $\mc{M}_{W \times B}$ is a linear space of finite and countably additive signed Borel measures on $\mc{W} \times \mc{B}$. Let $\overline{\mc{F}}_{W \times B}$ be the dual space of $\mc{M}_{W \times B}$, and let $\mc{F}_{W \times B}$ be the space of all bounded Borel measurable functions on $\mc{W} \times \mc{B}$. Note that $\mc{F}_{W \times B}$ is a linear subspace of $\overline{\mc{F}}_{W \times B}$.

For $P \in \mc{M}_{W \times B}$ and $f \in \overline{\mc{F}}_{W \times B}$, define the \emph{dual pairing}
\begin{displaymath}
    \langle P, f \rangle = \int fdP.
\end{displaymath}
Similarly, let $\mc{M}_{W}$ be the linear space of finite and countably additive signed Borel measures on $\mc{W}$. Let $\overline{\mc{F}}_{W}$ be the dual space of $\mc{M}_{W}$, and let $\mc{F}_{W}$ be the space of all bounded Borel measurable functions on $\mc{W}$. Note that $\mc{F}_{W}$ is a linear subspace of $\overline{\mc{F}}_{W}$. In addition, define $\mc{G} = \mb{R}^K \times \mc{M}_{W}$ and $\mc{H} = \mb{R}^K \times \overline{\mc{F}}_{W}$, and let $g = (g_1, \ldots, g_K, P_g)$ and $h=(\lambda_1, \ldots, \lambda_K, f_h)$ be their generic elements. Note that $\mc{H}$ is the dual space of $\mc{G}$. Define the dual pairing
\begin{displaymath}
    \langle g, h \rangle = \sum_{k=1}^K \lambda_k g_k + \int f_h dP_g.
\end{displaymath}

Next, define a linear map $A: \mc{M}_{W \times B} \mapsto \mc{G}$ by
\begin{displaymath}
    A(P) = \left( \int \phi_1 dP, \ldots, \int \phi_K dP, P(\cdot, \mc{B}) \right).
\end{displaymath}
Since each $\phi_k$ is bounded, $A$ is a bounded (thus continuous) linear operator. Moreover, note that
\begin{displaymath}
    \langle A(P), h \rangle = \sum_{k=1}^K \lambda_k \int \phi_k dP  + \int_{\mc{W}} f_h(w) P(dw, \mc{B}).
\end{displaymath}
It is straightforward to show that
\begin{displaymath}
    \int_{\mc{W}} f_h(w) P(dw, \mc{B}) = \int_{\mc{W} \times \mc{B}} f_h(w) dP(w,v).
\end{displaymath}
Then I obtain
\begin{equation}
    \langle A(P), h \rangle = \sum_{k=1}^K \lambda_k \int \phi_k dP  + \int f_h dP
    = \int \left[ \sum_{k=1}^K \lambda_k \phi_k + f_h \right] dP
    \equiv \langle P, A^*(h) \rangle,
    \label{eq.proof.adjoint}
\end{equation}
where $A^*(h): \mc{H} \mapsto \overline{\mc{F}}_{W \times B}$ is defined as
\begin{displaymath}
    A^*(h) = \sum_{k=1}^K \lambda_k \phi_k + f_h.
\end{displaymath}
Equation (\ref{eq.proof.adjoint}) shows that $A^*$ is the adjoint of $A$. With these notations, I rewrite (\ref{eq.primal}) as a standard form of linear programming:
\begin{equation}
    \min_{P \in \mc{M}_{W \times V}} \langle P, m \rangle \st A(P) = c, \qquad P \geq 0,
    \label{eq.proof.primal}
\end{equation}
where $c = (0, \ldots, 0, P_W)$. I now apply the duality theorem to (\ref{eq.proof.primal}), where I verify the conditions for the duality theorem later in the proof. The dual problem of (\ref{eq.proof.primal}) is given by:
\begin{displaymath}
    \max_{h \in \mc{H}} \langle c, h \rangle \st m - A^*(h) \geq 0,
\end{displaymath}
which I can write more concretely as:
\begin{equation}
    \max_{\lambda_1, \ldots, \lambda_K \in \mb{R}, ~ f_h \in \overline{\mc{F}}_W} \int f_h dP_W \st
    \sum_{k=1}^K \lambda_k \phi_k + f_h \leq m.
    \label{eq.proof.dual}
\end{equation}
I now simplify (\ref{eq.proof.dual}) further. I rearrange the constraint of (\ref{eq.proof.dual}):
\begin{displaymath}
    f_h(w) \leq m(w,b) - \sum_{k=1}^K \lambda_k \phi_k(w,b).
\end{displaymath}
The left-hand side does not involve $v$. Therefore:
\begin{displaymath}
    f_h(w) \leq \min_{b \in \mc{B}} \left[ m(w,b) - \sum_{k=1}^K \lambda_k \phi_k(w,b) \right] \quad \text{for all} \quad w \in \mc{W}.
\end{displaymath}
Since (\ref{eq.proof.dual}) maximizes the expectation of $f_h$, the optimal solution $f_h^*$ for a fixed $(\lambda_1, \ldots, \lambda_K)$ is given by:
\begin{equation}
    f_h^*(w) = \min_{b \in \mc{B}} \left[ m(w,b) - \sum_{k=1}^K \lambda_k \phi_k(w,b) \right]
    \label{eq.proof.dual.solution}
\end{equation}
almost surely in $P_W$. If not, i.e., if $f_h^*(w)$ is strictly less than the right-hand side of (\ref{eq.proof.dual.solution}) with positive probability in $P_W$, then one can increase the value of the objective by increasing $f_h^*$ on a set of positive probability. I then substitute (\ref{eq.proof.dual.solution}) into (\ref{eq.proof.dual}), which yields:
\begin{displaymath}
    \max_{\lambda_1, \ldots, \lambda_K \in \mb{R}} \int \min_{b \in \mc{B}} \left[ m(w,b) - \sum_{k=1}^K \lambda_k \phi_k(w,b) \right] dP_W(w).
\end{displaymath}
The above display remains equivalent even if the signs of $(\lambda_1, \ldots, \lambda_K)$ are switched, because they are choice variables supported on $\mb{R}^K$. Switching the signs of $(\lambda_1, \ldots, \lambda_K)$ then gives:
\begin{equation}
    \max_{\lambda_1, \ldots, \lambda_K \in \mb{R}} \int \min_{b \in \mc{B}} \left[ m(w,b) + \sum_{k=1}^K \lambda_k \phi_k(w,b) \right] dP_W(w)
    \label{eq.proof.dual.simple}
\end{equation}
which is exactly the expression in (\ref{eq.lb}).

It remains to verify that the conditions for the duality theorem holds, so that the optimal value of the primal problem in (\ref{eq.proof.primal}) equals to that of the dual problem in (\ref{eq.proof.dual.simple}). %
Note that, by the assumptions, the dual problem in (\ref{eq.proof.dual.simple}) attains a finite maximum value, which I denote by $L$. Then, since the set $\{P \in \mc{M}_{W \times V} ~|~ P \geq 0\}$ is closed in $\mc{M}_{W \times V}$, the sufficient conditions for Theorem 4 in \citet{anderson1983} are satisfied. This theorem then implies that the primal problem in (\ref{eq.proof.primal}) attains its minimum value at $L$, and thus strong duality holds. Note that Theorem 4 in \citet{anderson1983} requires that the dual problem has a sequence of choice variables $(\lambda_1, \ldots, \lambda_K)$ whose objective function values converge to the optimal value $L$, which trivially holds if the optimal value $L$ is attained. $~\square$ \newline

\noindent\textbf{Proof of Proposition \ref{prop.convex}}.
This proof shows that $G_L$ is concave in $\lambda$. The proof for the convexity of $G_U$ is similar. Let $\lambda_1 = (\lambda_{11}, \ldots, \lambda_{1K})$ and $\lambda_2 = (\lambda_{21}, \ldots, \lambda_{2K})$ be two arbitrary points in $\mb{R}^K$. Then, for any $t \in [0,1]$ and $w \in \mc{W}$:
\begin{displaymath}
    \begin{aligned}
        & G_L(t\lambda_1+(1-t)\lambda_2, w) \\
        &= \min_{b \in \mc{B}} \left\{ t\left[m(w, b) + \sum_{k=1}^{K}\lambda_{1k} \phi_k(w, b)\right] + (1-t)\left[m(w, b) + \sum_{k=1}^{K}\lambda_{2k} \phi_k(w, b)\right] \right\} \\
        &\geq t \min_{b \in \mc{B}} \left\{m(w, b) + \sum_{k=1}^{K}\lambda_{1k} \phi_k(w, b)\right\} + (1-t) \min_{b \in \mc{B}} \left\{ m(w, b) + \sum_{k=1}^{K}\lambda_{2k} \phi_k(w, b) \right\} \\
        &= tG_L(\lambda_1, w) + (1-t)G_L(\lambda_2, w),
    \end{aligned}
\end{displaymath}
which is the definition of concavity. $~\square$ \newline

\noindent\textbf{Proof of Lemma \ref{lemma.gmm}}. 
As stated in (\ref{eq.proof.dual}) in the proof of \Cref{prop.gmm}, the sharp lower bound of $\theta$ is given by
\begin{equation}
    \max_{\lambda_1, \ldots, \lambda_K \in \mb{R}, ~ f_h \in \overline{\mc{F}}_W} \int f_h dP_W \st
    \sum_{k=1}^K \lambda_k \phi_k + f_h \leq m
    \label{eq.proof.lemma.lb}
\end{equation}
where all notations follow the proof of \Cref{prop.gmm}. Similarly, the sharp upper bound of $\theta$ is given by
\begin{equation}
    \min_{\lambda_1, \ldots, \lambda_K \in \mb{R}, ~ f_h \in \overline{\mc{F}}_W} \int f_h dP_W \st
    \sum_{k=1}^K \lambda_k \phi_k + f_h \geq m.
    \label{eq.proof.lemma.ub}
\end{equation}

Suppose, by the way of contradiction, that $\theta$ is point-identified but there do not exist $S^* \in \overline{\mc{F}}_W$ and $\lambda_1^*, \ldots, \lambda_K^* \in \mb{R}$ such that, almost surely:
\begin{displaymath}
    \sum_{k=1}^K \lambda_k^* \phi_k + S^* = m.
\end{displaymath}
Then the optimal solution $(\lambda_1^l, \ldots, \lambda_K^l, S^l)$ to (\ref{eq.proof.lemma.lb}) must satisfy its constraint $\sum_{k=1}^K \lambda_k^l \phi_k + S^l \leq m$ with strict inequality on a set of positive Lebesgue measure on $\mc{W} \times \mc{B}$. Similarly, the optimal solution $(\lambda_1^u, \ldots, \lambda_K^u, S^u)$ to (\ref{eq.proof.lemma.ub}) must satisfy its constraint $\sum_{k=1}^K \lambda_k^u \phi_k + S^u \geq m$ with strict inequality on a set of positive Lebesgue measure on $\mc{W} \times \mc{B}$. Then it follows that:
\begin{displaymath}
    \E(S^l) = \E\left( \sum_{k=1}^K \lambda_k^l \phi_k + S^l \right)  < \E(m) < \E\left( \sum_{k=1}^K \lambda_k^u \phi_k + S^u \right)  = \E(S^u)
\end{displaymath}
where strict inequalities follow because the density of $(W_i, B_i)$ is strictly positive. This implies that the sharp lower bound $\E(S^l)$ is strictly less than the sharp upper bound $\E(S^u)$, which is a contradiction since $\theta$ is assumed to be point-identified.

Conversely, suppose there exists $(S^*, \lambda_1^*, \ldots, \lambda_K^*)$ such that $\sum_{k=1}^K \lambda_k^* \phi_k + S^* = m$. Then:
\begin{displaymath}
    \E(S^*) = \E\left(\sum_{k=1}^K \lambda_k \phi_k + S^*\right) = \E(m) = \theta,
\end{displaymath}
which shows that $\theta$ is point-identified by $\E(S^*)$. $~\square$ \newline

\section{Online Appendix: Extensions and Discussions}

\subsection{Extension to multivariate random coefficient models}

\label{sec.appendix.multivar}

Results from this paper extend to a multivariate version of (\ref{eq.crc}), a system of random coefficient models:
\begin{displaymath}
    \mf{Y}_{it} = \mf{Z}_{it}'\gamma_i + \mf{X}_{it}'\beta_i + \mf{e}_{it},
\end{displaymath}
where $\mf{Y}_{it}$ is a $D \times 1$ vector of dependent variables, $\mf{Z}_{it}$ is a $D \times q$ matrix of strictly exogenous regressors, $\mf{X}_{it}$ is a $D \times p$ matrix of sequentially exogenous regressors, and $\mf{e}_{it}$ is a $D \times 1$ vector of idiosyncratic error terms. Assume that
\begin{displaymath}
    \E(\mf{e}_{it}|\gamma_i, \beta_i, \mf{Z}_{i}, \mf{X}_{i}^t) = 0,
\end{displaymath}
which is a multivariate extension of (\ref{eq.meanindep}). The following is an example of a multivariate random coefficient model.

\begin{example}[Joint model of household earnings and consumption behavior]

    \label{sec.ex.joint}

    One can combine (\ref{eq.income}) and (\ref{eq.consumption}) to construct a joint lifecycle model of earnings and consumption behavior. In particular, when the time $t$ consumption equation is combined with the time $t+1$ earnings equation, a system of random coefficient models is obtained:
    \begin{displaymath}
        \begin{aligned}
            C_{it} &= \gamma_{i1} + \gamma_{i2}Y_{it} + \beta_{i1}A_{it} + \nu_{it}, \\
            Y_{i,t+1} &= \gamma_{i3} + \beta_{i2}Y_{it} + \veps_{it},
        \end{aligned}
    \end{displaymath}
    which can be written in the following matrix form:
    \begin{displaymath}
        \left(\begin{array}{c}
            C_{it} \\
            Y_{i,t+1}
        \end{array}\right)
        = \left(\begin{array}{ccc}
            1 & Y_{it} & 0 \\
            0 & 0 & 1
        \end{array}\right)\left(\begin{array}{c}
            \gamma_{i1} \\ \gamma_{i2} \\ \gamma_{i3}
        \end{array}\right) + 
        \left(\begin{array}{cc}
            A_{it} & 0 \\
            0 & Y_{it}
        \end{array}\right)\left(\begin{array}{c}
            \beta_{i1} \\ \beta_{i2}
        \end{array}\right)
        + \left(\begin{array}{c}
            \nu_{it} \\ \veps_{it}
        \end{array}\right).
    \end{displaymath}
    In this model, the $\gamma_i$s and $\beta_i$s can freely correlate among themselves and with $(Y_{i0}, A_{i1})$, allowing for correlation between earnings and consumption processes.

\end{example}

\subsection{Alternative proof of Proposition \ref{prop.mean.failure}}

\label{sec.appendix.proof}

This proof is an application of the general result in Online Appendix \ref{sec.appendix.gmm}. Suppose that the regularity conditions stated as \Cref{ass.appendix.regularity} in Online Appendix \ref{sec.appendix.gmm} hold. Also, for notational simplicity, suppose that $\mc{C} = \mc{C}_0^5$, where $\mc{C}_0$ is a compact subset of $\mb{R}$. The proof can be easily modified for a generic compact set $\mc{C}$.

Assume that $\E(\beta_i)$ is point-identified, from which I derive a contradiction. \Cref{lemma.appendix.gmm} in Online Appendix \ref{sec.appendix.gmm} implies that, if $\E(\beta_i)$ is point-identified, then:
\begin{equation}
    f^*(Y_{i0}, Y_{i1}, Y_{i2}) + g_1^*(\gamma_i, \beta_i, Y_{i0})\veps_{i1} + g_2^*(\gamma_i, \beta_i, Y_{i0}, Y_{i1})\veps_{i2} = \beta_i
    \label{eq.proof.pointid}
\end{equation}
almost surely in $(\gamma_i, \beta_i, Y_{i0}, Y_{i1}, Y_{i2})$, where $f^* : \mc{C}_0^3 \mapsto \mb{R}$, $g_1^* : \mc{C}_0^3 \mapsto \mb{R}$ and $g_2^* : \mc{C}_0^4 \mapsto \mb{R}$ are linear functionals on the spaces of finite and countably additive signed Borel measures that are absolutely continuous with respect to the Lebesgue measure. Substituting $\veps_{it} = Y_{it} - \gamma_i - \beta_iY_{i,t-1}$ in (\ref{eq.proof.pointid}) yields, almost surely in $(\gamma_i, \beta_i, Y_{i0}, Y_{i1}, Y_{i2})$,
\begin{equation}
    f^*(Y_{i0}, Y_{i1}, Y_{i2}) + g_1^*(\gamma_i, \beta_i, Y_{i0})(Y_{i1}-\gamma_i-\beta_i Y_{i0}) + g_2^*(\gamma_i, \beta_i, Y_{i0}, Y_{i1})(Y_{i2}-\gamma_i-\beta_iY_{i1}) = \beta_i.
    \label{eq.proof.pointid.substitute}
\end{equation}

Now, consider any $\gamma, \tilde\gamma, \beta, y_0, y_1, y_2 \in \mc{C}_0$ such that $\gamma \neq \tilde\gamma$. I then evaluate (\ref{eq.proof.pointid.substitute}) at $(\gamma_i, \beta_i, Y_{i0}, Y_{i1}, Y_{i2}) = (\gamma, \beta, y_0, y_1, y_2)$ and at $(\tilde\gamma, \beta, y_0, y_1, y_2)$ and take the difference between the two. This yields:
\begin{equation}
    \begin{aligned}
        & (y_1-\tilde\gamma-\beta y_0)\triangle_{\tilde\gamma,\gamma} g_1^* - (\tilde\gamma-\gamma)g_1^*(\gamma, \beta, y_0) \\
        &+ (y_2-\tilde\gamma-\beta y_1)\triangle_{\tilde\gamma,\gamma} g_2^* - (\tilde\gamma-\gamma)g_2^*(\gamma, \beta, y_0, y_1) = 0
    \end{aligned}
    \label{eq.proof.diff1}
\end{equation}
where $\triangle_{\tilde\gamma,\gamma} g_1^* \equiv g_1^*(\tilde \gamma, \beta, y_0) - g_1^*(\gamma, \beta, y_0)$ and $\triangle_{\tilde\gamma,\gamma} g_2^* \equiv g_2^*(\tilde \gamma, \beta, y_0, y_1) - g_2^*(\gamma, \beta, y_0, y_1)$.

In (\ref{eq.proof.diff1}), note that the variable $y_2$ appears only in the third term. In addition, (\ref{eq.proof.diff1}) must hold almost surely for all $\gamma, \tilde\gamma, \beta, y_0, y_1, y_2 \in \mc{C}_0$ such that $\gamma \neq \tilde\gamma$, and in particular for any choice of $y_2 \in \mc{C}_0$. This implies that the third term must not depend on $y_2$, which implies that, almost surely:
\begin{displaymath}
    \triangle_{\tilde\gamma,\gamma} g_2^* = 0,
\end{displaymath}
which means that $g_2^*$ is almost surely a constant function on $\gamma$:
\begin{equation}
    g_2^*(\gamma, \beta, y_0, y_1) = g_2^*(\beta, y_0, y_1).
    \label{eq.proof.diff1.sol.conclusion}
\end{equation}
If not, i.e., if $\triangle_{\tilde\gamma,\gamma} g_2^* \neq 0$ on a subset of $\mc{C}_0^6$ with positive Lebesgue measure, one can change the value of $y_2$ without altering $(\gamma, \tilde\gamma, \beta, y_0, y_1)$ within that subset to violate (\ref{eq.proof.diff1}) on a set of a positive measure.

Next, consider any $\gamma, \beta, \tilde\beta, y_0, y_1, y_2 \in \mc{C}_0$ such that $\beta \neq \tilde\beta$. I then evaluate (\ref{eq.proof.pointid.substitute}) at $(\gamma, \beta, y_0, y_1, y_2)$ and $(\gamma, \tilde\beta, y_0, y_1, y_2)$ and take the difference between the two. This yields:
\begin{equation}
    \begin{aligned}
        & (y_1-\gamma-\tilde\beta y_0)\triangle_{\tilde\beta,\beta} g_1^* - (\tilde\beta-\beta)y_0g_1^*(\gamma, \beta, y_0) \\
        &+ (y_2-\gamma-\tilde\beta y_1)\triangle_{\tilde\beta,\beta} g_2^* - (\tilde\beta-\beta)y_1g_2^*(\beta, y_0, y_1) = \tilde\beta - \beta
    \end{aligned}
    \label{eq.proof.diff2}
\end{equation}
where $\triangle_{\tilde\beta,\beta} g_1^* \equiv g_1^*(\gamma, \tilde\beta, y_0) - g_1^*(\gamma, \beta, y_0)$ and $\triangle_{\tilde\beta,\beta} g_2^* = g_2^*(\tilde\beta, y_0, y_1) - g_2^*(\beta, y_0, y_1)$. In (\ref{eq.proof.diff2}), note that $y_2$ appears only in the third term. This implies $g_2^*(\beta, y_0, y_1) = g_2^*(y_0, y_1)$ almost surely, similarly to the argument leading to (\ref{eq.proof.diff1.sol.conclusion}). Then (\ref{eq.proof.diff1}) simplifies to:
\begin{equation}
    (y_1-\tilde\gamma-\beta y_0)\triangle_{\tilde\gamma,\gamma} g_1^* - (\tilde\gamma-\gamma)g_1^*(\gamma, \beta, y_0) - (\tilde\gamma-\gamma)g_2^*(y_0, y_1) = 0.
    \label{eq.proof.diff1.simple}
\end{equation}
Let $\gamma, \tilde\gamma, \hat\gamma \in \mc{C}_0$ be such that $\hat\gamma - \tilde\gamma = \tilde\gamma - \gamma$. I then evaluate (\ref{eq.proof.diff1.simple}) at $(\gamma, \tilde\gamma, \beta, y_0, y_1)$ and $(\tilde\gamma, \hat\gamma, \beta, y_0, y_1)$ and take the difference between the two. This yields:
\begin{equation}
    (y_1-\hat\gamma-\beta y_0)\left(\triangle_{\hat\gamma,\tilde\gamma} g_1^* - \triangle_{\tilde\gamma,\gamma} g_1^*\right) - (\hat\gamma - \tilde\gamma)\triangle_{\tilde\gamma,\gamma} g_1^* - (\tilde\gamma-\gamma)\triangle_{\tilde\gamma,\gamma} g_1^* = 0.
    \label{eq.proof.diff11}
\end{equation}
In (\ref{eq.proof.diff11}), note that $y_1$ appears only in the first term, which implies $\triangle_{\hat\gamma,\tilde\gamma} g_1^* - \triangle_{\tilde\gamma,\gamma} g_1^* = 0$ almost surely, similarly to the argument leading to (\ref{eq.proof.diff1.sol.conclusion}). Then (\ref{eq.proof.diff11}) simplifies to:
\begin{equation}
    (\hat\gamma - \tilde\gamma)\triangle_{\tilde\gamma,\gamma} g_1^* + (\tilde\gamma-\gamma)\triangle_{\tilde\gamma,\gamma} g_1^* = 0,
    \label{eq.proof.diff11.conclusion}
\end{equation}
which implies $\triangle_{\tilde\gamma,\gamma} g_1^* = 0$ since $\hat\gamma - \tilde\gamma = \tilde\gamma - \gamma \neq 0$. This implies that $g_1^*$ is almost surely a constant function over $\gamma$, i.e., $g_1^*(\gamma, \beta, y_0) = g_1^*(\beta, y_0)$. Then (\ref{eq.proof.diff2}) simplifies to:
\begin{equation}
    \begin{aligned}
        & (y_1-\gamma-\tilde\beta y_0)\triangle_{\tilde\beta,\beta} g_1^* - (\tilde\beta-\beta)y_0g_1^*(\beta, y_0) \\
        &+ (y_2-\gamma-\tilde\beta y_1)\triangle_{\tilde\beta,\beta} g_2^* - (\tilde\beta-\beta)y_1g_2^*(y_0, y_1) = \tilde\beta - \beta.
    \end{aligned}
    \label{eq.proof.diff2.simple}
\end{equation}
Let $\beta, \tilde\beta, \hat\beta \in \mc{C}_0$ be such that $\hat\beta - \tilde\beta = \tilde\beta - \beta$. Evaluating (\ref{eq.proof.diff2.simple}) at $(\gamma, \hat\beta, \tilde\beta, y_0, y_1, y_2)$ and at $(\gamma, \tilde\beta, \beta, y_0, y_1, y_2)$ and taking the difference yields $g_1^*(\beta, y_0) = g_1^*(y_0)$, similarly to the argument leading to $g_1^*(\gamma, \beta, y_0) = g_1^*(\beta, y_0)$ from (\ref{eq.proof.diff11}). Then (\ref{eq.proof.pointid.substitute}) simplifies to:
\begin{displaymath}
    f^*(y_0, y_1, y_2) + g_1^*(y_0)(y_1-\gamma-\beta y_0) + g_2^*(y_0, y_1)(y_2-\gamma-\beta y_1) = \beta
\end{displaymath}
almost surely for all $(\gamma,\beta,y_0,y_1,y_2)$. This is a linear identity in $(\gamma, \beta)$, so the coefficients of $\gamma$ and $\beta$ must match on both sides of the equation. In particular, equating the coefficients on $\gamma$ implies that $-g_1^* - g_2^* = 0$, and equating the coefficients on $\beta$ implies $-y_0 g_1^* - y_1 g_2^* = 1$. Solving these two equations for $(g_1^*, g_2^*)$ yields, almost surely:
\begin{displaymath}
    g_1^* = \frac{1}{y_1-y_0}, \quad g_2^* = \frac{-1}{y_1-y_0}.
\end{displaymath}
However, $g_1^*$ cannot be a function of $y_1$, which is a contradiction. $~\square$

\subsection{Identification under conditional moment restrictions}

\label{sec.appendix.gmm}

This subsection studies moment equality models that incorporate both conditional and unconditional moment restrictions. Consider the following extension of \Cref{ass.gmm}.

\begin{assumption}
    \label{ass.appendix.gmm}
    The random vectors $(W_i, B_i)$ satisfy:
    \begin{displaymath}
        \begin{aligned}
            \E(\phi_k(W_i, B_i)) &= 0, \quad k=1, \ldots, K_U, \\
            \E(\psi_k(W_i, B_i) | A_{ik}) &= 0, \quad k=1, \ldots, K_C,
        \end{aligned}
    \end{displaymath}
    where $\phi_k, \psi_k : \mc{W} \times \mc{B} \mapsto \mb{R}$ are moment functions, $A_{i1}, \ldots, A_{iK_C}$ are subvectors of $(W_i, B_i)$, and $K_U,K_C \in \mb{N}$ are the number of unconditional and conditional moment restrictions, respectively.
\end{assumption}

Under \Cref{ass.appendix.gmm}, I characterize the identified set of
\begin{displaymath}
    \theta = \E(m(W_i, B_i))
\end{displaymath}
for some known function $m:\mc{W} \times \mc{B} \mapsto \mb{R}$. For brevity of notation, let $A_{ik}'$ be the subvector of $(W_i, B_i)$ collecting the variables not included in $A_{ik}$ so that $(A_{ik}, A_{ik}')$ is a rearrangement of $(W_i, B_i)$. Accordingly, any function $f(w,b)$ on $\mc{W} \times \mc{B}$ can be equivalently written as $f(a_k, a_k')$ on $\mc{A}_k \times \mc{A}_k'$, where $\mc{A}_k \times \mc{A}_k'$ is the rearrangement of $\mc{W} \times \mc{B}$ according to $(A_{ik}, A_{ik}')$. I assume the following regularity conditions.

\begin{assumption}
    \label{ass.appendix.regularity}
    The following conditions hold.
    \begin{itemize}
        \item[(i)] $\mc{W} \times \mc{B}$ is a compact set in a Euclidean space. 
        \item[(ii)] $(m, \phi_1, \ldots, \phi_{K_U}, \psi_1, \ldots, \psi_{K_C})$ are $L^\infty$ with respect to the Lebesgue measure.
        \item[(iii)] The distribution of $(W_i, B_i)$ is absolutely continuous with respect to the Lebesgue measure, and its density $p$ is $L^\infty$ with respect to the Lebesgue measure.
        \item[(iv)] There exists a joint density $p_0$ of $(W_i, B_i)$ such that it satisfies \Cref{ass.appendix.gmm}, its marginal density with respect to $W_i$ equals to the observed density of $W_i$, and it is strictly positive on $\mc{W} \times \mc{B}$.
    \end{itemize}
\end{assumption}

\Cref{ass.appendix.regularity} (i) and (ii) are similar to \Cref{ass.proof.regularity}. \Cref{ass.appendix.regularity} (iii) and (iv) are restrictive, but they are useful enough for \Cref{lemma.mean.failure} and an alternative proof of \Cref{prop.mean.failure} in Online Appendix \ref{sec.appendix.proof}. %

Under these assumptions, I obtain the following theorem and lemma, which are counterparts of \Cref{prop.gmm} and \Cref{lemma.gmm}, respectively, by characterizing the identified set $I$ of $\theta$ and providing a necessary and sufficient condition for point identification of $\theta$.

\begin{theorem}
    \label{prop.appendix.gmm}
    Suppose that either of these two conditions hold: 
    \begin{enumerate}
        \item[(A)] \Cref{ass.appendix.gmm,ass.appendix.regularity} (i)-(iii) hold, and the optimization problems
        \begin{equation}
            L = \max_{\{\lambda_k\}_{k=1}^{K_U}, \{\mu_k\}_{k=1}^{K_C}}
            \E\left[ \min_{b \in \mc{B}} \left\{ m(W_i, b) + \sum_{k=1}^{K_U}\lambda_k \phi_k(W_i, b) + \sum_{k=1}^{K_C}\mu_k(A_k(W_i, b))\psi_k(W_i, b)\right\} \right]
            \label{eq.appendix.lb}
        \end{equation}
        and
        \begin{equation}
            U = \min_{\{\lambda_k\}_{k=1}^{K_U}, \{\mu_k\}_{k=1}^{K_C}}
            \E\left[ \max_{b \in \mc{B}} \left\{ m(W_i, b) + \sum_{k=1}^{K_U}\lambda_k \phi_k(W_i, b) + \sum_{k=1}^{K_C}\mu_k(A_k(W_i, b))\psi_k(W_i, b)\right\} \right]
            \label{eq.appendix.ub}
        \end{equation}
        possess finite solutions, where $\lambda_k \in \mb{R}$ for $k=1, \ldots, K_U$ and $\mu_k : \mc{A}_k \mapsto \mb{R}$ for $k=1, \ldots, K_C$, and $A_k(w,b)$ denotes the value of $A_{ik}$ given $W_i =w$ and $B_i = b$.
        \item[(B)] \Cref{ass.appendix.gmm,ass.appendix.regularity} (i)-(iv) hold.
    \end{enumerate}
    If either (A) or (B) holds, then $I = [L,U]$.
\end{theorem}

\begin{proof}

    The proof focuses on showing (\ref{eq.appendix.lb}). The same argument applies to (\ref{eq.appendix.ub}). 
    
    Let $\mc{M}_{W \times B}$ be the space of finite and countably additive signed Borel measures that are absolutely continuous with respect to the Lebesgue measure. Using absolute continuity, I identify each element of $\mc{M}_{W \times B}$ by its density $p : \mc{W} \times \mc{B} \mapsto \mb{R}$. Let $p_W$ be the density of the observed data distribution $P_W$. Then, the identified set $I$ is defined by
    \begin{displaymath}
        \begin{aligned}
            I \equiv \left\{ \int m(w, b) p(w,b) \right.& d(w,b) ~\Big|~ p \in \mc{M}_{W \times B}, \quad p \geq 0, \\
            & \int \phi_k(w, b) p(w,b)d(w,b) = 0, \quad k=1, \ldots, K_U, \\
            & \int \psi_k(a_k, a_k') p(a_k, a_k') da_k' = 0 ~\text{for all } a_k \in \mc{A}_k, \quad k=1, \ldots, K_C,\\
            & \left. \int p(w, b)db = p_W(w) ~\text{ for all } w \in \mc{W} \right\},
        \end{aligned}
    \end{displaymath}
    where $a_k$ is an element of $\mc{A}_k$ and $a_k'$ is an element of $\mc{A}_k'$. The second line above represents the unconditional moment restrictions, while the third line represents the conditional moment restrictions.
    
    The lower bound of $I$ is then given by the infinite-dimensional linear program
    \begin{equation}
        \begin{aligned}
            \min_{p \in \mc{M}_{W \times B}, ~p \geq 0} \int & m(w, b) p(w,b) d(w,b) \st \\
            & \int \phi_k(w, b) p(w,b) d(w,b) = 0, \quad k=1, \ldots, K_U, \\
            & \int \psi_k(a_k, a_k') p(a_k, a_k') da_k' = 0, ~\text{for all } a_k \in \mc{A}_k, \quad k=1, \ldots, K_C, \\
            & \int p(w, b)db = p_W(w) ~\text{ for all } w \in \mc{W}.
        \end{aligned}
        \label{eq.appendix.primal}
    \end{equation}
    
    Now I show that (\ref{eq.appendix.lb}) is the dual representation of (\ref{eq.appendix.primal}), by directly applying the duality theorem of linear programming for topological vector spaces \citep{anderson1983}, for which I introduce additional notation. Let $L^2(\mc{W} \times \mc{B})$ be the space of all $L^2$ functions on $\mc{W} \times \mc{B}$, and let $L^2(\mc{W})$ be the space of all $L^2$ functions on $\mc{W}$. I also let $L^2(\mc{A}_k)$ be the space of all $L^2$ functions on $\mc{A}_k$.
    
    Define $\mc{G}$ and $\mc{H}$ as $\mc{G} = \mc{H} = \mb{R}^K \times L^2(\mc{A}_1) \times \ldots \times L^2(\mc{A}_{K_C}) \times L^2(\mc{W})$. I denote their generic elements as $g = (g_1, \ldots, g_{K_U}, \overline g_1, \ldots, \overline g_{K_C}, f_g)$ and $h=(\lambda_1, \ldots, \lambda_{K_U}, \mu_1, \ldots, \mu_{K_C}, f_h)$, respectively. Note that $\mc{H}$ is a dual space of $\mc{G}$.
    
    Define a linear map $A: \mc{M}_{W \times B} \mapsto \mc{G}$ by
    \begin{displaymath}
        A(p) = \left( \int \phi_1 p~d(w,b), \ldots, \int \phi_K p~d(w,b), \int \psi_k p~da_1', \ldots, \int \psi_k p~da_{K_C}', \int p~db \right).
    \end{displaymath}
    The map $A$ is a bounded (thus continuous) linear operator because the functions $\phi_k$ and $\psi_k$ are assumed to be bounded. Next, define the dual pairing as
    \begin{displaymath}
        \langle A(P), h \rangle = \sum_{k=1}^{K_U} \lambda_k \int \phi_k p~d(w,vb)  + \sum_{k=1}^{K_C} \iint \psi_k p ~da_k'~ \mu_k da_k + \int f_h \int p db dw.
    \end{displaymath}
    It is straightforward to show that
    \begin{displaymath}
        \iint \psi_k p ~da_k'~ \mu_k da_k = \int \psi_k \mu_k p ~ d(w,b)
    \end{displaymath}
    and
    \begin{displaymath}
        \int f_h \int p db dw = \int f_h p ~ d(w,b).
    \end{displaymath}
    Then, I obtain
    \begin{equation}
        \langle A(P), h \rangle = \int \left[\sum_{k=1}^{K_U} \lambda_k \phi_k  + \sum_{k=1}^{K_C} \mu_k \psi_k + f_h \right] p(w,b) d(w,b)
        \equiv \langle p, A^*(h) \rangle,
        \label{eq.appendix.adjoint}
    \end{equation}
    where $A^*(h): \mc{H} \mapsto L^2(\mc{W} \times \mc{B})$ is defined as
    \begin{displaymath}
        A^*(h) = \sum_{k=1}^{K_U} \lambda_k \phi_k  + \sum_{k=1}^{K_C} \mu_k \psi_k + f_h.
    \end{displaymath}
    Equation (\ref{eq.appendix.adjoint}) shows that $A^*$ is the adjoint of $A$.
    
    Then, similar to the proof of \Cref{prop.gmm}, I apply the strong duality theorem of linear programming for topological vector spaces to (\ref{eq.appendix.primal}). Under the conditions of (A), I apply Theorem 4 of \citet{anderson1983}, similarly to the proof of \Cref{prop.gmm}. Under the conditions of (B), I use Theorem 9 of \citet{anderson1983}. The sufficient conditions for this Theorem 9 are satisfied as follows. First, \Cref{ass.appendix.regularity} (iv) verifies the condition that ``there is $x_0$ in the interior of $P$ with $Ax_0 = b$'' in Theorem 9 of \citet{anderson1983}. Second, \Cref{ass.appendix.regularity} (i)-(iii) ensures that the primal problem in (\ref{eq.appendix.primal}) possesses a finite solution, which verifies the condition that ``EP has finite value'' in Theorem 9 of \citet{anderson1983}. In addition, its implicit condition that $A$ is continuous is also satisfied. 
    
    Consequently, under either conditions of (A) or (B), the strong duality holds. Therefore, the optimal solution to (\ref{eq.appendix.primal}) is equal to the solution to
    \begin{equation}
        \max_{\lambda_1, \ldots, \lambda_{K_U}, \mu_1, \ldots, \mu_{K_C}, f_h} \int f_h(w) p_w(w) dw \st \sum_{k=1}^{K_U} \lambda_k \phi_k  + \sum_{k=1}^{K_C} \mu_k \psi_k + f_h \leq m.
        \label{eq.appendix.dual}
    \end{equation}
    Then, simplifying (\ref{eq.appendix.dual}) as in the proof of \Cref{prop.gmm} yields the expression in (\ref{eq.appendix.lb}).
    
\end{proof}

\begin{lemma}
    \label{lemma.appendix.gmm}
    Suppose that the conditions (B) of \Cref{prop.appendix.gmm} hold. %
    Then $\theta$ is point-identified if and only if there exists a function $S^*$ which is a linear functional on $\mc{M}_{W}$ (which is the projection of $\mc{M}_{W \times B}$ onto $\mc{W}$), real numbers $\lambda_1^*, \ldots, \lambda_{K}^* \in \mb{R}$, and functions $\mu_1^*, \ldots, \mu_K^*$ which are $L^2(\mc{A}_1), \ldots, L^2(\mc{A}_{K_C})$ functions, respectively, such that:
    \begin{displaymath}
        m(W_i, b) + \sum_{k=1}^{K_U}\lambda_k \phi_k(W_i, b) + \sum_{k=1}^{K_C}\mu_k(A_k(W_i, b))\psi_k(W_i, b) = S^*(W_i)
    \end{displaymath}
    almost surely on $\mc{W} \times \mc{B}$. When such $S^*$ exists, $\theta$ is identified by $\theta = \E(S^*(W_i))$.
\end{lemma}

\begin{proof}

    As in (\ref{eq.appendix.dual}) in the proof of \Cref{prop.appendix.gmm}, the sharp lower bound of $\theta$ is given by
    \begin{displaymath}
        \max_{\lambda_1, \ldots, \lambda_{K_U}, \mu_1, \ldots, \mu_{K_C}, f_h} \int f_h(w) p_w(w) dw \st \sum_{k=1}^{K_U} \lambda_k \phi_k  + \sum_{k=1}^{K_C} \mu_k \psi_k + f_h \leq m,
    \end{displaymath}
    where all notation follows the proof of \Cref{prop.appendix.gmm}. Similarly, the sharp upper bound of $\theta$ is given by
    \begin{displaymath}
        \min_{\lambda_1, \ldots, \lambda_{K_U}, \mu_1, \ldots, \mu_{K_C}, f_h} \int f_h(w) p_w(w) dw \st \sum_{k=1}^{K_U} \lambda_k \phi_k  + \sum_{k=1}^{K_C} \mu_k \psi_k + f_h \geq m.
    \end{displaymath}
    \Cref{lemma.appendix.gmm} then follows by replicating the argument in the proof of \Cref{lemma.gmm}.
    
\end{proof}

\subsection{Identified set for a general variance parameter}

\label{sec.appendix.var}

In this subsection, I consider identification of a general variance parameter. Recall the dynamic random coefficient model defined in (\ref{eq.crc}) and (\ref{eq.meanindep}):
\begin{displaymath}
    Y_{it} = R_{it}'B_i + \veps_{it}, \qquad \E(\veps_{it}|B_i, Z_i, X_i^t) = 0, \qquad t=1, \ldots, T,
\end{displaymath}
where $R_{it} = (Z_{it}', X_{it}')'$. I consider the second moments of the random coefficients:
\begin{displaymath}
    V_e = \E(e_1'B_iB_i'e_2)
\end{displaymath}
where $e_1$ and $e_2$ are real-valued vectors that the researcher chooses. For example, if $e_1 = e_2 = (1,0,\ldots,0)'$, then $V_e$ is the second moment of the first entry of $B_i$.

The key idea for identifying the mean in \Cref{sec.mean} is to consider a moment condition that is quadratic in $B_i$ so that it can ``dominate'' the linear term $e'B_i$. By the same idea, to ``dominate'' the term $e_1'B_iB_i'e_2$, I consider a moment condition that is fourth order in $B_i$. Specifically, I consider the moment restrictions
\begin{equation}
    \E\left(\sum_{t=1}^T (R_{it}'B_i)^3\veps_{it}\right) = 0, \qand \E\left(S_{it}\veps_{it}\right) = 0 \quad\text{for}\quad t=1, \ldots, T,
    \label{eq.orthogonal.variance}
\end{equation}
where I replaced the first moment restriction in (\ref{eq.orthogonal.closedform.refined}) with a fourth order restriction in $B_i$. A direct application of \Cref{prop.gmm} then yields the following bounds for $V_e$, which I state without proof.

\begin{proposition}
    \label{prop.variance}
    Suppose that \Cref{ass.crc,ass.mean.nomulticollinearity,ass.mean.nomulticollinearity.s,ass.proof.regularity} hold, and that Equation (\ref{eq.orthogonal.variance}) holds. Then $L_V \leq V_e \leq U_V$ where
    \begin{displaymath}
        L_V = \max_{\lambda > 0, ~\mu \in \mb{R}^L} 
        \E\left[ \min_{b \in \mc{B}} \left\{ e_1'bb'e_2 + \lambda \sum_{t=1}^T (R_{it}'b)^3(Y_{it} - R_{it}'b) + \mu'S_i(Y_i - R_ib)  \right\} \right],
    \end{displaymath}
    and
    \begin{displaymath}
        U_V = \min_{\lambda < 0, ~\mu \in \mb{R}^L}
        \E\left[ \max_{b \in \mc{B}} \left\{ e_1'bb'e_2 + \lambda \sum_{t=1}^T (R_{it}'b)^3(Y_{it} - R_{it}'b) + \mu'S_i(Y_i - R_ib) \right\} \right].
    \end{displaymath}
\end{proposition}

In \Cref{prop.variance}, the optimization over $b$ requires global optimization of a fourth-order polynomial, for which no closed-form solution exists. Instead, this problem can be solved numerically using the semidefinite relaxation approach \citep{lasserre2010,lasserre2015}, which transforms the polynomial optimization problem into a convex optimization problem over a variable that is a semidefinite matrix.

\subsection{Overidentification in the general inference procedure}

\label{sec.inference.general.overid}

In this subsection, I discuss a simple heuristic modification of the general inference procedure in \Cref{sec.inference.general} that practically deals with the overidentification, assuming that the model is correctly specified. Note that this heuristic method may yield spuriously narrow confidence intervals for $\theta$ even if the model is partially identified. This occurs because the procedure does not formally account for overidentification or misspecification as in \citet{stoye2020} and \citet{andrews2024}. %
Nonetheless, I check the performance of the heuristic procedure via simulation at the end of this subsection, and I find that it achieves reasonable coverage rates.

I first note that the inference procedure of \citet{andrews2017} is still implementable under overidentification, since $T_{AS}(\theta)$ and $c_{AS,GMS}^{(b)}(\theta)$ can still be computed for each fixed $\theta$. The practical difficulty, however, is that the %
finite-sample optimizers $\hat \lambda_N^L$ and $\hat \lambda_N^U$ diverge, making it hard to search for the supremum in the neighborhoods of $\hat \lambda_N^L$ and $\hat \lambda_N^U$. To resolve this issue, I propose to compute the finite-sample optimizers with $L^1$ penalties:
\begin{equation}
    \begin{aligned}
        \tilde \lambda_N^L(\zeta) &= \max_{\lambda \in \mb{R}^K} \left[
    \frac{1}{N}\sum_{i=1}^N G_L(\lambda, W_i) - \zeta\sum_{k \in K_0} |\lambda_k| \right], \\
        \tilde \lambda_N^U(\zeta) &= \min_{\lambda \in \mb{R}^K} \left[
    \frac{1}{N}\sum_{i=1}^N G_U(\lambda, W_i) + \zeta\sum_{k \in K_0} |\lambda_k| \right],
    \end{aligned}
    \label{eq.inf.pen.bounds}
\end{equation}
where $\zeta > 0$ is the penalty parameter and $K_0 \subseteq \{1, \ldots, K\}$ is the index set of the penalized moment restrictions. In practice, one may simply take $K_0 = \{1, \ldots, K\}$. 

This modification has the following econometric interpretation. For $\zeta \geq 0$, consider the following relaxation of the moment conditions:
\begin{equation}
    \E(\phi_k(W_i, B_i)) = 0 \text{ for all } k \notin K_0 \qand
    |\E(\phi_k(W_i, B_i))| \leq \zeta \text{ for all } k \in K_0,
    \label{eq.relaxation}
\end{equation}
which reduces to \Cref{ass.gmm} if $\zeta = 0$. This is similar in spirit to the minimally relaxed identified set in \citet{andrews2024}. The following proposition then characterizes the smallest $\zeta \geq 0$ for which the finite sample satisfies the relaxed moment restrictions in (\ref{eq.relaxation}), therefore resolving the overidentification issue.

\begin{proposition}
    \label{prop.minimumDistance}
    Given the sample $(W_1, \ldots, W_N)$, consider the linear program
    \begin{equation}
        \begin{aligned}
            \min_{P \in \mc{M}_{W \times B}, ~P \geq 0, ~\zeta \geq 0} \zeta \st
            & \left|\int \phi_k(W_i, B_i) dP\right| \leq \zeta, \quad k \in K_0, \\
            & \int \phi_k(W_i, B_i) dP = 0, \quad k \notin K_0, \\
            & \int P(w, dB_i) = \hat P_W(w) ~\text{ for all } w \in \mc{W},
        \end{aligned}
        \label{eq.minimumDistance.primal}
    \end{equation}
    where $\hat P_W$ is the empirical distribution of $W_i$ constructed from $(W_1, \ldots, W_N)$. Then, under \Cref{ass.proof.regularity}, its solution equals to the solution to:
    \begin{equation}
        \max_{\lambda \in \mb{R}^K}
        \frac{1}{N}\sum_{i=1}^N \min_{b \in \mc{B}} \left\{ \sum_{k=1}^{K}\lambda_k \phi_k(W_i, b) \right\}
        \st \sum_{k \in K_0} |\lambda_k| \leq 1.
        \label{eq.minimumDistance}
    \end{equation}
\end{proposition}

\begin{proof}
    I can rewrite (\ref{eq.minimumDistance.primal}) as:
    \begin{displaymath}
        \begin{aligned}
            \min_{P \in \mc{M}_{W \times B}, ~P \geq 0, ~\zeta \geq 0} \zeta \st
            & \int \phi_k(W_i, B_i) dP \leq \zeta, \quad k \in K_0, \\
            & \int \phi_k(W_i, B_i) dP \geq -\zeta, \quad k \in K_0, \\
            & \int \phi_k(W_i, B_i) dP = 0, \quad k \notin K_0, \\
            & \int P(w, dB_i) = \hat P_W(w) ~\text{ for all } w \in \mc{W},
        \end{aligned}
    \end{displaymath}
    Then, similarly to the proof of \Cref{prop.gmm}, I can invoke the dualty theorem with inequality constraints \citep{anderson1983} to obtain (\ref{eq.minimumDistance}) as the simplified dual representation of (\ref{eq.minimumDistance.primal}).
\end{proof}

Let $\zeta^*$ be the solution to (\ref{eq.minimumDistance}). Then, for any $\zeta \geq \zeta^*$, it follows that the estimated bounds for $\theta$ under the $\zeta$-relaxed moment restrictions in (\ref{eq.relaxation}) are given by the $L^1$-penalized optimizations. The following proposition shows this result for the lower bound. The result for the upper bound holds similarly.

\begin{proposition}
    \label{prop.penalty}
    Given the sample $(W_1, \ldots, W_N)$ and given $\zeta \in \mb{R}$, consider the linear program that finds the smallest value of $\theta = \E(m(W_i, B_i))$ that satisfies (\ref{eq.relaxation}):
    \begin{equation}
        \begin{aligned}
            \min_{P \in \mc{M}_{W \times B}, ~P \geq 0} \int m(W_i, B_i) dP \st
            & \left|\int \phi_k(W_i, B_i) dP\right| \leq \zeta, \quad k \in K_0, \\
            & \int \phi_k(W_i, B_i) dP = 0, \quad k \notin K_0, \\
            & \int P(w, db) = \hat P_W(w) ~\text{ for all } w \in \mc{W},
        \end{aligned}
        \label{eq.lb.est.pen.primal}
    \end{equation}
    where $\hat P_W$ is the empirical distribution of $W_i$ constructed from $(W_1, \ldots, W_N)$. Then, under \Cref{ass.proof.regularity}, its solution equals to 
    \begin{equation}
        \tilde L = \max_{\lambda \in \mb{R}^K} \left[
        \frac{1}{N}\sum_{i=1}^N G_L(\lambda, W_i) - \zeta\sum_{k \in K_0} |\lambda_k| \right].
        \label{eq.lb.est.pen}
    \end{equation}
\end{proposition}

\begin{proof}
    I can rewrite (\ref{eq.lb.est.pen.primal}) as:
    \begin{displaymath}
        \begin{aligned}
            \min_{P \in \mc{M}_{W \times B}, ~P \geq 0} \int m(W_i, B_i) dP \st
            & \int \phi_k(W_i, B_i) dP \leq \zeta, \quad k \in K_0, \\
            & \int \phi_k(W_i, B_i) dP \geq -\zeta, \quad k \in K_0, \\
            & \int \phi_k(W_i, B_i) dP = 0, \quad k \notin K_0, \\
            & \int P(w, db) = \hat P_W(w) ~\text{ for all } w \in \mc{W}.
        \end{aligned}
    \end{displaymath}
    Then, similarly to the proof of \Cref{prop.gmm}, I can invoke the dualty theorem with inequality constraints \citep{anderson1983} to obtain (\ref{eq.lb.est.pen}) as the simplified dual representation of (\ref{eq.lb.est.pen.primal}).
\end{proof}

\Cref{prop.penalty} implies that the $L^1$-penalized finite sample optimizers $\tilde \lambda_N^L(\zeta)$ and $\tilde \lambda_N^U(\zeta)$ defined in (\ref{eq.inf.pen.bounds}) are precisely the maximizer in (\ref{eq.lb.est.pen}) for the lower bound and and the corresponding minimizer for the upper bound. I then implement the procedure of \citet{andrews2017} with a modification that I restrict the supremum to be calculated over the neighborhoods of $\tilde \lambda_N^L(\zeta)$ and $\tilde \lambda_N^U(\zeta)$ rather than the entire $Q^K$ space of $\lambda$. Note that, if the supremum were taken over all of $Q^K$, both $T_{AS}(\theta)$ and $c_{AS,GMS}^{(b)}(\theta)$ would all diverge. As overidentification is resolved as $N \rightarrow \infty$, the supremum over $Q^K$ becomes finite and the original procedure of \citet{andrews2017} can be applied.

In what follows, I check the performance of this heuristic method by simulation. Consider a data generating process
\begin{displaymath}
    Y_{it} = \alpha_i + \beta_i Y_{i,t-1} + \veps_{it}, \quad t=1, \ldots, T,
\end{displaymath}
with $T=10$, where $\alpha_i \sim Unif[-3,3]$, $\beta_i \sim Unif[0,1]$, and $\veps_{it} \sim N(0,1)$ are independent random variables. I generate the initial value by $Y_{i0} = \beta_i + z_i$, where $z_i \sim N(0,1)$ is another independent random variable. I construct the confidence interval for $\E(\beta_i)$ using the moment conditions in (\ref{eq.orthogonal.closedform.refined}), where I set $S_{it} = \{1, Y_{i,\max\{0, t-5\}}, \ldots, Y_{i,t-1}\}$ and relax the first moment restriction $\E\left(\sum_{t=1}^T (R_{it}'B_i)\veps_{it}\right) = 0$, which means that I take $K_0 = \{1\}$. To circumvent the computational difficulty in calcuating the population identified set of this model, I generate $100,000$ observations from this model and treat them as a finite population, for which the population identified set can be calculated using the bounds in \Cref{prop.mean.closedform.refined}.

Under this setup, I implement the heuristic modification of \citet{andrews2017} where I calculate the supremum over a finite grid of $L$ points in the neighborhoods of $\tilde \lambda_N^L(\zeta)$ and $\tilde \lambda_N^U(\zeta)$. I obtain this grid by adding Gaussian perturbations to $\tilde \lambda_N^L(\zeta)$ and $\tilde \lambda_N^U(\zeta)$, while including these optimizers themselves. I then obtain the critical values via $100$ bootstrap replications, for a nominal coverage rate of $0.9$. The simulated coverage rate is obtained by $1000$ Monte Carlo replications.

\Cref{table.simulation} reports the simulated coverage rates for various sample sizes $N$ and the grid sizes $L$. The results show that, as long as $L$ is sufficiently large, the heuristic confidence interval achieves reasonable coverage rates.

\begin{table}[!htbp]
    \centering %
    \begin{tabular}{c | c c c} %
    \hline\hline %
     & $L=50$ & $L=100$ & $L=200$ \\ %
    \hline\hline %
    $N=500$ & 0.878 & 0.890 & 0.932 \\
    $N=1000$ & 0.905 & 0.918 & 0.944 \\
    \hline %
    \end{tabular}
    \caption{Simulated coverage rates for the heuristic inference procedure, where the supremum is evaluated with a total of $L$ finite number of points in the neighborhoods of $\tilde \lambda_N^L(\zeta)$ and $\tilde \lambda_N^U(\zeta)$. The nominal coverage probability is $0.9$.}
    \label{table.simulation}
\end{table}

\subsection{Estimations with transitory income processes}

\label{sec.appendix.deconvolution}

In this subsection, I examine robustness of the results in \Cref{sec.application} to existence of a transitory income process. In particular, I examine if the confidence intervals for $\E(\rho_i)$ change if there is a transitory income process added to (\ref{eq.application.incomeprocess}). 

To estimate the confidence intervals in the presence of a transitory income process, I consider the model $Y_{it} = \tilde Y_{it} + \veps_{it}$, where $Y_{it}$ is the raw log-earnings data, $\tilde Y_{it}$ is the log-earnings without the transitory income process, and $\veps_{it}$ is an i.i.d. Gaussian transitory income process. Note that the models in the main text do not involve transitory income processes, which corresponds to $\veps_{it} = 0$.

I conduct estimation and inference for $\E(\rho_i)$ in the presence of $\veps_{it}$ in a two-step procedure. First, I use the approach \citet{arellano2021} to numerically recover pseudo-observations of $\tilde Y_{it}$ when $Y_{it}$ is observed and the distribution of $\veps_{it}$ is known. Second, I apply estimation procedures for the RIP-RC, HIP-RC, RIP-RC-J, and HIP-RC-J models to the numerically recovered pseudo-observations of $\tilde Y_{it}$. The estimation results from this procedure are presented in \Cref{table.application.mean.appendix.denoising}. The estimation results are qualitatively similar to those in \Cref{table.application.mean} in the main text. In particular, the upper confidence limits of $\E(\rho_i)$ are significantly less than 1, and the confidence intervals for the RIP and the HIP processes show substantial overlap.

\begin{table}[!htbp]
    \centering %
    \begin{tabular}{c c c c c} %
    \hline\hline %
    Parameter & RIP-RC & HIP-RC & RIP-RC-J & HIP-RC-J \\ %
    \hline\hline %
    $\E(\rho_i)$           & [0.424, 0.597] & [0.267, 0.572] & [0.450, 0.583] & [0.253, 0.529] \\
    \hline %
    \end{tabular}
    \caption{Confidence intervals of $\E(\rho_i)$ for the RIP type and the HIP type processes with heterogeneous coefficients, after obtaining pseudo-observations of $\tilde Y_{it}$ without the transitory income process using the method of \citet{arellano2021}. The nominal coverage probability is $0.95$. These confidence intervals are robust to overidentification and model misspecification.}
    \label{table.application.mean.appendix.denoising}
\end{table}

In what follows, I describe in detail the procedure that I used to numerically recover $\tilde Y_{it}$ from the model $Y_{it} = \tilde Y_{it} + \veps_{it}$. I first describe the method proposed in \citet{arellano2021}. They considered a model
\begin{equation}
    Z = X + \veps
    \label{eq.appendix.abmodel}
\end{equation}
where all variables are scalar\footnote{They also consider a more general case of multivariate factor models.} and $X$ is independent of $\veps$. In this model, $Z$ is observed, but $X$ and $\veps$ are not observed. Instead, the distribution of $\veps$ is known. The objective of \citet{arellano2021} is to obtain pseudo-observations of $X$, given the observations of $Z$ and the knowledge on the distribution of $\veps$.

Let $\pp_Z$, $\pp_X$ and $\pp_\veps$ be the probability distributions of $Z$, $X$ and $\veps$, respectively. Let $\pp_{X+\veps}$ be the distribution of $X+\veps$, which is equal to the convolution of $\pp_X$ and $\pp_\veps$. The second-order Wasserstein distance between $Z$ and $X + \veps$, denoted by $W_2(\pp_Z, \pp_{X+\veps})$, is defined by:
\begin{equation}
    W_2(\pp_Z, \pp_{X+\veps}) = \left( \min_{\pi \in \Pi(\pp_Z, \pp_{X+\veps})} \int || z - \hat z ||^2 d\pi(z, \hat z) \right)^{1/2},
    \label{eq.appendix.wasserstein}
\end{equation}
where $\Pi(\pp_Z, \pp_{X+\veps})$ is the set of couplings of $\pp_Z$ and $\pp_{X+\veps}$, i.e., joint distributions of $Z$ and $X + \veps$ whose marginal distributions are $\pp_Z$ and $\pp_{X+\veps}$. It is known that (\ref{eq.appendix.wasserstein}) is a metric for convergence in distribution among distributions with finite second moments, which means that it satisfies the axioms of distance and that $W_2(\nu_k, \mu) \rightarrow 0$ if and only if $\nu_k \overset{d}{\rightarrow} \mu$. Then, (\ref{eq.appendix.abmodel}) implies that
\begin{displaymath}
    W_2(\pp_Z, \pp_{X+\veps}) = 0.
\end{displaymath}
Based on this result, \citet{arellano2021} obtain pseudo-observations of $X$ by minimizing the sample version of (\ref{eq.appendix.wasserstein}).

I apply their approach to the panel data setting to obtain pseudo-observations of $\tilde Y_{it}$. I assume that $\veps_{it}$ follows an i.i.d. zero-mean Gaussian distribution such that $\var(\veps_{it}) = 0.047$, which is the variance estimate of the transitory income process in \citet{guvenen2009}. I then simulate $K = 200$ i.i.d. draws of $(\veps_{i1}, \ldots, \veps_{iT})$:
\begin{displaymath}
    \veps_{k} = (\veps_{k1}, \ldots, \veps_{kT}) \in \mb{R}^T, \quad k=1, \ldots, K.
\end{displaymath}
Then, given the initial values of $\tilde Y_{it}$, denoted by
\begin{displaymath}
    \tilde Y_i = (\tilde Y_{i1}, \ldots, \tilde Y_{iT}) \in \mb{R}^T, \quad i=1, \ldots, N,
\end{displaymath}
I obtain the \emph{synthetic} data of $Y_{it}$ by calculating: 
\begin{equation}
    \hat Y_{ik} = \tilde Y_i + \veps_k \in \mb{R}^T, \quad i=1, \ldots, N, \quad k=1, \ldots, K,
    \label{eq.appendix.convolution}
\end{equation}
where the synthetic data $\hat Y_{it}$ has size $NK$. Note that (\ref{eq.appendix.convolution}) computes a convolution of $\tilde Y_i$ and $\veps_{it}$, because the distribution of $\hat Y_{ik}$ is equal to the convolution of the empirical distribution of $\tilde Y_i$ and the empirical distribution of $\veps_k$.

I then compare the distribution of $\hat Y_{ik}$ with the observed distribution of $Y_{it}$. Let $\hat P_{\hat Y}$ and $\hat P_{Y}$ be the empirical distributions of $\hat Y_{ik}$ and $Y_i$, respectively. Then the (squared) second-order Wasserstein distance between the synthetic and the observed data is given by:
\begin{displaymath}
    \begin{aligned}
        W_2^2(\hat P_{Y}, \hat P_{\hat Y}) &= \min_{0 \leq p_{ijk} \leq 1} \sum_{i=1}^N \sum_{j=1}^N \sum_{k=1}^K p_{ijk} || Y_i - \hat Y_{jk} ||^2 \\
        &\st \sum_{i=1}^N p_{ijk} = 1, \quad \sum_{j=1}^N\sum_{k=1}^K p_{ijk} = 1,
    \end{aligned}
\end{displaymath}
which is the sample analogue of (\ref{eq.appendix.wasserstein}) and takes the form of a linear program. I then obtain pseudo-observations of $\tilde Y_{it}$, denoted by $\tilde Y_i = (\tilde Y_{i1}, \ldots, \tilde Y_{iT})$ for $i=1, \ldots, N$, by:
\begin{displaymath}
    \{\tilde Y_i\}_{i=1}^N = \argmin_{\tilde Y_1, \ldots, \tilde Y_N} W_2^2(\hat P_{Y}, \hat P_{\hat Y}),
\end{displaymath}
which can be shown to be a convex optimization problem. The minimizer of this problem is then the pseudo-observations $\tilde Y_i$.

\end{appendices}

\end{document}